\documentclass[letter,11pt]{article}
\usepackage[margin=1in]{geometry}
\usepackage[UKenglish]{babel}
\usepackage[noadjust]{cite}
\usepackage{amsmath}
\usepackage{amssymb}
\usepackage{amsthm} 
\usepackage{bm}
\usepackage{bbm}
\usepackage{mathtools}
\usepackage{mathrsfs}
\usepackage{csquotes}
\usepackage{todonotes}
\usepackage{tikz}
\usepackage{tikz-3dplot}
\usetikzlibrary{perspective}
\usepackage{graphicx,wrapfig}
\usepackage[title]{appendix}

\usepackage{hyperref} 
\hypersetup{colorlinks=true,citecolor=blue,linkcolor=blue,urlcolor=blue}

\newcommand{\ignore}[1]{}
\newcommand{\A}{\mathbf{A}}
\newcommand{\B}{\mathbf{B}}

\newcommand{\K}{\mathbf{K}}
\newcommand{\X}{\mathbf{X}}

\newcommand{\N}{\mathbb{N}}
\newcommand{\R}{\mathbb{R}}
\newcommand{\Q}{\mathbb{Q}}

\newcommand{\cT}{\mathcal{T}}
\newcommand{\bF}{\mathbb{F}}
\newcommand{\sM}{\mathscr{M}}

\renewcommand{\vec}[1]{\mathbf{#1}}
\newcommand{\ba}{\vec{a}}
\newcommand{\bb}{\vec{b}}
\newcommand{\bc}{\vec{c}}
\newcommand{\bd}{\vec{d}}
\newcommand{\bi}{\vec{i}}
\newcommand{\bj}{\vec{j}}

\newcommand{\bu}{\vec{u}}
\newcommand{\bx}{\vec{x}}
\newcommand{\bv}{\vec{v}}
\newcommand{\by}{\vec{y}}

\newcommand{\be}{\vec{e}}
\newcommand{\bell}{\ensuremath{\boldsymbol\ell}}

\DeclareMathOperator{\BLP}{BLP}

\DeclareMathOperator{\Pol}{Pol}
\DeclareMathOperator{\PCSP}{PCSP}
\DeclareMathOperator{\CSP}{CSP}

\DeclareMathOperator{\id}{id}
\DeclareMathOperator{\Imago}{\operatorname{Im}}
\DeclareMathOperator{\Hom}{\operatorname{Hom}}
\DeclareMathOperator{\Aut}{\operatorname{Aut}}
\DeclareMathOperator{\supp}{supp}
\DeclareMathOperator{\SA}{SA}
\DeclareMathOperator{\dom}{dom}
\DeclareMathOperator{\ar}{ar}

\DeclarePairedDelimiter{\floor}{\lfloor}{\rfloor}

\newcommand{\Qconv}{\ensuremath{\mathscr{Q}_{\operatorname{conv}}}}
\newcommand{\bone}{\mathbf{1}}  
 
\newcommand{\tensor}[1]{\textsuperscript{\raisebox{-.5pt}{\normalfont\textcircled{\raisebox{-.1pt}{\tiny #1}}}}}
\newcommand\cont[1]{\mathrel{\overset{\makebox[0pt]{\mbox{\normalfont\tiny\sffamily #1}}}{\ast }}}

\theoremstyle{plain}
\newtheorem{thm}{Theorem}
\newtheorem*{thm*}{Theorem}
\newtheorem{lem}[thm]{Lemma}
\newtheorem*{lem*}{Lemma}
\newtheorem{prop}[thm]{Proposition}
\newtheorem*{prop*}{Proposition}

\newtheorem*{cor*}{Corollary}

\theoremstyle{definition}
\newtheorem{defn}[thm]{Definition}
\newtheorem{rem}[thm]{Remark}
\newtheorem{example}[thm]{Example}

\pgfmathsetmacro{\opacityDefault}{.05}
\pgfmathsetmacro{\bigShift}{4.5}
\newcommand\cube[5]{
\draw[fill=blue,opacity=#2,shift={(-#3,#5,-#4)},scale=.7](0,0,0)--++(0,-1,0)--++(0,0,-1)--++(0,1,0)--cycle;
\draw[fill=blue,opacity=#2,shift={(-#3,#5,-#4)},scale=.7](0,0,0)--++(-1,0,0)--++(0,0,-1)--++(1,0,0)--cycle;
\draw[fill=blue,opacity=#2,shift={(-#3,#5,-#4)},scale=.7](0,0,0)--++(-1,0,0)--++(0,-1,0)--++(1,0,0)--cycle;
}

\begin{document}

\title{The Sherali-Adams Hierarchy\\ for Promise CSPs through Tensors\thanks{The research leading to these results has received funding from the European Research Council (ERC) under the European Union's Horizon 2020 research and innovation programme (grant agreement No 714532). The paper reflects only the authors' views and not the views of the ERC or the European Commission. The European Union is not liable for any use that may be made of the information contained therein.}}

\author{Lorenzo Ciardo\\
University of Oxford, UK\\
\texttt{lorenzo.ciardo@cs.ox.ac.uk} 
\and 
Stanislav {\v{Z}}ivn{\'y}\\
University of Oxford, UK\\
\texttt{standa.zivny@cs.ox.ac.uk}
}

\date{}
\maketitle

\begin{abstract}
\noindent
We study the Sherali-Adams linear programming hierarchy in the context of promise constraint satisfaction problems ($\PCSP$s). We characterise when a level of the hierarchy accepts an instance in terms of a homomorphism problem for an appropriate multilinear structure obtained through a tensor power of the constraint language. The geometry of this structure, which consists in a space of tensors satisfying certain symmetries, allows then to establish non-solvability of the approximate graph colouring problem via constantly many levels of Sherali-Adams.
 
Besides this primary application, our tensorisation construction introduces a new tool to the study of hierarchies of algorithmic relaxations for computational problems within (and, possibly, beyond) the context of constraint satisfaction. In particular, we see it as a key step towards the algebraic characterisation of the power of Sherali-Adams for $\PCSP$s.
\end{abstract}

\setcounter{page}{0}\thispagestyle{empty}\clearpage

\section{Introduction}
\label{sec:intro}

What are the limits of efficient algorithms and where is the precise borderline of tractability? The \emph{constraint satisfaction problem} ($\CSP$) offers a general framework for studying such fundamental questions for a large class of computational problems~\cite{Creignouetal:siam01,Creignou08:complexity,kz17:dagstuhl} but yet for a class that is amenable to identifying the mathematical structure governing tractability. Canonical examples of $\CSP$s are satisfiability or ``not-all-equal'' satisfiability of 3-CNF formulas  (called 3-SAT and 3-NAE-SAT, respectively), 
linear equations, various variants of (hyper)graph colourings, and the graph clique problem. All $\CSP$s can be seen as homomorphism problems between relational structures~\cite{Feder98:monotone}: Given two relational structures $\X$ and $\A$, is there a homomorphism from $\X$ to $\A$?
Intuitively, the  structure  $\X$ represents the variables of the $\CSP$ instance and their interactions, whereas the  structure $\A$  represents the constraint language; i.e., the alphabet and the allowed constraint relations.

Among the possible restrictions that can be imposed on a class of $\CSP$s,
one of the most studied types of such problems are so-called \emph{finite-domain} \emph{non-uniform} $\CSP$s~\cite{Jeavons97:closure,Feder98:monotone,Kolaitis00:jcss,BKW17}, in which both structures $\X$ and $\A$ have a finite universe, and the target structure $\A$ is fixed whereas the source structure $\X$ is given on input. From the examples above, 3-SAT, 3-NAE-SAT, linear equations over finite fields, graph colourings with constantly many colours are all examples of finite-domain non-uniform $\CSP$s. For instance, in the graph $c$-colouring problem the target structure $\A$ is a $c$-clique and the structure $\X$ is the input graph. The existence of a homomorphism from a graph to a $c$-clique is equivalent to the existence of a colouring of the graph with $c$ colours.

Non-examples of finite-domain non-uniform $\CSP$s are the graph clique problem and linear equations over the rationals. 
The former is 
a $\CSP$ with 
a fixed class of source structures~\cite{Grohe07:jacm,Marx13:jacm} (but an arbitrary target structure). The latter is 
an infinite-domain $\CSP$~\cite{Bodirsky19:sicomp,Barto20:sicomp,Bodirsky21:sicomp}.

We will be concerned with polynomial-time tractability of $\CSP$s. Studied research directions include investigating questions such as: Is there a
solution~\cite{Bulatov17:focs,Zhuk20:jacm}? How many solutions are there,
exactly~\cite{Creignou96:ic,Bulatov13:jacm-dichotomy,Dyer13:sicomp} or approximately~\cite{Bulatov13:jacm,Chen15:jcss}?
What is the maximum number of simultaneously satisfied constraints,
exactly~\cite{Creignou95:jcss,Huber14:sicomp,tz16:jacm-complexity} or approximately~\cite{Deineko08:jacm,Austrin10:sicomp,Raghavendra08:everycsp}? What is
the minimum number of simultaneously unsatisfied
constraints~\cite{Khanna00:approximability,Dalmau18:jcss}? Given an almost
satisfiable instance, can one find a somewhat satisfying
solution~\cite{Dalmau13:toct,Barto16:sicomp,Dalmau19:sicomp}? 
In this paper, we will focus on the following question:

\medskip

\emph{Given a satisfiable instance, can one find a solution that is satisfying in a weaker sense}~\cite{AGH17,BBKO21,BG21:sicomp}? 

\medskip

\noindent This was formalised as \emph{promise constraint satisfaction problems} ($\PCSP$s) by Austrin, Guruswami and H{\aa}stad~\cite{AGH17} and Brakensiek and Guruswami~\cite{BG21:sicomp}. Let $\A$ and $\B$ be two fixed relational structures such that there is a homomorphism from $\A$ to $\B$. 
Intuitively, the structure $\A$ represents the allowed ``strict'' constraints and the structure $\B$ represents the corresponding ``weak'' constraints.
An instance of the $\PCSP$ with the template $(\A,\B)$, denoted by $\PCSP(\A,\B)$, is a relational structure $\X$ such that there is a homomorphism from $\X$ to $\A$. The task is to find a homomorphism from $\X$ to $\B$, which exists by the composition of the two promised homomorphisms. 

$\PCSP$s are a vast generalisation of $\CSP$s including problems that cannot be
expressed as $\CSP$s.
The work of Barto, Bul\'in, Krokhin, and Opr\v{s}al~\cite{BBKO21} lifted
and greatly extended the algebraic framework developed for
$\CSP$s~\cite{Jeavons97:closure,Bulatov05:classifying,BOP18} to the realm of $\PCSP$s.
Subsequently, there has been a series of recent works on the computational
complexity of
$\PCSP$s building on~\cite{BBKO21}, including applicability of convex relaxations~\cite{BG19,BG20,bgwz20,Butti21:mfcs,Ciardo22:soda,Atserias22:soda} and complexity of fragments of $\PCSP$s~\cite{GS20:icalp,AB21,Barto21:stacs,BWZ21,BG21:talg,Barto22:soda}.
Other strong results on $\PCSP$s have also been established via
other techniques than those in~\cite{BBKO21},
e.g., hardness of various (hyper)graph colourings~\cite{Khot01,DRS05,Huang13,ABP20} and  other $\PCSP$s~\cite{Bhangale21:stoc,Braverman21:itcs,BGS21,Bhangale22:stoc}.

An example of a $\PCSP$, identified in~\cite{AGH17}, is finding a satisfying assignment to a $k$-CNF formula given that a $g$-satisfying assignment exists; i.e., an assignment that satisfies at least $g$ literals in each clause. Austrin et al. established that this problem is NP-hard if $g/k<1/2$ and solvable via a constant level of the Sherali-Adams linear programming relaxation otherwise~\cite{AGH17}. This classification was later extended to problems over arbitrary finite domains by Brandts et al.~\cite{BWZ21}.

Another example of a $\PCSP$, identified in~\cite{BG21:sicomp}, is finding a ``not-all-equal'' assignment to a monotone $3$-CNF formula given that a ``$1$-in-$3$'' assignment is promised to exist; i.e., given a $3$-CNF formula with positive literals only and the promise that an assignment exists that satisfies exactly one literal in each clause, the task is to find an assignment that satisfies one or two literals in each clause.
This problem is solvable in polynomial time via a constant level of the Sherali-Adams linear programming relaxation~\cite{BG21:sicomp} but not via a reduction to finite-domain $\CSP$s~\cite{BBKO21}.

A third example of a $\PCSP$ is the well-known \emph{approximate graph colouring} problem: Given a $c$-colourable graph, find a $d$-colouring of it, for $c\leq d$. In the decision version, the problem asks 
to distinguish graphs that are $c$-colourable from graphs that are not even $d$-colourable. This corresponds to $\PCSP(\K_c,\K_d)$, where $\K_p$ is the clique on $p$ vertices. Contrary to the two examples above -- and despite a long history dating back to 1976~\cite{GJ76} -- the complexity of this problem is still unknown in general; it is widely believed that it is NP-hard for any constant values of $c$ and $d$ with $3\leq c\leq d$. For $c=d$, it becomes the classic $c$-colouring problem, which appeared on Karp's original list of $21$ NP-complete problems~\cite{Karp72}. The case $c=3$, $d=4$ was only proved to be NP-hard in 2000 by Khanna, Linial, and Safra~\cite{KhannaLS00} (cf. also~\cite{GK04});
 more generally, they showed hardness of the case $d=c+2\floor{c/3}-1$. This was improved to $d=2c-2$ in 2016~\cite{BrakensiekG16}, and recently to $d=2c-1$ in~\cite{BBKO21}. In particular, this last result implies hardness of the case $c=3$, $d=5$; the complexity of the case $c=3$, $d=6$ is still open. Building on the work of Khot~\cite{Khot01} and Huang~\cite{Huang13}, Wrochna and \v{Z}ivn\'{y} established NP-hardness of $d={c\choose\floor{c/2}}-1$ for $c\geq 4$~\cite{WZ20} (cf. also~\cite{KrokhinOWZ20}).
 NP-hardness of approximate graph colouring was established for all constant $3\leq c\leq d$ by Dinur, Mossel, and Regev~\cite{Dinur09:sicomp} under a non-standard variant of the Unique Games Conjecture, and by Guruswami and Sandeep~\cite{GS20:icalp} under the $d$-to-1 conjecture for any fixed $d$.

\medskip
A key concept in the study of the computational complexity of, and the power of algorithms for, $\CSP$s is that of polymorphisms~\cite{BKW17}. Intuitively, one can see polymorphisms of a $\CSP$ as higher order symmetries of the solution spaces~\cite{BKW17};
e.g., given a, say, $3$-ary polymorphism of a $\CSP$ and any three solutions, the polymorphism combines the three solutions to produce another one. Crucial properties of polymorphisms are captured via identities, which can express features such as being symmetric (order invariant). Algebraically, polymorphisms form minions: sets of functions closed under identifying variables, permuting variables, and introducing dummy variables~\cite{BBKO21}. More abstractly, minions are functors from the category of finite sets to the category of finite sets.

The power of convex relaxations has been instrumental in the understanding of
the computational complexity of various variants of $\CSP$s, including constant
approximability of Min-$\CSP$s~\cite{Ene13:soda,Dalmau18:jcss} and
Max-$\CSP$s~\cite{Khot07:sicomp,Raghavendra08:everycsp}, robust satisfiability of
$\CSP$s~\cite{Zwick98:stoc,Kun12:itcs,Barto16:sicomp}, and exact solvability of
optimisation $\CSP$s~\cite{ktz15:sicomp,tz17:sicomp}. 
In the context of $\PCSP$s,
the more general view of minions mentioned above has been useful in establishing characterisations of convex relaxations~\cite{BBKO21,bgwz20,Ciardo22:soda}. In particular, 
Barto et al.~\cite{BBKO21} characterised (in terms of a minion and polymorphism
identities) the power of the basic linear programming (BLP) relaxation and the
power of the affine integer programming (AIP) relaxation. Furthermore,
Brakensiek et al.~\cite{bgwz20} characterised (again, in terms of a minion and
polymorphism identities) the power of the combined BLP+AIP relaxation. Essentially, the power of BLP is captured by symmetric polymorphisms of all arities, whereas the power of BLP+AIP is captured by polymorphisms of all odd arities which have the property that they are symmetric on odd and even coordinates. In recent
work, Ciardo and \v{Z}ivn\'{y} proposed a relaxation, called CLAP, that is stronger than
BLP+AIP and gave a minion-based characterisation of its power as well as a
sufficient condition in terms of polymorphism
identities~\cite{Ciardo22:soda}.\footnote{CLAP is an LP relaxation
sitting between the first and second levels of Sherali-Adams and augmented with
AIP.}

\paragraph{Contributions}
We study applicability of the Sherali-Adams linear programming hierarchy~\cite{Sherali1990} to $\PCSP$s and give three main contributions.

\medskip

\noindent\textbf{(1) Tensorisation}
We propose a new approach to the study of the Sherali-Adams hierarchy inspired by multilinear algebra. We interpret Sherali-Adams acceptance as a homomorphism problem involving a tensorised version of the original structures. Essentially, the problem of distinguishing the cases when the $k$-th level of the hierarchy accepts or rejects a given instance is cast as the problem of checking the existence of a homomorphism between the $k$-th tensor power of the instance and a space of tensors built from the $k$-th tensor power of the constraint language and the specific minion $\Qconv$ capturing the power of $\BLP$. Equivalently, the $k$-th level of Sherali-Adams is interpreted as $\BLP$ applied to the $k$-th tensor powers of the instance and the constraint language. This allows us to describe the functioning of the algorithms in the hierarchy by describing the geometry of a space of tensors -- which can be accomplished by using multilinear algebra.
As far as we know, this approach has not appeared in the literature on Sherali-Adams (and related algorithmic techniques such as local consistency~\cite{Barto14:jacm} and the high-dimensional Weisfeiler-Leman algorithm~\cite{AtseriasM13,Butti21:mfcs}).

The scope of this idea is not limited to the Sherali-Adams hierarchy and to the $\PCSP$ framework. By replacing the minion $\Qconv$ with different minions, one can use the tensorisation construction to analogously characterise acceptance for other hierarchies of relaxation algorithms. For example, one can cast acceptance for the $k$-th level of local consistency (which is less powerful than Sherali-Adams) as the same homomorphism problem as above, the only difference being that, in this case, the ``coarser'' minion $\mathscr{H}$ capturing arc consistency~\cite{BBKO21} takes the role of $\Qconv$; equivalently, the $k$-th level of local consistency can be interpreted as arc consistency applied to tensorised structures.
In addition, the multilinear pattern that we found at the core of Sherali-Adams appears to be transversal to the $\PCSP$ framework and, instead, inherently connected to the algorithmic technique itself, which can be applied to classes of computational problems not expressible as
$\PCSP$s.

In summary, while in this work we focus on Sherali-Adams for $\PCSP$s, the tensorisation construction seems to be more widely applicable for getting insight into different hierarchies of algorithmic relaxations, both within the realms of ($\operatorname{P}$)$\CSP$s and beyond.  

\medskip\noindent\textbf{(2) Non-solvability of approximate graph colouring}
Building on our characterisation and the tensorisation machinery, we 
obtain our second main contribution:
non-solvability of the approximate graph colouring problem via the Sherali-Adams linear programming hierarchy. 
The key step is establishing that the $k$-th level of Sherali-Adams applied to the $k$-clique $\K_k$ accepts for any input digraph $\X$. Intuitively, this result relies on the fact that the symmetries satisfied by the space of tensors corresponding to the $k$-th tensor power of $\K_k$ are not strong enough to prevent the existence of a homomorphism from any input. Considering the line digraph construction~\cite{HarnerE72} and observing that it preserves not only polynomial-time solvability~\cite{WZ20,KrokhinOWZ20}
but also acceptance by Sherali-Adams, we then use the previous result to establish that constant levels of Sherali-Adams do not solve the approximate graph colouring problem for any constant number of colours. 

\medskip
\noindent\textbf{(3) Towards the power of Sherali-Adams}
The first step in all minion characterisations of the power of
BLP~\cite{BBKO21},
AIP~\cite{BBKO21}, BLP+AIP~\cite{bgwz20}, and CLAP~\cite{Ciardo22:soda}
 was to understand the problem of whether an input relational structure $\X$ is \emph{accepted} by the relaxation as a homomorphism problem involving the so-called free structure of some minion $\mathscr{M}$; then, minion-theoretic results allow to characterise those templates that are \emph{solved} by the relaxation by checking whether $\mathscr{M}$ is homomorphic to the minion of polymorphisms of the template. 
We believe that our acceptance characterisation in~\textbf{(1)} is the first step towards an algebraic (possibly, minion-theoretic) description of the power of Sherali-Adams.  
Preliminary results in this direction are discussed in Section~\ref{sec:solv}, where we show
that a natural condition based on tensorisation, while sufficient, is not necessary for a $\PCSP$ to be solved by Sherali-Adams. 

\paragraph{Organisation}
$\PCSP$s, the Sherali-Adams hierarchy, and basic notions on tensors are formally introduced in Section~\ref{sec:prelims}. Section~\ref{sec_extended_abstract} is a brief summary of our main contributions and the techniques used to achieve them. The rest of the paper gives full details: In Section~\ref{sec_tensorisation}, we explore the properties of our tensorisation construction and prove various technical results about it; the Sherali-Adams acceptance characterisation is proved in Section~\ref{sec_characterisation_acceptance}; non-solvability of the approximate graph colouring problem via Sherali-Adams is established in Section~\ref{sec:approx}; in Section~\ref{sec:solv}, we discuss acceptance vs. solvability for Sherali-Adams.

\paragraph{Related work}
Linear programming (LP) is one of the most powerful algorithmic tools known for
designing efficient
relaxations~\cite{vazirani2001approximation,williamson2011design}. If P$\neq$NP,
we do not expect polynomial-sized LPs to compute optimal solutions or even
arbitrarily good approximations of optimal solutions to hard problems. Thus a well established line
of work has sought to prove lower bounds on the efficacy of small LPs. This was
pioneered by Arora, Bollob\'{a}s, Lov\'{a}sz, and Tourlakis~\cite{Arora06:toc},
who explored the limitations of LPs arising from lift-and-project hierarchies such as that of Sherali and Adams~\cite{Sherali1990}. The ultimate goal of this line of work is to prove unconditional lower bounds for small LPs and recent highlights include, e.g.,~\cite{Braun15:stoc,Chan16:jacm-lp,Kothari17:stoc,Ghos18:toc}. We note that an impressive line of work also established lower bounds against (more powerful) semidefinite programming relaxations, e.g.,~\cite{Tulsiani09:stoc,Lee15:stoc,Chan15:jacm}.

In the context of (exact solvability of) $\CSP$s, \cite{tz17:sicomp} characterised the power of Sherali-Adams for valued $\CSP$s, which implies a characterisation for $\CSP$s: The $k$-th level, for $k\geq 3$, solves a $\CSP$ if and only if the third level solves it, and this coincides with the condition that characterises the power of the local consistency algorithm~\cite{Barto14:jacm,Bulatov09:width},
where the collapse to the third level was  shown~\cite{Barto14:jloc}.
Butti and Dalmau~\cite{Butti21:mfcs} recently characterised for $\CSP$s when the $k$-th level of the Sherali-Adams linear programming hierarchy accepts in terms of a construction different from the one introduced in this work. Unlike the tensorisation, the construction considered in~\cite{Butti21:mfcs} yields a relational structure whose domain includes the set of constraints of the original structure.  

The line digraph construction used in our proofs also appeared in~\cite{WZ20} and~\cite{GS20:icalp} to establish results on NP-hardness of approximate graph colouring.

Recent work of Atserias and Dalmau~\cite{Atserias22:soda} established that constant (and in fact sublinear) levels of the local consistency algorithm do not solve the approximate graph colouring problem. While it is well known that the $k$-th level of the Sherali-Adams hierarchy is at least as strong as the (combinatorial) $k$-consistency algorithm, it is now also known that for some $\PCSP$s Sherali-Adams is more powerful: Indeed, Atserias and Dalmau showed that the simple $\PCSP$ ``$1$-in-$3$ vs. not-all-equal'' described above is not solvable by local consistency~\cite{Atserias22:soda} although 
it is solvable by a constant level of the Sherali-Adams hierarchy~\cite{BG21:sicomp}.
Our result on non-solvability of approximate graph colouring extends the result from~\cite{Atserias22:soda} established for the local consistency algorithm. The proof technique followed in~\cite{Atserias22:soda} is probabilistic, as it relies on variants of the sparse incomparability lemma~\cite{Nesetril89,Nesetril04:sparse}, which is a generalisation of Erd\"os' classic result on the existence of graphs of large chromatic number and large girth~\cite{erdos1959graph}. Instead, we prove non-solvability of approximate graph colouring via Sherali-Adams by adopting a multilinear framework, which consists in describing the geometry of specific tensor spaces associated with the problem. This results in a new, non-probabilistic approach to the study of the complexity of approximate graph colouring.

\section{Preliminaries}
\label{sec:prelims}

We denote by $\N$ the set of positive integers. For $k\in\N$, we denote by $[k]$ the set $\{1,\ldots,k\}$. Given a set $X$, a tuple $\bx=(x_1,\dots,x_k)\in X^k$, and a tuple $\bi=(i_1,\dots,i_\ell)\in [k]^\ell$, $\bx_\bi$ shall denote the \emph{projection}
of $\bx$ on $\bi$, i.e., the tuple in $X^\ell$ given by $\bx_\bi\coloneqq (x_{i_1},\dots,x_{i_\ell})$. Given two tuples $\bx=(x_1,\dots,x_k)\in X^k$ and $\by=(y_1,\dots,y_\ell)\in X^\ell$, their \emph{concatenation} is the tuple $(\bx,\by)\coloneqq (x_1,\dots,x_k,y_1,\dots,y_\ell)\in X^{k+\ell}$. We also define $\{\bx\}\coloneqq \{x_1,\dots,x_k\}$.
We indicate by
$\be_i$ the $i$-th standard unit vector of the appropriate size (which will
be clear from the context); i.e., the $i$-th entry of $\be_i$ is $1$, and all other entries are $0$. We denote by $\bone_p$ the all-one vector of size $p$.

\paragraph{Promise CSPs}
A (relational) signature $\sigma$ is a finite set of relation symbols $R$, each with arity 
$\ar(R)\in\N$.
A (relational) $\sigma$-structure $\A$ consists of a domain (universe) $A$ and, for each $R\in\sigma$, a relation $R^{\A}\subseteq A^{\ar(R)}$. A $\sigma$-structure $\A$ is finite if the size $|A|$ of its domain $A$
is finite. In this case, we shall often assume that the domain of $\A$ is $A=[n]$.
An undirected graph is seen as a relational structure having a unique binary relation containing both directions of each edge.

Let $\A$ and $\B$ be $\sigma$-structures. A \emph{homomorphism} from $\A$ to $\B$ is a map $h:A\to B$ such that, for each $R\in\sigma$ with $r=\ar(R)$ and for each $\ba=(a_1,\ldots,a_r)\in A^r$, if $\ba\in R^{\A}$ then $h(\ba)=(h(a_1),\ldots,h(a_r))\in R^{\B}$. We denote the existence of a homomorphism from $\A$ to $\B$ by $\A\to\B$.

A pair of  $\sigma$-structures $(\A,\B)$ with $\A\to\B$ is called a \emph{promise constraint satisfaction problem} ($\PCSP$) \emph{template}. The $\PCSP$ problem parameterised by the template $(\A,\B)$, denoted by $\PCSP(\A,\B)$, is the following computational problem. In the \emph{search} version, the input is a $\sigma$-structure $\X$ with $\X\to\A$ and the goal is to find a homomorphism from $\X$ to $\B$ (which necessarily exists by the assumptions, since composing two homomorphisms results in a new homomorphism). In the  \emph{decision} version, the input is a $\sigma$-structure $\X$ and the goal is to answer \textsc{Yes} if $\X\to\A$ and \textsc{No} if $\X\not\to\B$. The promise is that it is not the case that $\X\not\to\A$ and $\X\to\B$. 
It is known that the decision version reduces to the search version~\cite{BBKO21}, but the converse is not known to hold in general. In this paper, we will focus on the decision version.

We write $\CSP(\A)$ for $\PCSP(\A,\A)$, the classic (non-promise) constraint satisfaction problem.

\paragraph{Algebraic approach to PCSPs}
The algebraic theory of $\PCSP$s developed in~\cite{BBKO21} relies on the notions of polymorphism and minion. 

Let $\A$ be a $\sigma$-structure. For $L\in\N$, the \emph{$L$-th power} of $\A$ is the $\sigma$-structure $\A^L$ with domain $A^L$ whose relations are defined as follows: Given $R\in\sigma$ and an $L\times \ar(R)$ matrix $M$ such that all rows of $M$ are tuples in $R^\A$, the columns of $M$ form a tuple in $R^{\A^L}$. An $L$-ary \emph{polymorphism} of a $\PCSP$ template $(\A,\B)$ is a homomorphism from $\A^L$ to $\B$.
Minions were defined in~\cite{BBKO21} as sets of functions with certain properties. We shall use here the abstract definition of minions, as first done in~\cite{bgwz20}, cf. also~\cite{Ciardo22:soda}.
A \emph{minion} $\mathscr{M}$ consists in the disjoint union of sets $\mathscr{M}^{(L)}$ for $L\in \N$ equipped with (so-called \emph{minor}) operations $(\cdot)_{/\pi}:\mathscr{M}^{(L)}\rightarrow\mathscr{M}^{(L')}$ for all functions $\pi:[L]\rightarrow [L']$, which satisfy
$M_{/\operatorname{id}}=M$ and, for
 $\pi:[L]\rightarrow [L']$ and $\tilde{\pi}:[L']\rightarrow [L'']$,
 $(M_{/\pi})_{/\tilde{\pi}}=M_{/\tilde{\pi}\circ \pi}$
for all $M\in\mathscr{M}^{(L)}$. 

\begin{example}
The set $\Pol(\A,\B)$ of all polymorphisms of a $\PCSP$ template $(\A,\B)$ is a minion with the minor operations defined by 
$f_{/\pi}(a_1,\dots,a_{L'})\coloneqq f(a_{\pi(1)},\dots,a_{\pi(L)})$
for $f:\A^L\to\B$ and $\pi:[L]\to [L']$.
In this minion, the minor operations correspond to identifying coordinates, permuting coordinates, and introducing dummy coordinates (of polymorphisms).
\end{example}
\begin{example}
Another example of minion that shall appear frequently is $\Qconv$. Its $L$-ary elements are rational vectors of size $L$ that are stochastic (i.e., whose entries are nonnegative and sum up to $1$), while the minor
operations are defined as follows: If $q\in\Qconv^{(L)}$
and $\pi:[L]\to[L']$, then $q_{/\pi}\coloneqq Pq$, where $P$ is the $L'\times
L$ matrix whose $(i,j)$-th entry is $1$ if $\pi(j)=i$, and $0$ otherwise.
\end{example}

For two minions $\mathscr{M}$ and $\mathscr{N}$, a \emph{minion homomorphism} $\xi:\mathscr{M}\rightarrow\mathscr{N}$ is a map that preserves arities and minors: Given $M\in\mathscr{M}^{(L)}$ and $\pi:[L]\rightarrow[L']$, $\xi(M)\in \mathscr{N}^{(L)}$ and $\xi(M_{/\pi})=\xi(M)_{/\pi}$. We denote the existence of a minion homomorphism from $\mathscr{M}$ to $\mathscr{N}$ by $\mathscr{M}\to\mathscr{N}$.

We will also need the concept of free structure from~\cite{BBKO21}. 
Let $\mathscr{M}$ be a minion and let $\A$ be a (finite) $\sigma$-structure. The \emph{free structure} of $\sM$ generated by $\A$ is a $\sigma$-structure $\bF_{\sM}(\A)$ with domain $\sM^{(|A|)}$ (potentially infinite). Given a relation symbol $R\in\sigma$ of arity $r$, a tuple $(M_1,\dots,M_r)$ of elements of $\sM^{(|A|)}$ belongs to $R^{\bF_{\sM}(\A)}$ if and only if there is some $q\in \sM^{(|R^\A|)}$ such that $M_i=q_{/\pi_i}$ for each $i\in[r]$, where $\pi_i:R^\A\to A$ maps $\ba\in R^\A$ to its $i$-th coordinate $a_i$.
\begin{example}
\label{ex_free_structure}
The free structure $\bF_{\Qconv}(\A)$ has domain consisting of all rational stochastic vectors of size $|A|$. Given $R\in\sigma$ of arity $r$, a tuple $(v_1,\dots,v_r)$ of $r$ rational stochastic vectors of size $|A|$ is an element of $R^{\bF_{\Qconv}(\A)}$ if and only if there exists some rational stochastic vector $q$ of size $|R^\A|$ for which $v_i=P_iq$ for each $i\in [r]$, where $P_i$ is the $|A|\times|R^\A|$ matrix such that, for $a\in A$ and $\ba=(a_1,\dots,a_r)\in R^\A$, the $(a,\ba)$-th entry of $P_i$ is $1$ if $a_i=a$, and $0$ otherwise.
\end{example}
\paragraph{Sherali-Adams LP hierarchy}

The Sherali-Adams linear programming hierarchy consists in a refinement of the well-known $\BLP$ relaxation (cf.~Appendix~\ref{app:BLP}). Essentially, the $k$-th level of the hierarchy is obtained by enforcing consistency of probability distributions over partial assignments on up to $k$ variables in the input structure. We shall follow the definition as presented in~\cite{Butti21:mfcs}.\footnote{We note that \cite{Atserias22:soda} has a slightly different definition of the Sherali-Adams hierarchy, 
 but the two interleave, cf.~\cite[Appendix~A]{Butti21:mfcs}. In particular, the class of $\PCSP$s solved by constant levels of the hierarchy is the same for both definitions.} 

Given two $\sigma$-structures $\X,\A$, we introduce a variable $\lambda_V(f)$ for every subset $V\subseteq X$ with $1\leq |V|\leq k$ and every function $f:V\to A$, and a variable $\lambda_{R,\bx}(f)$ for every $R\in\sigma$, every $\bx\in R^\X$, and every $f:\{\bx\}\to A$. The $k$-th level of Sherali-Adams is given by the following constraints:
\[
\begin{array}{lll}
\mbox{(SA1)} & \displaystyle\sum_{f:V\to A}\lambda_V(f)=1 & V\subseteq X \mbox{ such that }1\leq|V|\leq k\\
\mbox{(SA2)} & \displaystyle \lambda_U(f)=\sum_{g:V\to A,\;g|_{U}=f}\lambda_V(g) & 
U\subseteq V\subseteq X \mbox{ such that } 1\leq|V|\leq k, f:U\to A 
\\
\mbox{(SA3)} & \displaystyle \lambda_U(f)=\sum_{g:\{\bx\}\to A,\;g|_{U}=f}\lambda_{R,\bx}(g)
&
R\in\sigma, \bx\in R^\X, U\subseteq \{\bx\}
\mbox{ such that }1\leq|U|\leq k, f:U\to A 
\\
\mbox{(SA4)} & \displaystyle \lambda_{R,\bx}(f)=0 & 
R\in\sigma, \bx\in R^\X, f:\{\bx\}\to A
\mbox{ such that } f(\bx)\not\in R^\A.
\end{array}
\]
We say that $\SA^k(\X,\A)$ \emph{accepts} if the system above admits a solution such that all variables take real (equivalently, rational) values in the interval $[0,1]$. This can be checked in polynomial time in the size of the input $\X$ (as it corresponds to checking feasibility of a polynomial-sized LP)~\cite{schrijver1998theory}.
We say that $\SA^k$ \emph{solves} a $\PCSP$ template $(\A,\B)$ if,
for every instance $\X$ of $\PCSP(\A,\B)$, we have (i) if $\X\to\A$ then $\SA^k(\X,\A)$ accepts, and (ii) if $\SA^k(\X,\A)$ accepts then $\X\to\B$. (Note that (i) always holds as, from the definition of $\SA^k$, if $\SA^k(\X,\A)$ does not accept then $\X\not\to\A$.)

\paragraph{Tensors}
Let $S$ be a set. For $n_1,\dots,n_k\in\N$, $\cT^{n_1\times\dots\times n_k}(S)$ denotes the set of all functions from $[n_1]\times\dots\times [n_k]$ to $S$, which we visualise as hypermatrices or tensors. We sometimes denote an element of $\cT^{n_1\times\dots\times n_k}(S)$ by $M=[m_\bi]$, where $\bi\in [n_1]\times\dots\times [n_k]$ and $m_\bi$ is the image of $\bi$ under $M$. 
If $S$ is a ring, the \emph{contraction} of two tensors 
$
M=[m_\bi]\in \cT^{n_1\times\dots\times n_p\times n_{p+1}\times\dots\times n_{p+k}}(S)$,
$\tilde M=[\tilde m_\bi]\in \cT^{n_{p+1}\times\dots\times n_{p+k}\times n_{p+k+1}\times\dots\times n_{p+k+s}}(S)
$,  
denoted by $M\cont{k}\tilde{M}$, is the tensor in $
\cT^{n_1\times\dots\times n_p\times n_{p+k+1}\times\dots\times n_{p+k+s}}(S)
$
such that, for $\bi\in [n_1]\times\dots\times[n_p]$ and $\bj\in [n_{p+k+1}]\times\dots\times[n_{p+k+s}]$, the $(\bi,\bj)$-th entry of $M\cont{k}\tilde{M}$
is given by
$
\sum_{\bell\in [n_{p+1}]\times\dots\times [n_{p+k}]}m_{(\bi,\bell)}\tilde{m}_{(\bell,\bj)}
$.
When $p=0$ or $s=0$, we write $M\ast  \tilde{M}$ for $M\cont{k}\tilde{M}$.
\begin{example}
Given two vectors $\bu,\bv\in\cT^{n}(\mathbb{R})$ and two matrices $M\in \cT^{m\times n}(\R)$, $N\in\cT^{n\times p}(\R)$, we have that
$\bu\cont{1}\bv=\bu\ast \bv=\bu\cdot\bv$, the dot product of $\bu$ and $\bv$;
$M\cont{1}\bu=M\ast\bu=M\bu$, the matrix-vector product of $M$ and $\bu$;
$M\cont{1}N=MN$, the matrix 
product of $M$ and $N$.
\end{example}
Given $\bi\in [n_1]\times\dots\times [n_k]$, we denote by $E_\bi$ the tensor in $\cT^{n_1\times\dots\times n_k}(\R)$ all of whose entries are $0$, except the $\bi$-th entry that is $1$. For $M\in \cT^{n_1\times\dots\times n_k}(\mathbb{R})$, notice that $E_\bi\ast M$ is the $\bi$-th entry of $M$. The \emph{support} of $M$ is the set of indices of all nonzero entries of $M$ -- i.e., the set 
$\supp(M)=\{\bi\in [n_1]\times\dots\times [n_k]:E_\bi\ast M\neq 0\}$.
Most tensors in this paper are cubical, meaning that $n_1=n_2=\dots=n_k$. In this case, we write $\cT^{k;n}(S)$ for $\cT^{n_1\times\dots\times n_k}(S)$, where $n=n_1=\dots=n_k$. For example, $\cT^{1;n}(\R)$ is the set of real $n$-vectors, while $\cT^{2;n}(\R)$ is the set of real $n\times n$ matrices.

\section{Overview of Results and Techniques}
\label{sec_extended_abstract}

What structure lies at the core of the basic linear programming relaxation? Given two $\sigma$-structures $\X$ and $\A$, the $\BLP$ algorithm applied to $\X$ and $\A$ tries to assign probability distributions (i.e., stochastic vectors) over the domain of $\A$ to the elements of $\X$, in such a way that certain marginality requirements are satisfied (see Appendix~\ref{app:BLP}).
Formally, the problem of understanding whether $\BLP(\X,\A)$ accepts was shown in~\cite{BBKO21} to be equivalent to the problem of understanding whether $\X$ is homomorphic to the free structure $\mathbb{F}_{\Qconv}(\A)$. In other words, $\BLP$ acceptance can be cast as $\CSP(\mathbb{F}_{\Qconv}(\A))$. This $\CSP$ enjoys an interesting description: $\mathbb{F}_{\Qconv}(\A)$ consists in an infinite set of stochastic vectors equipped with relations defined through linear identities (cf.~Example~\ref{ex_free_structure}) -- in other words, it is a \emph{linear} object. As a consequence, one can capture the functioning of $\BLP$ by using linear algebra. To give an example of the usefulness of this linear description, observe that the set of stochastic vectors includes certain particularly symmetric members -- the constant vectors $\frac{1}{p}\bone_p$. It follows that any problem that is solved by $\BLP$ must enjoy, in some sense, the same symmetries as those of the constant vectors. This ultimately results in the simple characterisation of the templates solved by $\BLP$, as those templates admitting symmetric polymorphisms of all arities (cf.~Theorem~\ref{thm:blp}).  

The first main goal of this work is to extend the scope of this idea by investigating the structure lying at the core of the Sherali-Adams hierarchy. Rather than stochastic vectors, we will find that the hierarchy can be naturally described by means of stochastic tensors. In other words, we shall argue that the core of Sherali-Adams is essentially a \emph{multilinear} object. The geometric features of this object will then allow us to rule out the Sherali-Adams hierarchy as an algorithmic technique to solve the approximate graph colouring problem. 

We now give a formal overview of these ideas. For $k\in\N$ and a signature $\sigma$, $\sigma^{\tensor{k}}$ is the signature consisting of the same symbols as $\sigma$ such that each symbol $R$ of arity $r$ in $\sigma$ has arity $r^k$ in $\sigma^{\tensor{k}}$. 

\begin{defn}
The \emph{$k$-th tensor power} of a $\sigma$-structure $\A$ is the $\sigma^{\tensor{k}}$-structure $\A^{\tensor{k}}$ having domain $A^k$ and relations defined as follows: For each symbol $R\in\sigma$ of arity $r$ in $\sigma$, we set
$
R^{\A^\tensor{k}}=\left\{\ba^{\tensor{k}}:\ba\in R^\A\right\}
$,
where, for $\ba\in R^\A$, $\ba^\tensor{k}$ is the tensor in $\cT^{k;r}(A^k)$ defined by $E_\bi\ast\ba^\tensor{k}=\ba_\bi$ for any $\bi\in [r]^k$ -- i.e., the $(i_1,i_2,\dots,i_k)$-th element of $\ba^\tensor{k}$ is $(a_{i_1},a_{i_2},\dots,a_{i_k})$.\footnote{\label{footnote_segre}We can visualise $\ba^\tensor{k}$ as the formal Segre outer product of $k$ copies of $\ba$ (cf.~\cite{lim2013tensors}; see also Remark~\ref{remark_rank_one_segre}).}$^{,}$\footnote{\label{footnote_RA_RAk}Notice that $\A^\tensor{1}=\A$. Also, the function $R^\A\to R^{\A^\tensor{k}}$ given by $\ba\mapsto\ba^\tensor{k}$ is a bijection, so $\left|R^{\A^\tensor{k}}\right|=|R^\A|$.}
\end{defn}
\begin{example}
\label{example_tensorisation_1406}
Let us describe the third tensor power of the $3$-clique -- i.e., the structure $\K_3^\tensor{3}$. The domain of $\K_3^\tensor{3}$ is $[3]^3$, i.e., the set of tuples of elements in $[3]$ having length $3$. Let $R$ be the symbol corresponding to the binary edge relation in $\K_3$, so that $R^{\K_3}=\{(1,2),(2,1),(2,3),(3,2),(3,1),(1,3)\}$. Then, $R^{\K_3^\tensor{3}}$ has arity $2^3=8$ and it is a subset of $\mathcal{T}^{3;2}([3]^3)$. 
Specifically, $R^{\K_3^\tensor{3}}=\{(1,2)^\tensor{3},(2,1)^\tensor{3},\linebreak(2,3)^\tensor{3},(3,2)^\tensor{3},(3,1)^\tensor{3},(1,3)^\tensor{3}\}$ where, e.g., $(2,3)^\tensor{3}=\left[\begin{array}{@{}cc|cc@{}}
(2,2,2)&(2,2,3)&(3,2,2)&(3,2,3)\\
(2,3,2)&(2,3,3)&(3,3,2)&(3,3,3)
\end{array}\right]$.\footnote{The vertical line separates the two layers of the tensor.}
\end{example}

The tensorisation construction described above captures algebraically the functioning of the Sherali-Adams hierarchy, as stated in the following main result proved in Section~\ref{sec_characterisation_acceptance}.  We say that a $\sigma$-structure $\A$ is \emph{$k$-enhanced} if $\sigma$ contains a $k$-ary symbol $R_k$ and $R_k^\A=A^k$. Observe that any two $\sigma$-structures $\A,\B$ are homomorphic if and only if the structures $\tilde{\A},\tilde{\B}$ obtained by adding $R_k$ to their signatures are homomorphic. Hence, $\PCSP(\A,\B)$ is equivalent to $\PCSP(\tilde{\A},\tilde{\B})$, and considering $k$-enhanced structures results in no loss of generality.

\begin{thm}\label{thm:main}
Let $k\in\N$ with $k\geq 2$, and let $\X,\A$ be two $k$-enhanced $\sigma$-structures.
Then $\SA^k(\X,\A)$ accepts if and only if
$\X^{\tensor{k}}\to\bF_{\Qconv}(\A^{\tensor{k}})$.\footnote{An equivalent way to phrase Theorem~\ref{thm:main} is stating that $\SA^k(\X,\A)$ accepts exactly when $\BLP(\X^\tensor{k},\A^\tensor{k})$ accepts.}
\end{thm}

It follows from Theorem~\ref{thm:main} that Sherali-Adams acceptance can be cast as the problem \linebreak $\CSP(\bF_{\Qconv}(\A^{\tensor{k}}))$ (applied to the instance $\X^{\tensor{k}}$).
Hence, one can investigate the features of the Sherali-Adams hierarchy by exploring the free structure $\bF_{\Qconv}(\A^{\tensor{k}})$ -- which in this story takes the role of the main character. The domain of $\mathbb{F}_{\Qconv}(\A^{\tensor{k}})$ is $\Qconv^{(n^k)}$, which we visualise as the set of nonnegative tensors in $\cT^{k;n}(\Q)$ whose entries sum up to $1$ (as usual, we are supposing $A=[n]$). Let us now describe the relations in $\mathbb{F}_{\Qconv}(\A^{\tensor{k}})$. Take $R\in \sigma$ of arity $r$, and consider a block tensor $M\in\cT^{k;r}(\cT^{k;n}(\Q))=\cT^{k;rn}(\Q)$ whose $\bi$-th block we denote by $M_\bi$, for $\bi\in [r]^k$.
We have that $M\in R^{\mathbb{F}_{\Qconv}(\A^{\tensor{k}})}$ if and only if there exists 
$
q\in\Qconv^{(m)}
$
(where $m=\big|R^\A\big|=\big|R^{\A^\tensor{k}}\big|$, cf.~Footnote~\ref{footnote_RA_RAk})
such that $M_\bi=q_{/\pi_\bi}$ for each $\bi\in [r]^k$, where
$
\pi_\bi:R^\A\to {A}^k$ is the function defined by
$\ba\mapsto \ba_\bi
$ for $\ba\in R^\A$.
Observe that we can write $q_{/\pi_\bi}$ as $q_{/\pi_\bi}=P_\bi\ast q$, where
$
P_\bi\in \cT^{\begin{scriptsize}\underbrace{n\times\dots\times n}_k\end{scriptsize}\,\times\,m}(\{0,1\})
$
is the tensor defined by 
\begin{align}
\label{defn_tensor_Pi}
E_\ba\ast P_\bi\ast \be_{\ba'}=
\left\{
\begin{array}{cl}
1 & \mbox{if }\ba'_\bi=\ba\\
0 & \mbox{otherwise}
\end{array}
\right.
 \hspace{1cm}
\forall \ba\in {A}^k, \forall \ba'\in R^\A.
\end{align}
\begin{example}
\label{example_free_structure}
Let us take a closer look at the free structure $\mathbb{F}_{\Qconv}(\K_3^\tensor{3})$ (which we denote by $\textbf{F}$ in this example) generated by the structure $\K_3^\tensor{3}$ described in Example~\ref{example_tensorisation_1406}.
The domain of $\textbf{F}$ is the set of nonnegative tensors in $\mathcal{T}^{3,3}(\Q)$ whose entries sum up to $1$. The relation $R^{\textbf{F}}$ is the set of those tensors $M\in\mathcal{T}^{3;2}(\mathcal{T}^{3;3}(\Q))=\mathcal{T}^{3;6}(\Q)$ such that there exists a stochastic vector $q=(q_1,\dots,q_6)\in\Qconv^{(6)}$ (which should be interpreted as a probability distribution over the elements of $R^{\K_3}$) for which the $\bi$-th block $M_\bi$ of $M$ satisfies $M_\bi=q_{/\pi_\bi}$ for any $\bi\in [2]^3$. As $q_{/\pi_\bi}=P_\bi\ast q$, making use of~\eqref{defn_tensor_Pi} or Lemma~\ref{lem_multiplication_rule_P}, we find that, for example,
\begin{small}
\begin{align*}
M_{(1,1,1)}&=
\left[\begin{array}{@{}ccc|ccc|ccc@{}}
q_1+q_6&0&0&0&0&0&0&0&0\\
0&0&0&0&q_2+q_3&0&0&0&0\\
0&0&0&0&0&0&0&0&q_4+q_5
\end{array}\right]\!\!,\;\;
M_{(2,1,2)}=
\left[\begin{array}{@{}ccc|ccc|ccc@{}}
0&0&0       &0&q_1&0      &0&0&q_6\\
q_2&0&0     &0&0&0        &0&0&q_3\\
q_5&0&0     &0&q_4&0      &0&0&0
\end{array}\right]\!\!.
\end{align*}
\end{small}
Figure~\ref{fig_hypermatrix} is a visual representation of the tensor $M\in R^{\textbf{F}}$ corresponding to the uniform distribution $q=\frac{1}{6}\bone_6$.
\end{example}
Example~\ref{example_free_structure} shows that $\mathbb{F}_{\Qconv}(\A^\tensor{k})$ has a rich geometric description. In particular, Figure~\ref{fig_hypermatrix} suggests that the tensors in the relations of such a structure are typically rather sparse, and the few nonzero entries form a regular pattern. This feature becomes more evident for higher values of $k$.
Lemmas~\ref{lem_multiplication_rule_P} through~\ref{prop_correspondence_morphisms_tensor_structures}
in Section~\ref{sec_tensorisation} give more insight into the tensorisation construction and, specifically, the geometry of $\mathbb{F}_{\Qconv}(\A^\tensor{k})$. They are used to prove Theorem~\ref{thm:main} and all other results in this work.

In turn, the geometry of $\mathbb{F}_{\Qconv}(\A^\tensor{k})$ is reflected in the properties of homomorphisms from $\X^{\tensor{k}}$ to $\bF_{\Qconv}(\A^{\tensor{k}})$ -- which, by virtue of Theorem~\ref{thm:main}, correspond to Sherali-Adams solutions. 
Specifically, if $\X$ and $\A$ are $k$-enhanced, any such homomorphism $\xi$ must satisfy certain symmetries
that ultimately depend on the fact that $\xi$ preserves $R_k$. As shown below in Proposition~\ref{lem_a_symmetry} (proved in Section~\ref{sec_tensorisation}), these symmetries can be concisely expressed through a tensor equation. Given a tuple $\bi\in [k]^k$, we let $\Pi_\bi\in \cT^{2k;n}(\{0,1\})$ be the tensor defined by 
\begin{align}
\label{defn_tensor_Pi_i}
E_\ba\ast \Pi_\bi\ast E_{\ba'}=
\left\{
\begin{array}{ll}
1 & \mbox{if }\ba'_\bi=\ba\\
0 & \mbox{otherwise}
\end{array}
\right.
\hspace{1cm}
\forall \ba,\ba'\in {A}^k.
\end{align}
\begin{prop}
\label{lem_a_symmetry}
Let $k\in\N$, let $\X,\A$ be two $k$-enhanced $\sigma$-structures, and let $\xi:X^k\to\Qconv^{(n^k)}$ be a map. Then $\xi$ preserves $R_k$ if and only if 
\begin{align}
\label{eqn_consistency}
\xi(\bx_\bi)=\Pi_\bi\ast \xi(\bx)
\end{align}
for any $\bx\in X^k$, $\bi\in [k]^k$.
\end{prop}

\begin{wrapfigure}{r}{7.75cm}
\begin{tikzpicture}
\begin{scope}[3d view={110}{15},local bounding box=C,scale=.775]
\begin{scope}[shift={(0,0,0)}]
  \foreach \x in {0,...,2}
    \foreach \y in {0,...,2} 
    	\foreach \z in {0,...,2}
    		\cube{0}{\opacityDefault}{\x}{\y}{\z};
\cube{0}{.9}{0}{0}{0};
\cube{0}{.9}{1}{1}{1};
\cube{0}{.9}{2}{2}{2};
\end{scope}

%
%
%
%
\begin{scope}[shift={(0,\bigShift,0)}]
  \foreach \x in {0,...,2}
    \foreach \y in {0,...,2} 
    	\foreach \z in {0,...,2}
    		\cube{0}{\opacityDefault}{\x}{\y}{\z};
\cube{0}{.4}{2}{2}{1};  
\cube{0}{.4}{2}{2}{0};  
\cube{0}{.4}{1}{1}{2};	
\cube{0}{.4}{1}{1}{0};
\cube{0}{.4}{0}{0}{2};	
\cube{0}{.4}{0}{0}{1};
\end{scope}
%
%
%
%
%
\begin{scope}[shift={(0,0,-\bigShift)}]
  \foreach \x in {0,...,2}
    \foreach \y in {0,...,2} 
    	\foreach \z in {0,...,2}
    		\cube{0}{\opacityDefault}{\x}{\y}{\z};
\cube{0}{.4}{2}{1}{2};  
\cube{0}{.4}{2}{0}{2};  
\cube{0}{.4}{1}{2}{1};	
\cube{0}{.4}{1}{0}{1};
\cube{0}{.4}{0}{2}{0};	
\cube{0}{.4}{0}{1}{0};
\end{scope}
%
%
%
%
%
%
\begin{scope}[shift={(0,\bigShift,-\bigShift)}]
  \foreach \x in {0,...,2}
    \foreach \y in {0,...,2} 
    	\foreach \z in {0,...,2}
    		\cube{0}{\opacityDefault}{\x}{\y}{\z};
\cube{0}{.4}{2}{1}{1};  
\cube{0}{.4}{2}{0}{0};  
\cube{0}{.4}{1}{2}{2};	
\cube{0}{.4}{1}{0}{0};
\cube{0}{.4}{0}{2}{2};	
\cube{0}{.4}{0}{1}{1};
\end{scope}
%
%
%
%
%
\begin{scope}[shift={(-\bigShift-2,0,0)}]
  \foreach \x in {0,...,2}
    \foreach \y in {0,...,2} 
    	\foreach \z in {0,...,2}
    		\cube{0}{\opacityDefault}{\x}{\y}{\z};
\cube{0}{.4}{2}{1}{1};  
\cube{0}{.4}{2}{0}{0};  
\cube{0}{.4}{1}{2}{2};	
\cube{0}{.4}{1}{0}{0};
\cube{0}{.4}{0}{2}{2};	
\cube{0}{.4}{0}{1}{1};
\end{scope}
%
%
%
%
%
\begin{scope}[shift={(-\bigShift-2,\bigShift,0)}]
  \foreach \x in {0,...,2}
    \foreach \y in {0,...,2} 
    	\foreach \z in {0,...,2}
    		\cube{0}{\opacityDefault}{\x}{\y}{\z};
\cube{0}{.4}{2}{1}{2};  
\cube{0}{.4}{2}{0}{2};  
\cube{0}{.4}{1}{2}{1};	
\cube{0}{.4}{1}{0}{1};
\cube{0}{.4}{0}{2}{0};	
\cube{0}{.4}{0}{1}{0};
\end{scope}
%
%
%
%
\begin{scope}[shift={(-\bigShift-2,0,-\bigShift)}]
  \foreach \x in {0,...,2}
    \foreach \y in {0,...,2} 
    	\foreach \z in {0,...,2}
    		\cube{0}{\opacityDefault}{\x}{\y}{\z};
\cube{0}{.4}{2}{2}{1};  
\cube{0}{.4}{2}{2}{0};  
\cube{0}{.4}{1}{1}{2};	
\cube{0}{.4}{1}{1}{0};
\cube{0}{.4}{0}{0}{2};	
\cube{0}{.4}{0}{0}{1};
\end{scope}
%
%
%
%
\begin{scope}[shift={(-\bigShift-2,\bigShift,-\bigShift)}]
  \foreach \x in {0,...,2}
    \foreach \y in {0,...,2} 
    	\foreach \z in {0,...,2}
    		\cube{0}{\opacityDefault}{\x}{\y}{\z};
\cube{0}{.9}{0}{0}{0};
\cube{0}{.9}{1}{1}{1};
\cube{0}{.9}{2}{2}{2};
\end{scope}

\end{scope}

\end{tikzpicture}
\caption{A tensor $M\in R^{\textbf{F}}$ from Example~\ref{example_free_structure}, corresponding to the uniform distribution on the set of edges of $\K_3$. The opacity of a cell is proportional to the value of the corresponding entry:\hspace{.5cm} 
{\protect\begin{tikzpicture}[baseline=-1.2ex] \begin{scope}[3d view={110}{15},local bounding box=C,scale=.4]\protect\cube{0}{.9}{0}{0}{0};\end{scope}\end{tikzpicture}} = $\frac{1}{3}$,\;\;
{\protect\begin{tikzpicture}[baseline=-1.2ex] \begin{scope}[3d view={110}{15},local bounding box=C,scale=.4]\protect\cube{0}{.4}{0}{0}{0};\end{scope}\end{tikzpicture}} = $\frac{1}{6}$,\;\;
{\protect\begin{tikzpicture}[baseline=-1.2ex] \begin{scope}[3d view={110}{15},local bounding box=C,scale=.4]\protect\cube{0}{.05}{0}{0}{0};\end{scope}\end{tikzpicture}} = $0$.
}
\label{fig_hypermatrix}
\end{wrapfigure}

Equation~\eqref{eqn_consistency} -- which shall be called the \emph{consistency equation} -- is the multilinear translation of the requirement (SA2) defining the $k$-th level of Sherali-Adams. Observe that, for $k=1$, the equation is vacuous, since in this case $\Pi_\bi$ is the identity matrix of order $n$ and $\bx_\bi=\bx$. As $k$ increases, it produces a progressively richer system of symmetries that must be satisfied by $\xi$, which corresponds -- by virtue of Theorem~\ref{thm:main} -- to a progressively stronger relaxation. 
Concretely, we shall often use the consistency equation as a tool to check that a map preserves $R_k$ and, thus, can be a homomorphism.

One consequence of Theorem~\ref{thm:main}, discussed in  Section~\ref{sec_characterisation_acceptance}, is that known \emph{algorithmic} facts about the Sherali-Adams hierarchy can be revisited as \emph{geometric} properties of the tensorisation construction. In more detail, Propositions~\ref{prop_f_q_a_k} and~\ref{prop_sherali_adams_smaller_level} provide a geometric view of the fact that the Sherali-Adams hierarchy becomes more powerful as the levels increase, Proposition~\ref{lem_partial_homo} corresponds to the fact that Sherali-Adams solutions yield partial homomorphisms, while Proposition~\ref{prop_sherali_adams_exact} reflects the fact that the $k$-th level of Sherali-Adams correctly classifies instances of size $k$ or less.

The tensorisation construction also yields a method to \emph{count} the homomorphisms between two structures $\X$ and $\A$: Proposition~\ref{cor_number_homomorphisms_from_tensors} shows that the number of such homomorphisms can be read off on the image of a specific tuple under a homomorphism from $\X^\tensor{k}$ to $\mathbb{F}_{\Qconv}(\A^\tensor{k})$ having maximum support, for a suitable $k$. Finally, Proposition~\ref{prop_SA_and_automorphisms} allows to lift the symmetries of $\A$ to symmetries of the solutions of $\SA^k(\X,\A)$; more specifically, it implies that there exists a solution that is invariant under the group of automorphisms of $\A$.

\medskip

The approximate graph colouring problem is believed to be NP-hard~\cite{GJ76}; i.e., for every constant $3\leq c\leq d$, it is believed that $\PCSP(\K_c,\K_d)$ is NP-hard. Since this result seems out of reach of current techniques, we continue the line of work described briefly in Section~\ref{sec:intro} on unconditional lower bounds on small LPs and investigate the question of ruling out the Sherali-Adams hierarchy as an algorithmic tool to solve approximate graph colouring.
The second main goal of this work is to use the multilinear framework outlined above to establish the following result. 

\begin{thm}
\label{main_theorem_approximate_colouring}
No constant level of Sherali-Adams solves the approximate graph colouring problem;
i.e., for any fixed $3\leq c\leq d$, there is no constant $k$ such that the $k$-th level of Sherali-Adams solves $\PCSP(\K_c,\K_d)$.
\end{thm}

The crucial tool to prove Theorem~\ref{main_theorem_approximate_colouring} is the next proposition, which shows that the $k$-th level of the hierarchy applied to the clique of order $k$ accepts for any loopless input digraph $\X$. 
\begin{prop}
\label{prop_SA_cliques}
Let $\X$ be a loopless digraph and let $k\in\N$, $k\geq 2$. Then $\SA^k(\X,\K_k)$ accepts.
\end{prop}
This result follows naturally from the geometry of the tensors in the free structure generated by the $k$-th tensor power of $\K_k$. Essentially, the symmetries satisfied by these tensors 
are not strong enough to prevent the existence of a homomorphism from any input.
The proof proceeds by exhibiting a map $g:X^k\to\mathcal{T}^{k,k}(\Q)$ and then checking that it preserves the edge relation and it satisfies the consistency equation~\eqref{eqn_consistency}, so that it yields a homomorphism from $\X^\tensor{k}$ to $\mathbb{F}_{\Qconv}(\K_k^\tensor{k})$.\footnote{In fact, $\X$ and $\K_k$ need to be $k$-enhanced first. This is discussed in the full proof of Proposition~\ref{prop_SA_cliques} in Section~\ref{sec:approx}.}
The map $g$ assigns to any element $\bx\in X^k$ a tensor whose nonzero pattern corresponds to the set of injections from $\{\bx\}$ to $[k]$. As an example, for $k=3$ and $x_1,x_2,x_3$ distinct vertices of $\X$, we get\footnote{
Notice that, unlike the other two, the rightmost tensor in~\eqref{blocks_1408_2003} does not appear as a block in an element of $R^{\mathbb{F}_{\Qconv}(\K_3^\tensor{3})}$ (cf.~Example~\ref{example_free_structure} and Figure~\ref{fig_hypermatrix}). This does not prevent $g$ from being a homomorphism, as the tuple $(x_1,x_2,x_3)$ does not appear as an entry of an element of $R^{\X^\tensor{3}}$ (cf.~Example~\ref{example_tensorisation_1406}).
}
\begin{align}
\label{blocks_1408_2003}
\mbox{
$(x_1,x_1,x_1)\mapsto$
\begin{tikzpicture}[baseline=-4ex]
\begin{scope}[3d view={110}{15},local bounding box=C,scale=.45]
\begin{scope}[shift={(0,0,0)}]
  \foreach \x in {0,...,2}
    \foreach \y in {0,...,2} 
    	\foreach \z in {0,...,2}
    		\cube{0}{\opacityDefault}{\x}{\y}{\z};
\cube{0}{.9}{0}{0}{0};
\cube{0}{.9}{1}{1}{1};
\cube{0}{.9}{2}{2}{2};
\end{scope}
\end{scope}
\end{tikzpicture}\;,
\hspace{.5cm}
$(x_3,x_3,x_2)\mapsto$
\begin{tikzpicture}[baseline=-4ex]
\begin{scope}[3d view={110}{15},local bounding box=C,scale=.45]
\begin{scope}[shift={(0,0,0)}]
  \foreach \x in {0,...,2}
    \foreach \y in {0,...,2} 
    	\foreach \z in {0,...,2}
    		\cube{0}{\opacityDefault}{\x}{\y}{\z};
\cube{0}{.4}{2}{2}{1};  
\cube{0}{.4}{2}{2}{0};  
\cube{0}{.4}{1}{1}{2};	
\cube{0}{.4}{1}{1}{0};
\cube{0}{.4}{0}{0}{2};	
\cube{0}{.4}{0}{0}{1};
\end{scope}
\end{scope}
\end{tikzpicture}\;,
\hspace{.5cm}
$(x_1,x_2,x_3)\mapsto$
\begin{tikzpicture}[baseline=-4ex]
\begin{scope}[3d view={110}{15},local bounding box=C,scale=.45]
\begin{scope}[shift={(0,0,0)}]
  \foreach \x in {0,...,2}
    \foreach \y in {0,...,2} 
    	\foreach \z in {0,...,2}
    		\cube{0}{\opacityDefault}{\x}{\y}{\z};
\cube{0}{.4}{2}{1}{0};  
\cube{0}{.4}{2}{0}{1};  
\cube{0}{.4}{1}{2}{0};	
\cube{0}{.4}{1}{0}{2};
\cube{0}{.4}{0}{2}{1};	
\cube{0}{.4}{0}{1}{2};
\end{scope}
\end{scope}
\end{tikzpicture}\;}.
\end{align}

Taking $\X=\K_{d+1}$ in the statement of Proposition~\ref{prop_SA_cliques}, we see that $\SA^k(\K_{d+1},\K_k)$ accepts but, clearly, $\K_{d+1}\not\to\K_d$. This means that $\SA^k$ does not solve $\PCSP(\K_k,\K_d)$. By slightly modifying this argument (cf.~Remark~\ref{cor_SAk_no_solves_approx_graph_color}), one immediately derives from Proposition~\ref{prop_SA_cliques} that the $k$-th level of Sherali-Adams does not solve $\PCSP(\K_c,\K_d)$ whenever $k\leq c$.

In fact, the scope of Proposition~\ref{prop_SA_cliques} is much wider. To see its full power -- and thus to prove Theorem~\ref{main_theorem_approximate_colouring} for any choice of $c$, $d$, and $k$ -- we shall make use of the \emph{line digraph} construction. The line digraph of a digraph $\X$ is the digraph $\delta\X$ whose vertices are the arcs in $\X$ and whose arcs are pairs of consecutive arcs in $\X$. A remarkable property of this construction is that it decreases the chromatic number in a controlled way, roughly logarithmically: It was proved in~\cite{HarnerE72} that, for any natural $p$, $\delta\X\to\K_p$ implies $\X\to\K_{2^p}$, while $\X\to\K_{{p\choose \floor{p/2}}}$ implies $\delta\X\to \K_p$ (we recall that a digraph mapping homomorphically to $\K_p$ is equivalent to the chromatic number of the digraph being less than or equal to $p$). Another interesting feature of the line digraph is that it preserves acceptance by Sherali-Adams, at the only cost of halving the level.
\begin{prop}
\label{prop_reduction_line_digraph}
Let $k\in\N$ with $k\geq 2$, let $\X,\A$ be digraphs, and suppose that $\SA^{2k}(\X,\A)$ accepts. Then $\SA^k(\delta\X,\delta\A)$ accepts.
\end{prop}
The key point here is that taking the line digraph impacts \emph{logarithmically} on the chromatic number, but \emph{linearly} on the Sherali-Adams level. This is the final ingredient we need to prove Theorem~\ref{main_theorem_approximate_colouring}. The idea of the proof (whose full details are presented in Section~\ref{sec:approx}) is the following: If $k\leq c$, the result is a direct consequence of Proposition~\ref{prop_SA_cliques}, as argued above; if $k>c$, we preprocess the input digraphs through the line digraph construction (iterating it, if necessary) until we obtain a situation that is suitable for Proposition~\ref{prop_SA_cliques}, we apply it, and we then use the properties of the tensorisation construction to conclude. 
We illustrate this procedure for the case $k=7$, $c=6$. Suppose, for the sake of contradiction, that $\SA^7$ solves $\PCSP(\K_6,\K_d)$. Consider the graph $\X=\K_{2^d+1}$. Proposition~\ref{prop_SA_cliques} implies that $\SA^{20}(\X,\K_{20})$ accepts, whence it follows that $\SA^{14}(\X,\K_{20})$ accepts (since higher levels of Sherali-Adams are tighter, cf.~Proposition~\ref{prop_sherali_adams_smaller_level}). Proposition~\ref{prop_reduction_line_digraph} then yields that $\SA^7(\delta\X,\delta\K_{20})$ accepts. Theorem~\ref{thm:main} translates\footnote{Technically, to apply Theorem~\ref{thm:main} we need that the structures be $7$-enhanced. We deal with this issue in the full proof of Theorem~\ref{main_theorem_approximate_colouring} in Section~\ref{sec:approx}.} 
 this fact into a homomorphism $(\delta\X)^\tensor{7}\to \mathbb{F}_{\Qconv}((\delta\K_{20})^\tensor{7})$. Observing that ${6\choose \floor{6/2}}=20$, we use the result from~\cite{HarnerE72} to find that $\delta\K_{20}\to\K_6$. The properties of the tensorisation construction then yield $\mathbb{F}_{\Qconv}((\delta\K_{20})^\tensor{7})\to \mathbb{F}_{\Qconv}(\K_6^\tensor{7})$ so that, composing the two homomorphisms, we obtain $(\delta\X)^\tensor{7}\to \mathbb{F}_{\Qconv}(\K_6^\tensor{7})$. Translating this back via Theorem~\ref{thm:main}, we find that $\SA^7(\delta\X,\K_6)$ accepts, which means that $\delta\X\to\K_d$ as we are supposing that $\SA^7$ solves $\PCSP(\K_6,\K_d)$. Again from the result in~\cite{HarnerE72}, it follows that $\X\to\K_{2^d}$, a contradiction.

\section{Tensorisation of a relational structure}
\label{sec_tensorisation}

In this section, we present various technical results on the tensorisation construction described in Section~\ref{sec_extended_abstract}. These results shall be used throughout the rest of the paper. 

We start with a simple identity satisfied by the tensor $P_\bi$ defined in~\eqref{defn_tensor_Pi}.

\begin{lem}
\label{lem_multiplication_rule_P}
Let $k\in\N$, let $\A$ be a $\sigma$-structure, let $R\in\sigma$ of arity $r$, and consider the tuples $\ba\in {A}^k$ and $\bi\in [r]^k$. Then
\begin{align*}
E_\ba\ast P_\bi=\sum_{\substack{\bb\in R^\A\\ \bb_\bi=\ba}}\be_\bb.
\end{align*}
\end{lem}

\begin{proof}
For any $\ba'\in R^\A$, we have 
\begin{align*}
\left(\sum_{\substack{\bb\in R^\A\\ \bb_\bi=\ba}}\be_\bb\right)\ast \be_{\ba'}
=
\sum_{\substack{\bb\in R^\A\\ \bb_\bi=\ba}}(\be_\bb\ast \be_{\ba'})
=
\sum_{\substack{\bb\in R^\A\\ \bb_\bi=\ba\\ \bb=\ba'}}1
=
\left\{
\begin{array}{cl}
1 & \mbox{if }\ba'_\bi=\ba\\
0 & \mbox{otherwise}
\end{array}
\right.
=
(E_\ba\ast P_\bi)\ast\be_{\ba'},
\end{align*}
from which the result follows.
\end{proof}

Given two sets $S,T$ and two tuples $\textbf{s}=(s_1,\dots,s_k)\in S^k$, $\textbf{t}=(t_1,\dots,t_k)\in T^k$, we write $\textbf{s}\prec\textbf{t}$ if, for any $\alpha,\beta\in [k]$, $s_\alpha=s_\beta$ implies $t_\alpha=t_\beta$. We write $\textbf{s}\sim\textbf{t}$ to indicate that $\textbf{s}\prec\textbf{t}$ and $\textbf{t}\prec\textbf{s}$, while $\textbf{s}\not\prec\textbf{t}$ (resp. $\textbf{s}\not\sim\textbf{t}$) shall mean the negation of $\textbf{s}\prec\textbf{t}$ (resp. $\textbf{s}\sim\textbf{t}$). 

The next lemma shows that certain entries of a tensor in the relation $R^{\mathbb{F}_{\Qconv}(\A^{\tensor{k}})}$ (i.e., the interpretation of $R$ in the free structure of $\Qconv$ generated by $\A^\tensor{k}$) need to be zero.
\begin{lem}
\label{lem_vanishing_free_structure_0510}
Let $k\in \N$, let $\A$ be a $\sigma$-structure, let $R\in \sigma$ of arity $r$, and suppose $M=[M_\bi]_{\bi\in [r]^k}\in R^{\mathbb{F}_{\Qconv}(\A^{\tensor{k}})}$. Then $E_\ba\ast M_\bi=0$ for any $\bi\in [r]^k$, $\ba\in {A}^k$ such that $\bi\not\prec\ba$.
\end{lem}

\begin{proof}
Observe that there exists $q\in\Qconv^{(|R^\A|)}$ such that $M_\bi=q_{/\pi_\bi}$ for each $\bi\in [r]^k$. Using Lemma~\ref{lem_multiplication_rule_P}, we obtain
\begin{align*}
E_\ba\ast M_\bi 
&=
E_\ba\ast q_{/\pi_\bi}
=
E_\ba\ast P_\bi\ast q
=
\sum_{\substack{\bb\in R^\A\\ \bb_\bi=\ba}}\be_\bb \ast q
=
0,
\end{align*}
where the last equality follows from the fact that $\bb_\bi=\ba$ implies $\bi\prec\ba$; indeed, in that case, $i_\alpha=i_\beta$ implies $a_\alpha=b_{i_\alpha}=b_{i_\beta}=a_\beta$.
\end{proof}

The next result follows directly from Lemma~\ref{lem_vanishing_free_structure_0510}.
\begin{lem}
\label{cor_vanishing_contractions}
Let $k\in\N$, let $\X,\A$ be two $k$-enhanced $\sigma$-structures, and let $\xi:\X^{\tensor{k}}\to\mathbb{F}_{\Qconv}(\A^{\tensor{k}})$ be a homomorphism. Then $E_\ba\ast \xi(\bx)=0$ for any $\bx\in X^k$, $\ba\in A^k$ such that $\bx\not\prec\ba$.  
\end{lem}

\begin{proof}
From $\bx\in X^k=R_k^\X$, we derive $\bx^\tensor{k}\in R_k^{\X^\tensor{k}}$; since $\xi$ is a homomorphism, this yields $\xi(\bx^\tensor{k})\in R_k^{\mathbb{F}_{\Qconv}(\A^{\tensor{k}})}$. Writing $\xi(\bx^\tensor{k})$ in block form as $\xi(\bx^\tensor{k})=[\xi(\bx_\bi)]_{\bi\in [k]^k}$ and applying Lemma~\ref{lem_vanishing_free_structure_0510}, we obtain $E_\ba\ast\xi(\bx_\bi)=0$ for any $\bi\in [k]^k$ such that $\bi\not\prec\ba$. Write $\bx=(x_1,\dots,x_k)$ and $\ba=(a_1,\dots,a_k)$. Since $\bx\not\prec\ba$, there exist $\alpha,\beta\in [k]$ such that $x_\alpha=x_\beta$ and $a_\alpha\neq a_\beta$. Let $\bi'\in [k]^k$ be the tuple obtained from $(1,\dots,k)$ by replacing the $\beta$-th entry with $\alpha$. Observe that $\bx_{\bi'}=\bx$ and $\bi'\not\prec\ba$. Hence,
\begin{align*}
0=E_\ba\ast\xi(\bx_{\bi'})=E_\ba\ast\xi(\bx),
\end{align*}
as required.
\end{proof}

Using Lemma~\ref{cor_vanishing_contractions}, we can obtain some more information on the image of a homomorphism from $\X^{\tensor{k}}$ to $\mathbb{F}_{\Qconv}(\A^{\tensor{k}})$. 

\begin{lem}
\label{lem_ea_ast_Q}
Let $k\in\N$ with $k\geq 2$, let $\X,\A$ be two $k$-enhanced $\sigma$-structures, and let $\xi:\X^{\tensor{k}}\to\mathbb{F}_{\Qconv}(\A^{\tensor{k}})$ be a homomorphism. For $R\in\sigma$ of arity $r$, let $\bx\in R^\X$ and $\ba\in R^\A$ be such that $\bx\not\prec\ba$. Let $q\in\Qconv^{(|R^\A|)}$ be such that $q_{/\pi_\bi}=\xi(\bx_\bi)$ for each $\bi\in [r]^k$. Then $\be_\ba\ast q=0$.
\end{lem}

\begin{proof}
Write $\bx=(x_1,\dots,x_r)$ and $\ba=(a_1,\dots,a_r)$. Since $\bx\not\prec\ba$, there exist $\alpha,\beta\in [r]$ such that $x_\alpha=x_\beta$ and $a_\alpha\neq a_\beta$. Take the tuple $\bi=(\alpha,\alpha,\dots,\alpha,\beta)\in [r]^k$, and consider $\bx_\bi=(x_\alpha,x_\alpha,\dots,x_\alpha,x_\beta)\in X^k$, $\ba_\bi=(a_\alpha,a_\alpha,\dots,a_\alpha,a_\beta)\in A^k$. Notice that $\bx_\bi\not\prec\ba_\bi$ since $k\geq 2$. Using Lemma~\ref{cor_vanishing_contractions}, we find 
\begin{align*}
0=E_{\ba_\bi}\ast\xi(\bx_\bi)=E_{\ba_\bi}\ast q_{/\pi_\bi}
=
E_{\ba_\bi}\ast P_\bi \ast q.
\end{align*}
Observe that $E_{\ba_\bi}\ast P_\bi\ast \be_\ba=1$, so $E_{\ba_\bi}\ast P_\bi\geq \be_\ba$ entrywise. Since $q$ is nonnegative, it follows that $E_{\ba_\bi}\ast P_\bi \ast q\geq \be_\ba\ast q$. Therefore, we obtain $\be_\ba\ast q=0$.

\end{proof}

Next, we present a simple identity satisfied by the tensor $\Pi_\bi$ defined in~\eqref{defn_tensor_Pi_i} (which should be compared to the one in Lemma~\ref{lem_multiplication_rule_P} concerning $P_\bi$). 
\begin{lem}
\label{lem_multiplication_rule}
For any $\ba\in {A}^k$ and $\bi\in [k]^k$, the following identity holds:
\begin{align*}
E_\ba\ast \Pi_\bi=\sum_{\substack{\bb\in {A}^k\\\bb_\bi=\ba}}E_\bb.
\end{align*}
\end{lem}
\begin{proof}
For any $\ba'\in {A}^k$, we have
\begin{align*}
\left(\sum_{\substack{\bb\in {A}^k\\\bb_\bi=\ba}}E_\bb\right)\ast E_{\ba'}=
\sum_{\substack{\bb\in {A}^k\\\bb_\bi=\ba}}(E_\bb\ast E_{\ba'})
=
\sum_{\substack{\bb\in {A}^k\\\bb_\bi=\ba\\ \bb=\ba'}}1=
\left\{
\begin{array}{ll}
1 & \mbox{if }\ba'_\bi=\ba\\
0 & \mbox{otherwise}
\end{array}
\right.
=
\left(E_\ba\ast \Pi_\bi\right)\ast E_{\ba'},
\end{align*}
from which the result follows.
\end{proof}

We now have the necessary tools to prove Proposition~\ref{lem_a_symmetry}.

\begin{prop*}[Proposition~\ref{lem_a_symmetry} restated]
Let $k\in\N$, let $\X,\A$ be two $k$-enhanced $\sigma$-structures, and let $\xi:X^k\to\Qconv^{(n^k)}$ be a map. Then $\xi$ preserves $R_k$ if and only if 
\begin{align*}
\xi(\bx_\bi)=\Pi_\bi\ast \xi(\bx)
\end{align*}
for any $\bx\in X^k$, $\bi\in [k]^k$.
\end{prop*}

\begin{proof}
Suppose that $\xi$ preserves $R_k$. Observe that $|R_k^\A|=|A^k|=n^k$.
Take $\bx\in X^k=R_k^\X$, so $\bx^\tensor{k}\in R_k^{\X^\tensor{k}}$; since $\xi$ preserves $R_k$, this yields $\xi(\bx^\tensor{k})\in R_k^{\mathbb{F}_{\Qconv}(\A^{\tensor{k}})}$. Therefore, $\exists q\in \Qconv^{(n^k)}$ such that $\xi(\bx_\bi)=q_{/\pi_\bi}=P_\bi\ast q$ $\forall \bi\in [k]^k$. Consider now the tensor ${Q}\in\cT^{k;n}(\Q)$ defined by $E_\ba\ast{Q}=\be_\ba\ast q$ for each $\ba\in A^k$. In other words, $q$ and ${Q}$ contain the same elements; however, in the former they are arranged as a long vector, while in the latter they form a cubical tensor. We claim that $\xi(\bx_\bi)=\Pi_\bi\ast{Q}$ for each $\bi\in [k]^k$. Indeed, for any $\ba\in A^k$, we find
\begin{align*}
E_\ba\ast\Pi_\bi\ast{Q}
&=
\sum_{\substack{\bb\in A^k\\ \bb_\bi=\ba}}E_\bb\ast{Q}
=
\sum_{\substack{\bb\in A^k\\ \bb_\bi=\ba}}\be_\bb\ast q
=
E_\ba\ast P_\bi\ast q
=
E_\ba\ast \xi(\bx_\bi),
\end{align*}
where the first and third equalities come from Lemma~\ref{lem_multiplication_rule} and Lemma~\ref{lem_multiplication_rule_P}, respectively. Hence, the claim is proved. Consider now the tuple $\bi'=(1,\dots,k)\in [k]^k$, and observe that $\bx_{\bi'}=\bx$. Noticing that the contraction by $\Pi_{\bi'}$ acts on $\cT^{k;n}(\Q)$ as the identity, we conclude that 
\begin{align*}
\xi(\bx)=\xi(\bx_{\bi'})=\Pi_{\bi'}\ast{Q}={Q},
\end{align*}
whence the result follows. Conversely, suppose that $\xi(\bx_\bi)=\Pi_\bi\ast \xi(\bx)$
for any $\bx\in X^k$, $\bi\in [k]^k$. Take $\bx\in X^k=R_k^\X$, so $\bx^\tensor{k}\in R_k^{\X^\tensor{k}}$. We need to show that $\xi(\bx^\tensor{k})\in R_k^{\mathbb{F}_{\Qconv}(\A^\tensor{k})}$. Consider the tuple $q\in\Qconv^{(|R_k^\A|)}=\Qconv^{(n^k)}$ defined by $\be_\ba\ast q= E_\ba\ast\xi(\bx)$ for $\ba\in A^k$. Given $\bi\in [k]^k$ and $\ba\in A^k$, using Lemma~\ref{lem_multiplication_rule_P} and Lemma~\ref{lem_multiplication_rule}, we find
\begin{align*}
E_\ba\ast q_{/\pi_\bi}
&=
E_\ba\ast P_\bi\ast q
=
\sum_{\substack{\bb\in A^k\\\bb_\bi=\ba}}\be_\bb\ast q
=
\sum_{\substack{\bb\in A^k\\\bb_\bi=\ba}}E_\bb\ast\xi(\bx)
=
E_\ba\ast\Pi_\bi\ast\xi(\bx)
=
E_\ba\ast\xi(\bx_\bi),
\end{align*}
which shows that $q_{/\pi_\bi}=\xi(\bx_\bi)$. It follows that $\xi(\bx^\tensor{k})\in R_k^{\mathbb{F}_{\Qconv}(\A^\tensor{k})}$ and, hence, $\xi$ preserves $R_k$.
\end{proof}

Building on Proposition~\ref{lem_a_symmetry}, the next result contains an invariance property of homomorphisms from $\X^{\tensor{k}}$ to $\mathbb{F}_{\Qconv}(\A^{\tensor{k}})$.
\begin{lem}
\label{lem_relabelling_is_fine}
Let $k\in\N$, let $\X,\A$ be two $k$-enhanced $\sigma$-structures, and let $\xi:\X^{\tensor{k}}\to\mathbb{F}_{\Qconv}(\A^{\tensor{k}})$ be a homomorphism. Then
\begin{align*}
E_{\ba_\bi}\ast \xi(\bx_\bi)=E_\ba\ast\xi(\bx)
\end{align*}
for any $\ba\in {A}^k$, $\bx\in X^k$, $\bi\in [k]^k$ such that $\bx\prec\ba$ and $\{\bx_\bi\}=\{\bx\}$.
\end{lem}

\begin{proof}
Using Proposition~\ref{lem_a_symmetry} and Lemma~\ref{lem_multiplication_rule}, we find
\begin{align*}
E_{\ba_\bi}\ast \xi(\bx_\bi)&=E_{\ba_\bi}\ast \Pi_\bi\ast \xi(\bx)
=
\sum_{\substack{\bb\in {A}^k\\\bb_\bi=\ba_\bi}}E_\bb\ast \xi(\bx)
=
E_\ba\ast \xi(\bx)+
\sum_{\substack{\bb\in {A}^k\\\bb_\bi=\ba_\bi\\\bb\neq\ba}}E_\bb\ast \xi(\bx).
\end{align*}
Applying Lemma~\ref{cor_vanishing_contractions} allows to further refine the expression above, thus giving
\begin{align*}
E_{\ba_\bi}\ast \xi(\bx_\bi)=
E_\ba\ast \xi(\bx)+
\sum_{\substack{\bb\in {A}^k\\\bb_\bi=\ba_\bi\\\bb\neq\ba\\\bx\prec \bb}}E_\bb\ast \xi(\bx).
\end{align*}
The result follows once we show that the second term of the right-hand side in the identity above is zero. For the sake of contradiction, suppose that there exists $\bb\in {A}^k$ such that $\bb_\bi=\ba_\bi$, $\bb\neq \ba$, and $\bx\prec\bb$. Henceforth in this proof, it shall be convenient to consider tuples as functions: We write $\ba,\bb:[k]\to {A}$, $\bi:[k]\to [k]$, and $\bx:[k]\to X$. In this notation, the assumptions yield $\Imago(\bx\circ\bi)=\Imago(\bx)$ and $\bb\circ\bi=\ba\circ\bi$ (where $\Imago$ denotes the image of the function). From $\bb\neq\ba$, we deduce that $\bb(\alpha)\neq \ba(\alpha)$ for some $\alpha\in [k]$. Since $\bx(\alpha)\in\Imago(\bx)=\Imago(\bx\circ\bi)$, there exists $\beta\in [k]$ such that $\bx(\alpha)=\bx(\bi(\beta))=\bx(\gamma)$, where we have set $\gamma\coloneqq\bi(\beta)$. Since $\bx\prec\ba$ and $\bx\prec\bb$, this implies $\ba(\alpha)=\ba(\gamma)$ and $\bb(\alpha)=\bb(\gamma)$. On the other hand, $\bb(\gamma)=(\bb\circ \bi)(\beta)=(\ba\circ \bi)(\beta)=\ba(\gamma)$. We conclude that $\ba(\alpha)=\bb(\alpha)$, which contradicts our hypothesis.
\end{proof}

The next lemma shows that the tensorisation construction does not alter whether two structures are homomorphic or not. We let $\Hom(\A,\B)$ denote the set of homomorphisms from $\A$ to $\B$.

\begin{lem}
\label{prop_correspondence_morphisms_tensor_structures}
Let $k\in \N$ and let $\A,\B$ be two $\sigma$-structures. Then
\begin{enumerate}
\item[$(i)$]
 $\A\to\B$ if and only if $\A^\tensor{k}\to\B^\tensor{k}$;
\item[$(ii)$]
if $\A$ is $k$-enhanced, there is a bijection $\rho:\Hom(\A,\B)\to\Hom(\A^\tensor{k},\B^\tensor{k})$.
\end{enumerate}
\end{lem}

\begin{proof}
Let ${f}:\A\to\B$ be a homomorphism, and consider the function $f^\ast:A^k\to B^k$ defined by $f^\ast((a_1,\dots,a_k))\coloneqq ({f}(a_1),\dots,{f}(a_k))$. Take $R\in\sigma$ of arity $r$, and consider $\ba^\tensor{k}\in R^{\A^\tensor{k}}$, where $\ba\in R^\A$. Since ${f}$ is a homomorphism, $\bb\coloneqq {f}(\ba)\in R^\B$, so $\bb^\tensor{k}\in R^{\B^\tensor{k}}$. For any $\bi\in [r]^k$, we have 
\begin{align*}
E_\bi\ast f^\ast\left(\ba^\tensor{k}\right)&=f^\ast\left(E_\bi\ast\ba^\tensor{k}\right)=f^\ast(\ba_\bi)={f}(\ba_\bi)=({f}(\ba))_\bi=\bb_\bi=E_\bi\ast\bb^\tensor{k},
\end{align*}
which yields $f^\ast(\ba^\tensor{k})=\bb^\tensor{k}\in R^{\B^\tensor{k}}$. Hence, $f^\ast:\A^\tensor{k}\to\B^\tensor{k}$ is a homomorphism.

Conversely, let $g:\A^\tensor{k}\to\B^\tensor{k}$ be a homomorphism. We define the function $g_\ast:A\to B$ by setting $g_\ast(a)\coloneqq\be_1^Tg((a,\dots,a))$ for each $a\in A$. Take $R\in\sigma$ of arity $r$, and consider a tuple $\ba=(a_1,\dots,a_r)\in R^\A$. Since $\ba^\tensor{k}\in R^{\A^\tensor{k}}$ and $g$ is a homomorphism, we have that $g(\ba^\tensor{k})\in R^{\B^\tensor{k}}$. Therefore, $g(\ba^\tensor{k})=\bb^\tensor{k}$ for some $\bb=(b_1,\dots,b_r)\in R^\B$. For each $j\in [r]$, consider the tuple $\bi=(j,\dots,j)\in [r]^k$ and observe that
\begin{align*}
g((a_j,\dots,a_j))=g(\ba_\bi)=g\left(E_\bi\ast\ba^\tensor{k}\right)=E_\bi\ast g(\ba^\tensor{k})
=
E_\bi\ast\bb^\tensor{k}
=\bb_\bi
=
(b_j,\dots,b_j).
\end{align*}
Hence, we find
\begin{align*}
g_\ast(\ba)=\left(\be_1^Tg((a_1,\dots,a_1)),\dots,\be_1^Tg((a_r,\dots,a_r))\right)
=
\left(b_1,\dots,b_r\right)=\bb\in R^\B.
\end{align*}
Therefore, $g_\ast:\A\to\B$ is a homomorphism. This concludes the proof of $(i)$.

To prove $(ii)$, observe first that, if $\A\not\to\B$, then $\Hom(\A,\B)=\Hom(\A^\tensor{k},\B^\tensor{k})=\emptyset$, so there is a trivial bijection in this case. If $\A\to\B$, consider the map $\rho:\Hom(\A,\B)\to\Hom(\A^\tensor{k},\B^\tensor{k})$ defined by $f\mapsto f^\ast$ and the map $\rho':\Hom(\A^\tensor{k},\B^\tensor{k})\to\Hom(\A,\B)$ defined by $g\mapsto g_\ast$. For $f:\A\to\B$ and $a\in A$, we have 
\begin{align*}
(f^\ast)_\ast(a)
=
\be_1^Tf^\ast((a,\dots,a))
=
\be_1^T(f(a),\dots,f(a))=f(a)
\end{align*}
so that $\rho'\circ\rho=\id_{\Hom(\A,\B)}$. Consider now $g:\A^\tensor{k}\to\B^\tensor{k}$, and take $\ba=(a_1,\dots,a_k)\in A^k$. Using the assumption that $\A$ is $k$-enhanced, we have $\ba\in R_k^\A$, which implies $\ba^\tensor{k}\in R_k^{\A^\tensor{k}}$. Hence, $g(\ba^\tensor{k})\in R_k^{\B^\tensor{k}}$, so $g(\ba^\tensor{k})=\bb^\tensor{k}$ for some $\bb=(b_1,\dots,b_k)\in R_k^\B\subseteq B^k$. For $j\in [k]$ and $\bi=(j,\dots,j)\in [k]^k$, we have
\begin{align*}
g((a_j,\dots,a_j))
=
g(\ba_\bi)
=
g\left(E_\bi\ast\ba^\tensor{k}\right)
=
E_\bi\ast g\left(\ba^\tensor{k}\right)
=
E_\bi\ast\bb^\tensor{k}
=
\bb_\bi
=
(b_j,\dots,b_j).
\end{align*}
Letting $\bi'=(1,\dots,k)\in [k]^k$, we obtain
\begin{align*}
(g_\ast)^\ast(\ba)
&=
(g_\ast(a_1),\dots,g_\ast(a_k))
=
\left(\be_1^Tg((a_1,\dots,a_1)),\dots,\be_1^Tg((a_k,\dots,a_k))\right)
=
(b_1,\dots,b_k)\\
&=
\bb
=
\bb_{\bi'}
=
E_{\bi'}\ast\bb^\tensor{k}
=
E_{\bi'}\ast g\left(\ba^\tensor{k}\right)
=
g\left(E_{\bi'}\ast\ba^\tensor{k}\right)
=
g(\ba_{\bi'})
=
g(\ba),
\end{align*}
so that $\rho\circ\rho'=\id_{\Hom(\A^\tensor{k},\B^\tensor{k})}$, which concludes the proof of $(ii)$. 
\end{proof}

\begin{rem}
Part $(ii)$ of Lemma~\ref{prop_correspondence_morphisms_tensor_structures} does not hold in general if we relax the requirement that $\A$ be $k$-enhanced. More precisely, in this case, the function $\rho:\Hom(\A,\B)\to\Hom(\A^\tensor{k},\B^\tensor{k})$ defined in the proof of Lemma~\ref{prop_correspondence_morphisms_tensor_structures} still needs to be injective, but may not be surjective. Therefore, we have $\left |\Hom(\A,\B)\right |\leq\left |\Hom(\A^\tensor{k},\B^\tensor{k})\right |$, and the inequality may be strict.

For example, consider the Boolean structure $\A$ having a unique unary relation $R_1^\A=A=\{0,1\}$. So, $\A$ is $1$-enhanced but not $2$-enhanced. Observe that $|\Hom(\A,\A)|=4$. The tensorised structure $\A^\tensor{2}$ has domain $\{0,1\}^2=\{(0,0),(0,1),\allowbreak(1,0),(1,1)\}$, and its (unary) relation is $R_1^{\A^\tensor{2}}=\{0^\tensor{2},1^\tensor{2}\}=\{(0,0),(1,1)\}$. Therefore, each map $f:\{0,1\}^2\to\{0,1\}^2$ such that $f((0,0))\in\{(0,0),(1,1)\}$ and $f((1,1))\in\{(0,0),\allowbreak(1,1)\}$ yields a proper homomorphism $\A^\tensor{2}\to\A^\tensor{2}$. Hence, $\left|\Hom(\A^\tensor{2},\A^\tensor{2})\right|=64$, so $\Hom(\A,\A)$ and $\Hom(\A^\tensor{2},\A^\tensor{2})$ are not in bijection.
\end{rem}

\section{Sherali-Adams acceptance}
\label{sec_characterisation_acceptance}
In this section, we connect the tensorisation construction to the Sherali-Adams hierarchy: We characterise the instances $\X$ for which the $k$-th level of Sherali-Adams accepts as those
for which $\X^\tensor{k}$ is homomorphic to the free structure of the minion $\Qconv$ generated by $\A^\tensor{k}$, as stated in the following main result. 
\begin{thm*}[Theorem~\ref{thm:main} restated]
Let $k\in\N$ with $k\geq 2$, and let $\X,\A$ be two $k$-enhanced $\sigma$-structures.
Then $\SA^k(\X,\A)$ accepts if and only if
$\X^{\tensor{k}}\to\bF_{\Qconv}(\A^{\tensor{k}})$.
\end{thm*}
The two implications of Theorem~\ref{thm:main} are proved separately, in Propositions~\ref{prop_first_implication_homo_SAk} and~\ref{prop_SAkaccepts_hom_tensor}. The ``only if'' implication, established in Proposition~\ref{prop_SAkaccepts_hom_tensor}, actually works for $k\geq 1$ and does not require $k$-enhancement of the structures $\X$ and $\A$. In any case, assuming $k$-enhancement results in no loss of generality, in the sense that, as noted in Section~\ref{sec_extended_abstract}, any $\PCSP$ template is equivalent to the $\PCSP$ template in which both structures are $k$-enhanced. 
Subsequently, we show how certain known facts and new simple results on the Sherali-Adams hierarchy naturally follow from the geometry of $\bF_{\Qconv}(\A^{\tensor{k}})$ through this characterisation.

Given two sets $S,T$, an integer $p\in\N$, and two tuples $\textbf{s}=(s_1,\dots,s_p)\in S^p, \textbf{t}=(t_1,\dots,t_p)\in T^p$ such that $\textbf{s}\prec\textbf{t}$, we shall consider the function $f_{\textbf{s},\textbf{t}}:\{\textbf{s}\}\to T$ defined by $f_{\textbf{s},\textbf{t}}(s_\alpha)=t_\alpha$ for each $\alpha\in [p]$.

\begin{prop}
\label{prop_first_implication_homo_SAk}
Let $k\in\N$ with $k\geq 2$, let $\X,\A$ be two $k$-enhanced $\sigma$-structures, and let $\xi:\X^{\tensor{k}}\to\mathbb{F}_{\Qconv}(\A^{\tensor{k}})$ be a homomorphism. Then $\SA^k(\X,\A)$ accepts.
\end{prop}
\begin{proof}
Given $\bx=(x_1,\dots,x_k)\in X^k$ and $\ba=(a_1,\dots,a_k)\in A^k$ such that $\bx\prec\ba$, we define 
\begin{align}
\label{eqn_variables_sherali_adams_from_homo}
\lambda_{\{\bx\}}(f_{\bx,\ba})\coloneqq E_\ba\ast \xi(\bx).
\end{align}
Clearly, any pair $(V,f)$ with $V\subseteq X$, $1\leq |V|\leq k$, and $f:V\to A$ can be written as $(\{\bx\},f_{\bx,\ba})$ for some such tuples $\bx$ and $\ba$. We claim that the assignment~\eqref{eqn_variables_sherali_adams_from_homo} is consistent; i.e., that $E_\ba\ast\xi(\bx)=E_{\ba'}\ast\xi(\bx')$ whenever $\bx'\in X^k$ and $\ba'=(a'_1,\dots,a'_k)\in A^k$ are such that $\bx'\prec\ba'$, $\{\bx'\}=\{\bx\}$, and $f_{\bx',\ba'}=f_{\bx,\ba}$. From $\{\bx'\}=\{\bx\}$, it follows that $\bx'=\bx_\bi$ for some $\bi=(i_1,\dots,i_k)\in [k]^k$. From $f_{\bx',\ba'}=f_{\bx,\ba}$, it follows that, for any $\alpha\in [k]$,
\begin{align*}
a_{i_\alpha}=f_{\bx,\ba}(x_{i_\alpha})=f_{\bx',\ba'}(x_{i_\alpha})
=
f_{\bx_\bi,\ba'}(x_{i_\alpha})
=a'_\alpha,
\end{align*}
so that $\ba'=\ba_\bi$. Then, the claim can be rewritten as $E_\ba\ast\xi(\bx)=E_{\ba_\bi}\ast\xi(\bx_\bi)$, which directly follows from Lemma~\ref{lem_relabelling_is_fine}. 

Consider now $R\in\sigma$ of arity $r$, $\bx\in R^\X$, $\ba\in A^r$ such that $\bx\prec \ba$. Let $m=|R^\A|$. Since $\bx\in R^\X$, we have $\bx^\tensor{k}\in R^{\X^\tensor{k}}$; being $\xi$ a homomorphism, this implies that $\xi(\bx^\tensor{k})\in R^{\mathbb{F}_{\Qconv}(\A^{\tensor{k}})}$ -- i.e., there exists $q^\bx\in \Qconv^{(m)}$ such that $\xi(\bx_\bi)=\xi(E_\bi\ast\bx^\tensor{k})={q^\bx}_{/\pi_\bi}$ for each $\bi\in [r]^k$. We define 
\begin{align}
\label{eqn_constraints_sherali_adams_from_homo}
\lambda_{R,\bx}(f_{\bx,\ba})\coloneqq 
\left\{
\begin{array}{cl}
\be_\ba\ast q^\bx & \mbox{if } \ba\in R^\A\\
0 & \mbox{otherwise}.
\end{array}
\right.
\end{align}
Clearly, any function $f:\{\bx\}\to A$ can be (uniquely) written as $f_{\bx,\ba}$ for some tuple $\ba$. We claim that~\eqref{eqn_variables_sherali_adams_from_homo} and~\eqref{eqn_constraints_sherali_adams_from_homo} constitute a proper assignment for $\SA^k(\X,\A)$ as defined in Section~\ref{sec:prelims}.

To check (SA1), let $V=\{x_1,\dots,x_k\}\subseteq X$, and let $\bx=(x_1,\dots,x_k)$.
\begin{align*}
\sum_{f:V\to A} \lambda_V(f)&=\sum_{\substack{\ba\in A^k\\\bx\prec\ba}}\lambda_{\{\bx\}}(f_{\bx,\ba})=
\sum_{\substack{\ba\in A^k\\\bx\prec\ba}}E_\ba\ast\xi(\bx)
=
\sum_{\ba\in A^k}E_\ba\ast\xi(\bx)=1,
\end{align*}
where the third equality follows from Lemma~\ref{cor_vanishing_contractions}.

To check (SA2), consider $U\subseteq V\subseteq X$ such that $1\leq |V|\leq k$ and $f:U\to A$. We can write $V=\{\bx\}$ for some $\bx\in X^k$, $U=\{\bx_\bi\}$ for some $\bi\in [k]^k$, and $f=f_{\bx_\bi,\ba}$ for some $\ba\in A^k$ such that $\bx_\bi\prec\ba$. We obtain
\begin{align*}
\sum_{\substack{g:V\to A\\ g|_U=f}}\lambda_V(g)
&=\sum_{\substack{\bb\in A^k\\ \bx\prec\bb\\ \bb_\bi=\ba}}\lambda_{\{\bx\}}(f_{\bx,\bb})
=
\sum_{\substack{\bb\in A^k\\ \bx\prec\bb\\ \bb_\bi=\ba}}E_\bb\ast\xi(\bx)
=
\sum_{\substack{\bb\in A^k\\ \bb_\bi=\ba}}E_\bb\ast\xi(\bx)
=
E_\ba\ast\Pi_\bi\ast\xi(\bx)
=
E_\ba\ast \xi(\bx_\bi)
\\
&=
\lambda_{\{\bx_\bi\}}(f_{\bx_\bi,\ba})
=
\lambda_U(f),
\end{align*}
where the third, fourth, and fifth equalities come from Lemma~\ref{cor_vanishing_contractions}, Lemma~\ref{lem_multiplication_rule}, and Proposition~\ref{lem_a_symmetry}, respectively.

To check (SA3), consider $R\in\sigma$ of arity $r$, $\bx\in R^\X$, $U\subseteq \{\bx\}$ such that $1\leq |U|\leq k$, and $f:U\to A$. Write $U=\{\bx_\bi\}$ for some $\bi\in [r]^k$ and $f=f_{\bx_\bi,\ba}$ for some $\ba\in A^k$ such that $\bx_\bi\prec\ba$.
\begin{align*}
\sum_{\substack{g:\{\bx\}\to A\\ g|_U=f}}\lambda_{R,\bx}(g)
&=
\sum_{\substack{\bb\in A^r\\ \bx\prec\bb\\ \bb_\bi=\ba}}\lambda_{R,\bx}(f_{\bx,\bb})
=
\sum_{\substack{\bb\in R^\A\\ \bx\prec\bb\\ \bb_\bi=\ba}}
\be_\bb\ast q^\bx
=
\sum_{\substack{\bb\in R^\A\\ \bb_\bi=\ba}}
\be_\bb\ast q^\bx
=
E_\ba\ast P_\bi \ast q^\bx
=
E_\ba\ast (q^\bx_{/\pi_\bi})\\
&=
E_\ba\ast \xi(\bx_\bi)
=
\lambda_{\{\bx_\bi\}}(f_{\bx_\bi,\ba})
=
\lambda_U(f),
\end{align*}
where the third and fourth equalities come from Lemma~\ref{lem_ea_ast_Q} and Lemma~\ref{lem_multiplication_rule_P}, respectively.
Finally, (SA4) directly follows from~\eqref{eqn_constraints_sherali_adams_from_homo}.

By the definition of the minion $\Qconv$, both~\eqref{eqn_variables_sherali_adams_from_homo} and~\eqref{eqn_constraints_sherali_adams_from_homo} assign rational values in the interval $[0,1]$ to the variables. We conclude that $\SA^k(\X,\A)$ accepts.  
\end{proof}

\begin{prop}
\label{prop_SAkaccepts_hom_tensor}
Let $k\in\N$, let $\X,\A$ be two $\sigma$-structures, and suppose that $\SA^k(\X,\A)$ accepts. Then $\X^\tensor{k}\to\mathbb{F}_{\Qconv}(\A^{\tensor{k}})$.
\end{prop}
\begin{proof}
Consider a rational solution of $\SA^k(\X,\A)$ expressed in the notation of Section~\ref{sec:prelims}, and let $\xi:X^k\to\cT^{k;n}(\Q)$ be the function defined by
\begin{align}
\label{eqn_defn_xi_1437}
E_\ba\ast\xi(\bx)=
\left\{
\begin{array}{ll}
\lambda_{\{\bx\}}(f_{\bx,\ba}) & \mbox{if }\bx\prec\ba\\
0 & \mbox{otherwise}
\end{array}
\right.
\hspace{1cm}
\forall \bx\in X^k, \forall \ba\in {A}^k.
\end{align}
Observe that $\xi(\bx)\geq 0$ entrywise, and
\begin{align*}
\sum_{\ba\in {A}^k}E_\ba\ast\xi(\bx)
=
\sum_{\substack{\ba\in {A}^k\\ \bx\prec\ba}}\lambda_{\{\bx\}}(f_{\bx,\ba})=\sum_{f:\{\bx\}\to A}\lambda_{\{\bx\}}(f)=1
\end{align*}
by (SA1). Therefore, $\xi(\bx)\in\Qconv^{(n^k)}$. Take $R\in\sigma$ of arity $r$, and let $\bx=(x_1,\dots,x_r)\in R^\X$. We need to show that $\xi(\bx^\tensor{k})\in R^{\mathbb{F}_{\Qconv}(\A^{\tensor{k}})}$. Let $m=|R^\A|$, and take $q\in \cT^{1;m}(\Q)$ defined by
\begin{align*}
\be_\ba\ast q=
\left\{
\begin{array}{ll}
\lambda_{R,\bx}(f_{\bx,\ba}) & \mbox{if }\bx\prec\ba\\
0 & \mbox{otherwise}
\end{array}
\right.
\hspace{1cm}
\forall \ba\in R^\A.
\end{align*}
Observe that $q\geq 0$ entrywise, and 
\begin{align*}
\sum_{\ba\in R^\A}\be_\ba\ast q
&=
\sum_{\substack{\ba\in R^\A\\ \bx\prec\ba}}\lambda_{R,\bx}(f_{\bx,\ba})
=
\sum_{\substack{\ba\in {A}^r\\ \bx\prec\ba}}\lambda_{R,\bx}(f_{\bx,\ba})
=
\sum_{g:\{\bx\}\to A}\lambda_{R,\bx}(g)
=
\sum_{a\in A}\sum_{\substack{g:\{\bx\}\to A\\ g(x_1)=a}}\lambda_{R,\bx}(g)\\
&=
\sum_{a\in A}\lambda_{\{x_1\}}(f_{x_1,a})
=
1,
\end{align*}
where the second, fifth, and sixth equalities come from (SA4), (SA3), and (SA1), respectively. It follows that $q\in \Qconv^{(m)}$. We claim that $q_{/\pi_\bi}=\xi(\bx_\bi)$ for each $\bi\in [r]^k$. For $\ba\in {A}^k$, we find
\begin{align}
\label{eqn_1438_0510}
E_\ba\ast q_{/\pi_\bi}=E_\ba\ast P_\bi\ast q 
=
\sum_{\substack{\bb\in R^\A\\ \bb_\bi=\ba}}\be_\bb\ast q
=
\sum_{\substack{\bb\in R^\A\\ \bb_\bi=\ba\\ \bx\prec\bb}}\lambda_{R,\bx}(f_{\bx,\bb}).
\end{align}
For each $\bb\in R^\A$ satisfying $\bx\prec\bb$, we have $\bx_\bi\prec\bb_\bi$. Hence, if $\bx_\bi\not\prec\ba$, then $E_\ba\ast q_{/\pi_\bi}=0$. From~\eqref{eqn_defn_xi_1437}, in this case, we also have that $E_\ba\ast\xi(\bx_\bi)=0$. Suppose now that $\bx_\bi\prec\ba$. Then,~\eqref{eqn_1438_0510} yields
\begin{align*}
E_\ba\ast q_{/\pi_\bi}
&=
\sum_{\substack{\bb\in R^\A\\ \bb_\bi=\ba\\ \bx\prec\bb}}\lambda_{R,\bx}(f_{\bx,\bb})
=
\sum_{\substack{\bb\in {A}^r\\ \bb_\bi=\ba\\ \bx\prec\bb}}\lambda_{R,\bx}(f_{\bx,\bb})
=
\sum_{\substack{g:\{\bx\}\to A\\ g|_{\{\bx_\bi\}}=f_{\bx_\bi,\ba}}}\lambda_{R,\bx}(g)
=
\lambda_{\{\bx_\bi\}}(f_{\bx_\bi,\ba})
=
E_\ba\ast\xi(\bx_\bi),
\end{align*}
where the second and fourth equalities come from (SA4) and (SA3), respectively. We conclude that $E_\ba\ast q_{/\pi_\bi}=E_\ba\ast\xi(\bx_\bi)$ in all cases, so that $q_{/\pi_\bi}=\xi(\bx_\bi)$, as claimed. It follows that $\xi(\bx^\tensor{k})\in R^{\mathbb{F}_{\Qconv}(\A^{\tensor{k}})}$, and hence $\xi:\X^\tensor{k}\to\mathbb{F}_{\Qconv}(\A^{\tensor{k}})$ is a homomorphism.
\end{proof}

\begin{rem}
\label{rem_local_consistency_1}
The characterisation of Sherali-Adams acceptance expressed in Theorem~\ref{thm:main} can be straightforwardly adapted to yield a characterisation of acceptance for the (less powerful) $k$-consistency algorithm. Following~\cite{BBKO21}, the $k$-consistency algorithm applied to two $\sigma$-structures $\X$ and $\A$ accepts if and only if there exists a nonempty collection $\mathcal{F}$ of partial homomorphisms from $\X$ to $\A$ with at most $k$-element domains such that $(i)$ $\mathcal{F}$ is closed under restrictions, i.e., for every $f\in \mathcal{F}$ and every $V\subseteq\dom(f)$, $f|_{V}\in \mathcal{F}$, and $(ii)$ $\mathcal{F}$ has the extension property up to $k$, i.e., for every $f\in \mathcal{F}$ and every $V\subseteq X$ with $|V|\leq k$ and $\dom(f)\subseteq V$, there exists $g\in \mathcal{F}$ such that $g$ extends $f$ and $\dom(g)=V$. The role of the minion $\Qconv$ is now taken by the minion $\mathscr{H}=\Pol(\operatorname{HORN-3-SAT})$, whose $L$-ary elements are all nonempty subsets of $[L]$ (cf.~\cite{BBKO21}) and whose free structure is the so-called power structure introduced in~\cite{Feder98:monotone} (see also~\cite{ChenDG13}). 
The minion $\mathscr{H}$ is closely related to $\Qconv$: For each $L\in\N$, we can view $\mathscr{H}^{(L)}$ as the set $\{\supp(q):q\in\Qconv^{(L)}\}$; given $h=\supp(q)\in\mathscr{H}^{(L)}$ and $\pi:[L]\to[L']$, $h_{/\pi}=\supp(q_{/\pi})$. It follows that the results in Section~\ref{sec_tensorisation} can be suitably modified to fit the minion $\mathscr{H}$. As a result, one easily adapts the proof of Proposition~\ref{prop_first_implication_homo_SAk} to show that, if $\xi:\X^\tensor{k}\to\mathbb{F}_{\mathscr{H}}(\A^\tensor{k})$ is a homomorphism,
the collection $\mathcal{F}=\{f_{\bx,\ba}:\bx\in X^k,\ba\in\supp(\xi(\bx))\}$ witnesses that the $k$-consistency algorithm accepts. Conversely, again using $\mathscr{H}$-analogues of the results in Section~\ref{sec_tensorisation}, the proof of Proposition~\ref{prop_SAkaccepts_hom_tensor} can be adapted to show that, given a collection $\mathcal{F}$ of partial homomorphisms witnessing that the $k$-consistency algorithm accepts, the map $\xi:X^k\to\mathcal{T}^{k;n}(\{0,1\})$ defined by setting, for $\bx\in X^k$ and $\ba\in A^k$, $E_\ba\ast\xi(\bx)=1$ if $\bx\prec\ba$ and $f_{\bx,\ba}\in \mathcal{F}$, $E_\ba\ast\xi(\bx)=0$ otherwise yields a homomorphism from $\X^\tensor{k}$ to $\mathbb{F}_{\mathscr{H}}(\A^\tensor{k})$. 
We obtain the following fact: 

\vspace{.1cm}
\noindent\emph{Let $\X,\A$ be two $k$-enhanced $\sigma$-structures and let $k\in\N$, $k\geq \max(2,\max_{R\in\sigma}\ar(R))$. 
Then, the $k$-consistency algorithm applied to $\X$ and $\A$ accepts if and only if $\X^\tensor{k}\to\mathbb{F}_{\mathscr{H}}(\A^\tensor{k})$.}\footnote{Unlike in Theorem~\ref{thm:main}, here we require that $k$ is at least the maximum arity. This is an unsubstantial technicality due to this specific definition of $k$-consistency. Essentially, the reason is that any constraint whose scope has more than $k$ distinct variables does not appear among the constraints of the partial homomorphisms witnessing $k$-consistency, while it does appear in the requirement (SA3) of Sherali-Adams.}
\end{rem}

The results presented in the remaining part of this section show how the multilinear nature of $\mathbb{F}_{\Qconv}(\A^\tensor{k})$ helps describing the functioning of the Sherali-Adams hierarchy.  
\begin{prop}
\label{prop_f_q_a_k}
Let $k\in\N$ and let $\A$ be a $\sigma$-structure. Then $\mathbb{F}_{\Qconv}(\A^\tensor{k})\to[\mathbb{F}_{\Qconv}(\A)]^\tensor{k}$.
\end{prop}
\begin{proof}
If $k=1$ the result is trivial (cf.~Footnote~\ref{footnote_RA_RAk}), so we assume $k\geq 2$. Let $J\in\cT^{k-1;n}(\Q)$ be the all-one tensor, and consider, for each $j\in [k]$, some tuple $\bell^{(j)}\in [k]^k$ whose $k$-th entry is $j$.
Recalling the definition in~\eqref{defn_tensor_Pi_i}, observe that the map
\begin{align*}
\psi:\Qconv^{(n^k)}&\to [\Qconv^{(n)}]^k\\
T&\mapsto (J\ast\Pi_{\bell^{(1)}}\ast T,\dots,J\ast\Pi_{\bell^{(k)}}\ast T)
\end{align*}
is well defined since, for each $j\in [k]$, 
\begin{align*}
\sum_{a\in {A}}\be_a\ast(J\ast\Pi_{\bell^{(j)}}\ast T)&=\sum_{\ba\in {A}^k}E_\ba\ast \Pi_{\bell^{(j)}}\ast T
=
\sum_{\ba,\ba'\in {A}^k}(E_\ba\ast \Pi_{\bell^{(j)}}\ast E_{\ba'})(E_{\ba'}\ast T)\\
&=\sum_{\ba'\in {A}^k}E_{\ba'}\ast T=1.
\end{align*}
We claim that $\psi$ is a homomorphism from $\mathbb{F}_{\Qconv}(\A^\tensor{k})$ to $[\mathbb{F}_{\Qconv}(\A)]^\tensor{k}$. To prove the claim, consider a relation symbol $R\in\sigma$ of arity $r$, and take $M=[M_\bi]_{\bi\in [r]^k}\in R^{\mathbb{F}_{\Qconv}(\A^\tensor{k})}$. By the definition of free structure, we know that there exists a tuple $q\in\Qconv^{(|R^\A|)}$ such that $M_\bi=q_{/\pi_\bi}=P_\bi\ast q$ for each $\bi\in [r]^k$, where $P_\bi$ is the tensor defined by~\eqref{defn_tensor_Pi}. Consider now the tuple $\bv=(P_1\ast q,\dots,P_r\ast q)$, where, for $j\in [r]$, $P_j$ is the $|A|\times |R^\A|$ matrix defined by 
\begin{align*}
\be_a\ast P_j\ast \be_\bb=
\left\{
\begin{array}{ll}
1&\mbox{if }b_j=a\\
0&\mbox{otherwise}
\end{array}
\right.
\hspace{2cm}
\forall a\in {A},\forall\bb=(b_1,\dots,b_r)\in R^\A.
\end{align*}
Again by the definition of free structure (cf.~Example~\ref{ex_free_structure}), we see that $\bv\in R^{\mathbb{F}_{\Qconv}(\A)}$; hence, $\bv^\tensor{k}\in R^{[\mathbb{F}_{\Qconv}(\A)]^\tensor{k}}$. The claim is proved if we show that $\psi(M)=\bv^\tensor{k}$ -- or, equivalently, that $\psi(M_\bi)=\bv_\bi$ for each $\bi=(i_1,\dots,i_k)\in [r]^k$. First, notice that 
\begin{align*}
\psi(M_\bi)&=\psi(P_\bi\ast q)=
(J\ast\Pi_{\bell^{(1)}}\ast (P_\bi\ast q),\dots,J\ast\Pi_{\bell^{(k)}}\ast (P_\bi\ast q))\\
&=
((J\ast\Pi_{\bell^{(1)}}\cont{k} P_\bi)\ast q,\dots,(J\ast\Pi_{\bell^{(k)}}\cont{k} P_\bi)\ast q),
\\
\bv_\bi&=(P_{i_1}\ast q,\dots, P_{i_k}\ast q).
\end{align*}
Therefore, to prove the claim, it suffices to argue that $J\ast\Pi_{\bell^{(j)}}\cont{k} P_\bi = P_{i_j}$ for each $j\in [k]$. For $a\in {A}$ and $\bb=(b_1,\dots,b_r)\in R^\A$, we find 
\begin{align*}
\be_a\ast(J\ast\Pi_{\bell^{(j)}}\cont{k} P_\bi)\ast \be_\bb
&=
\sum_{\substack{\ba\in {A}^k \\ a_k=a}}E_\ba\ast \Pi_{\bell^{(j)}}\ast (P_\bi\ast \be_\bb)
=
\sum_{\substack{\ba\in {A}^k \\ a_k=a}}\sum_{\ba'\in {A}^k}(E_\ba\ast \Pi_{\bell^{(j)}}\ast E_{\ba'})(E_{\ba'}\ast P_\bi\ast \be_\bb)\\
&=
\sum_{\substack{\ba\in {A}^k \\ a_k=a}}(E_\ba\ast \Pi_{\bell^{(j)}}\ast E_{\bb_\bi})
=
\left\{
\begin{array}{ll}
1&\mbox{if }b_{i_j}=a\\
0&\mbox{otherwise}
\end{array}
\right.
=
\be_a\ast P_{i_j} \ast \be_\bb,
\end{align*}
where the second-to-last equality is due to the fact that the $k$-th entry of the tuple $\bell^{(j)}$ is $j$. Hence, the claim is true. 
\end{proof}

Combining Proposition~\ref{prop_f_q_a_k} with Proposition~\ref{prop_SAkaccepts_hom_tensor}, one finds that, if $\SA^k(\X,\A)$ accepts, then 
\begin{align*}
\X^\tensor{k}\to \mathbb{F}_{\Qconv}(\A^\tensor{k})\to[\mathbb{F}_{\Qconv}(\A)]^\tensor{k}.
\end{align*}
By Lemma~\ref{prop_correspondence_morphisms_tensor_structures}, this yields $\X\to\mathbb{F}_{\Qconv}(\A)$, which implies that $\BLP(\X,\A)$ accepts (cf.~Appendix~\ref{app:BLP}). Hence, Proposition~\ref{prop_f_q_a_k} reflects the well-known fact that any level of Sherali-Adams is at least as powerful as $\BLP$ (see also Remark~\ref{remark_rank_one_segre}). In a similar way, the next proposition is the geometric counterpart of the algorithmic fact that higher levels of Sherali-Adams are tighter than lower levels.  

\begin{prop}
\label{prop_sherali_adams_smaller_level}
Let $k\geq p\in\N$, let $\X,\A$ be $\sigma$-structures, and suppose that $\X^\tensor{k}\to \mathbb{F}_{\Qconv}(\A^\tensor{k})$. Then $\X^\tensor{p}\to \mathbb{F}_{\Qconv}(\A^\tensor{p})$. 
\end{prop}
\begin{proof}
Fix a homomorphism $f:\X^\tensor{k}\to\mathbb{F}_{\Qconv}(\A^\tensor{k})$, and consider the function $g:X^p\to\cT^{p;n}(\Q)$ defined by $(x_1,\dots,x_p)\mapsto f((x_1,\dots,x_p,x_1,x_1,\dots,x_1))\ast J$, where $J\in\cT^{k-p;n}(\Q)$ is the all-one tensor. Observe that
\begin{align*}
\sum_{\ba\in {A}^p}E_\ba\ast g((x_1,\dots,x_p))
  &=\sum_{\ba\in {A}^p}E_\ba\ast f((x_1,\dots,x_p,x_1,\dots,x_1))\ast J\\
  &=
\sum_{\ba\in {A}^p}\sum_{\bb\in {A}^{k-p}}E_\ba\ast f((x_1,\dots,x_p,x_1,\dots,x_1))\ast E_\bb\\
&=
\sum_{\bc\in {A}^k}E_\bc\ast f((x_1,\dots,x_p,x_1,\dots,x_1))=1.
\end{align*}
We now show that $g$ yields a homomorphism from $\X^\tensor{p}$ to $\mathbb{F}_{\Qconv}(\A^\tensor{p})$. Take $R\in\sigma$ of arity $r$, and let $\bx=(x_1,\dots,x_r)\in R^\X$, so $\bx^\tensor{p}\in R^{\X^\tensor{p}}$. We need to show that $g(\bx^\tensor{p})\in R^{\mathbb{F}_{\Qconv}(\A^\tensor{p})}$. Observe that $\bx^\tensor{k}\in R^{\X^\tensor{k}}$; since $f$ is a homomorphism, we have $ f(\bx^\tensor{k})\in R^{\mathbb{F}_{\Qconv}(\A^\tensor{k})}$. Therefore, there exists $q\in\Qconv^{(|R^\A|)}$ such that $f(\bx_\bi)=q_{/\pi_\bi}$ for each $\bi\in [r]^k$. Take $\bell=(\ell_1,\dots,\ell_p)\in [r]^p$. We claim that $g(\bx_{\bell})=q_{/\pi_{\bell}}$. Consider $\bj\coloneqq(\ell_1,\dots,\ell_p,\ell_1,\ell_1,\dots,\ell_1)\in [r]^k$. For $\ba\in {A}^p$, we have 
\begin{align*}
E_\ba\ast g(\bx_{\bell})
&=
E_\ba\ast f((x_{\ell_1},\dots,x_{\ell_p},x_{\ell_1},\dots,x_{\ell_1}))\ast J
=
E_\ba\ast f(\bx_\bj)\ast J\\
&=
E_\ba\ast q_{/\pi_\bj}\ast J
=
\sum_{\bb\in {A}^{k-p}} E_\ba\ast q_{/\pi_\bj}\ast E_\bb
=
\sum_{\bb\in {A}^{k-p}} E_{(\ba,\bb)}\ast q_{/\pi_\bj}\\
&=
\sum_{\bb\in {A}^{k-p}} E_{(\ba,\bb)}\ast P_\bj\ast q
=
\sum_{\bb\in {A}^{k-p}}\sum_{\ba'\in R^\A} (E_{(\ba,\bb)}\ast P_\bj\ast \be_{\ba'})(\be_{\ba'}\ast q)\\
&=
\sum_{\ba'\in R^\A}(\be_{\ba'}\ast q)\sum_{\bb\in {A}^{k-p}}(E_{(\ba,\bb)}\ast P_\bj\ast \be_{\ba'})
=
\sum_{\substack{\ba'\in R^\A\\\ba'_{\bell}=\ba}}\be_{\ba'}\ast q\\
&=
\sum_{\ba'\in R^\A}(E_\ba\ast P_{\bell}\ast \be_{\ba'})(\be_{\ba'}\ast q)
=
E_\ba\ast P_{\bell}\ast q
=
E_\ba\ast q_{/\pi_{\bell}}.
\end{align*}
This shows that the claim is true, thus concluding the proof of the proposition.
\end{proof}

The following result will be used in Section~\ref{sec:approx} (specifically, in the proof of Proposition~\ref{prop_reduction_line_digraph}). It reflects the fact that Sherali-Adams solutions only give a nonzero weight to those assignments that yield partial homomorphisms from $\X$ to $\A$.

\begin{prop}
\label{lem_partial_homo}
Let $k\in\N$ with $k\geq 2$, let $\X,\A$ be two $k$-enhanced $\sigma$-structures, and let $\xi:\X^{\tensor{k}}\to\mathbb{F}_{\Qconv}(\A^{\tensor{k}})$ be a homomorphism. Let $R\in\sigma$ have arity $r$, and take $\bx\in X^k$, $\ba\in A^k$, and $\bi\in [k]^r$. If $\bx_\bi\in R^\X$ and $\ba_\bi\not\in R^\A$, then $E_\ba\ast\xi(\bx)=0$.
\end{prop}
\begin{proof}
From $\bx_\bi\in R^\X$ we have $\bx_\bi^\tensor{k}\in R^{\X^\tensor{k}}$ and, thus, $\xi(\bx_\bi^\tensor{k})\in R^{\mathbb{F}_{\Qconv}(\A^{\tensor{k}})}$. It follows that there exists $q\in\Qconv^{(|R^\A|)}$ such that $\xi(\bx_{\bi_\bj})=q_{/\pi_\bj}$ for each $\bj\in [r]^k$. Proposition~\ref{lem_a_symmetry} then yields
\begin{align}
\label{eqn_1038_2702}
\Pi_{\bi_\bj}\ast\xi(\bx)
&=
\xi(\bx_{\bi_\bj})
=
q_{/\pi_{\bj}}
=
P_\bj\ast q.
\end{align}
Consider, for each $\alpha\in [k]$, the set $S_\alpha=\{\beta\in [r]:i_\beta=\alpha\}$, and fix an element $\hat{\beta}\in [r]$. The tuple $\bj\in [r]^k$ defined by setting $j_\alpha=\min S_\alpha$ if $S_\alpha\neq\emptyset$, $j_\alpha=\hat{\beta}$ otherwise satisfies $\bi_{\bj_\bi}=\bi$. We obtain
\begin{align}
\label{eqn_2903_1640}
E_{\ba_{\bi_\bj}}\ast P_\bj\ast q
&=
\sum_{\substack{\bb\in R^\A\\\bb_\bj=\ba_{\bi_\bj}}}\be_\bb\ast q
=
\sum_{\substack{\bb\in R^\A\\ \bx_\bi\prec\bb \\\bb_\bj=\ba_{\bi_\bj}}}\be_\bb\ast q,
\end{align}
where the first equality comes from Lemma~\ref{lem_multiplication_rule_P} and the second from Lemma~\ref{lem_ea_ast_Q}. We claim that the sum on the right-hand side of~\eqref{eqn_2903_1640} is zero. Indeed, let $\bb\in A^r$ satisfy $\bx_\bi\prec\bb$ and $\bb_\bj=\ba_{\bi_\bj}$. For any $\alpha\in [r]$ we have $i_\alpha=i_{j_{i_\alpha}}$, so $x_{i_\alpha}=x_{i_{j_{i_\alpha}}}$ and, hence, $b_\alpha=b_{j_{i_\alpha}}$. It follows that $\bb=\bb_{\bj_\bi}=\ba_{\bi_{\bj_\bi}}=\ba_\bi\not\in R^\A$, which proves the claim. Combining this with~\eqref{eqn_1038_2702},~\eqref{eqn_2903_1640}, and Lemma~\ref{lem_multiplication_rule}, we find
\begin{align*}
0
&=
E_{\ba_{\bi_\bj}}\ast P_\bj\ast q
=
E_{\ba_{\bi_\bj}}\ast
\Pi_{\bi_\bj}\ast\xi(\bx)
=
\sum_{\substack{\bb\in A^k\\ \bb_{\bi_\bj}=\ba_{\bi_\bj}}}E_\bb\ast\xi(\bx)
\geq 
E_\ba\ast\xi(\bx),
\end{align*}
which concludes the proof.
\end{proof}

The next result -- in combination with Theorem~\ref{thm:main} -- yields a geometric explanation for the known fact that, for every $\PCSP$ template $(\A,\B)$, the $k$-th level of Sherali-Adams correctly classifies any instance $\X$ of $\PCSP(\A,\B)$ such that the domain of $\X$ contains $k$ (or fewer) elements.

\begin{prop}
\label{prop_sherali_adams_exact}
Let $k\in\N$ with $k\geq 2$, let $\X,\A$ be two $k$-enhanced $\sigma$-structures such that $X=[k]$, and let $g:\X^\tensor{k}\to\mathbb{F}_{\Qconv}(\A^\tensor{k})$ be a homomorphism. Given $\ba=(a_1,\dots,a_k)\in\supp(g((1,\dots,k)))$, the function $f:X\to A$ defined by $j\mapsto a_j$ yields a homomorphism from $\X$ to $\A$.
\end{prop}
\begin{proof}
Let $R\in\sigma$ be a relation symbol of arity $r$, and take a tuple $\bx=(x_1,\dots,x_r)\in R^\X$. Notice that $f(\bx)=(a_{x_1},\dots,a_{x_r})=\ba_\bx$. For the sake of contradiction, suppose that $\ba_\bx\not\in R^\A$. Since $\bx\in R^\X$, we have $\bx^\tensor{k}\in R^{\X^\tensor{k}}$; being $g$ a homomorphism, this implies that $g(\bx^\tensor{k})\in R^{\mathbb{F}_{\Qconv}(\A^\tensor{k})}$. Therefore, $\exists q\in\Qconv^{(|R^\A|)}$ such that $g(\bx_\bi)=q_{/\pi_\bi}$ for each $\bi\in [r]^k$. Set $\by\coloneqq\bx_\bi\in [k]^k$, and observe that
\begin{align*}
E_{\ba_\by}\ast g(\by)
&=
E_{\ba_\by}\ast g(\bx_\bi)
=
E_{\ba_\by}\ast q_{/\pi_\bi}
=
E_{\ba_\by}\ast P_\bi\ast q
=
\sum_{\substack{\bc\in R^\A \\ \bc_\bi=\ba_\by}}\be_\bc\ast q
=
\sum_{\substack{\bc\in R^\A \\ \bc_\bi=\ba_\by \\ \bx\prec\bc}}\be_\bc\ast q
,
\end{align*}
where the fourth and fifth equalities comes from Lemma~\ref{lem_multiplication_rule_P} and Lemma~\ref{lem_ea_ast_Q}, respectively. On the other hand, Proposition~\ref{lem_a_symmetry} and Lemma~\ref{lem_multiplication_rule} yield
\begin{align*}
E_{\ba_\by}\ast g(\by)
&=
E_{\ba_\by}\ast g((1,\dots,k)_\by)
=
E_{\ba_\by}\ast \Pi_\by\ast g((1,\dots,k))\\
&=
\sum_{\substack{\bb\in {A}^k \\ \bb_\by=\ba_\by}}E_\bb\ast g((1,\dots,k))
\geq E_\ba\ast g((1,\dots,k))>0.
\end{align*}
We deduce that, for each $\bi\in [r]^k$, there exists some $\bc=(c_1,\dots,c_r)\in R^\A$ such that $\bc_\bi=\ba_{\bx_\bi}$ and $\bx\prec\bc$. Choose $\bi$ so that $\{\bx\}=\{\bx_\bi\}$. Since $\bc\in R^\A$ and $\ba_\bx\not\in R^\A$, we have $\bc\neq\ba_\bx$; so, $c_p\neq a_{x_p}$ for some $p\in [r]$. From $x_p\in\{\bx\}=\{\bx_\bi\}$, we obtain $x_p=x_{i_t}$ for some $t\in [k]$. Since $\bx\prec\bc$, this yields $c_p=c_{i_t}$. Therefore, $c_{i_t}=c_p\neq a_{x_p}=a_{x_{i_t}}$, so $\bc_\bi\neq \ba_{\bx_\bi}$, a contradiction. We conclude that $f(\bx)=\ba_\bx\in R^\A$, so that $f$ is a homomorphism.
\end{proof}
\begin{rem}
\label{rem_prop_false_if_no_k_enhanced}
Proposition~\ref{prop_sherali_adams_exact} is not true if the requirement that $\X$ and $\A$ be $k$-enhanced is dropped. For example, let $k=3$, and suppose that $\X$ is a directed $3$-cycle on domain $[3]$ and $\A$ is $\K_2$. Consider the function $g:[3]^3\to \cT^{3;2}(\mathbb{Q})$ defined by
\[
\begin{array}{lllllll}
&g((1,1,1))=g((2,2,2))=g((3,3,3))&=
\left[\begin{array}{@{}cc|cc@{}}
\frac{1}{2}&0&0&0\\
0&0&0&\frac{1}{2}
\end{array}\right]\\[10pt]
&g((1,1,2))=g((2,2,1))=g((2,2,3))=g((3,3,2))=g((3,3,1))
=g((1,1,3))&=
\left[\begin{array}{@{}cc|cc@{}}
0&\frac{1}{2}&0&0\\
0&0&\frac{1}{2}&0
\end{array}\right]\\[10pt]
&g((1,2,2))=g((2,1,1))=g((2,3,3))=g((3,2,2))=g((3,1,1))
=g((1,3,3))&=
\left[\begin{array}{@{}cc|cc@{}}
0&0&\frac{1}{2}&0\\
0&\frac{1}{2}&0&0
\end{array}\right]\\[10pt]
&g((1,2,1))=g((2,1,2))=g((2,3,2))=g((3,2,3))=g((3,1,3))=
g((1,3,1))&=
\left[\begin{array}{@{}cc|cc@{}}
0&0&0&\frac{1}{2}\\
\frac{1}{2}&0&0&0
\end{array}\right]\\[10pt]
&g((1,2,3))=g((1,3,2))=g((2,1,3))=g((2,3,1))=g((3,1,2))
=g((3,2,1))&=
\left[\begin{array}{@{}cc|cc@{}}
\frac{1}{8}&\frac{1}{8}&\frac{1}{8}&\frac{1}{8}\\
\frac{1}{8}&\frac{1}{8}&\frac{1}{8}&\frac{1}{8}
\end{array}\right].
\end{array}
\]
One can check that $g$ yields a homomorphism from $\X^\tensor{3}$ to $\mathbb{F}_{\Qconv}(\A^\tensor{3})$. Observe that \linebreak$\supp(g((1,2,3)))=[2]^3$, so Proposition~\ref{prop_sherali_adams_exact} would imply that any map $[3]\to [2]$ yields a homomorphism from $\X$ to $\A$ -- which is certainly false, as $\X\not\to\A$.

Proposition~\ref{prop_sherali_adams_exact} is also not true if the requirement $k\geq 2$ is dropped. For example, consider the $1$-enhanced structures $\X=([1];R_1^\X=[1],R^\X=\{(1,1)\})$ and $\A=([2];R_1^\A=[2],R^\A=\{(1,2),(2,1)\})$. The assignment $1\mapsto (0.5,0.5)$ yields a homomorphism $g:\X\to \mathbb{F}_{\Qconv}(\A)$ with $\supp(g(1))=[2]$. However, $\X\not\to\A$.  
\end{rem}

We now present an interesting application of the tensorisation machinery to the problem of counting homomorphisms: The number of homomorphisms between two $\sigma$-structures $\X,\A$ can be read off on the image of a specific tuple under a maximum-support homomorphism from $\X^\tensor{k}$ to $\mathbb{F}_{\Qconv}(\A^\tensor{k})$ for suitable $k$.

Given two homomorphisms $g_A,g_B:\X\to\mathbb{F}_{\Qconv}(\A)$, one easily checks that the assignment $x\mapsto \frac{1}{2}(g_A(x)+g_B(x))$ yields a new homomorphism $g_C:\X\to\mathbb{F}_{\Qconv}(\A)$. As a consequence, if $\X\to\mathbb{F}_{\Qconv}(\A)$, there exists some homomorphism $g:\X\to\mathbb{F}_{\Qconv}(\A)$ having \emph{maximum support} -- i.e., such that $\supp(g(x))\supseteq\supp(g'(x))$ for any $g':\X\to\mathbb{F}_{\Qconv}(\A)$ and for any $x\in X$. Such homomorphism can be found in polynomial time (in the size of $\X$) by searching a relative-interior point in the polytope described by the constraints of $\BLP(\X,\A)$~\cite{schrijver1998theory}.
\begin{prop}
\label{cor_number_homomorphisms_from_tensors}
Let $k\in\N$ with $k\geq 2$, and let $\X,\A$ be two $k$-enhanced $\sigma$-structures such that $X=[k]$. Then
\begin{itemize}
\item
$|\Hom(\X,\A)|=0$\; if \; $\X^\tensor{k}\not\to\mathbb{F}_{\Qconv}(\A^\tensor{k})$;
\vspace{.2cm}
\item
$|\Hom(\X,\A)|=|\supp(g((1,\dots,k)))|$
   \; if \; $g:\X^\tensor{k}\to\mathbb{F}_{\Qconv}(\A^\tensor{k})$ has maximum support.
\end{itemize}
\end{prop}
\begin{proof}
If $\X^\tensor{k}\not\to\mathbb{F}_{\Qconv}(\A^\tensor{k})$, by Proposition~\ref{prop_SAkaccepts_hom_tensor} $\SA^k(\X,\A)$ does not accept. Hence, $\X\not\to\A$, so that $\Hom(\X,\A)=\emptyset$. 

Suppose now that $\X^\tensor{k}\to\mathbb{F}_{\Qconv}(\A^\tensor{k})$ and $g$ is a homomorphism having maximum support. As a consequence of Proposition~\ref{prop_sherali_adams_exact}, we can build the map
\begin{align*}
m:\supp(g((1,\dots,k)))&\to \Hom(\X,\A)\\
\ba=(a_1,\dots,a_k)&\mapsto f_\ba
\end{align*}
where $f_\ba$ is defined by $j\mapsto a_j$. Clearly, $m$ is injective. Given a homomorphism $f:\X\to\A$, consider the map $g_f:X^k\to \cT^{k;n}(\Q)$ defined by $\bx\mapsto E_{f(\bx)}$. We claim that $g_f$ is a homomorphism from $\X^\tensor{k}$ to $\mathbb{F}_{\Qconv}(\A^\tensor{k})$. Indeed, take $R\in\sigma$ of arity $r$ and take $\bx\in R^\X$, so that $\bx^\tensor{k}\in R^{\X^\tensor{k}}$. Since $f$ is a homomorphism, $f(\bx)\in R^\A$; consider then $q\coloneqq\be_{f(\bx)}\in \Qconv^{(|R^\A|)}$. For any $\bi\in [r]^k$, $\ba'\in {A}^k$, observe that
\begin{align*}
E_{\ba'}\ast q_{/\pi_\bi}
&=
E_{\ba'}\ast P_\bi\ast \be_{f(\bx)}
=
\left\{
\begin{array}{ll}
1 & \mbox{if }f(\bx_\bi)=\ba'\\
0 & \mbox{otherwise}
\end{array}
\right.
=
E_{\ba'}\ast E_{f(\bx_\bi)}
=
E_{\ba'}\ast g_f(\bx_\bi),
\end{align*}
so that $g_f(\bx_\bi)=q_{/\pi_\bi}$ $\forall \bi\in [r]^k$, which means that $g_f(\bx^\tensor{k})\in R^{\mathbb{F}_{\Qconv}(\A^\tensor{k})}$. Therefore, the claim is true.

Let $\ba\coloneqq (f(1),\dots,f(k))\in {A}^k$, and observe that $\ba\in\supp(g_f((1,\dots,k)))\subseteq \supp(g((1,\dots,k)))$. For any $j\in X$, we have $[m(\ba)](j)=f_\ba(j)=f(j)$, so $f=m(\ba)$. We deduce that $m$ is surjective, too, whence the result follows.
\end{proof}

The next proposition establishes that symmetries of $\A$ can be lifted to symmetries of the solutions of $\SA^k(\X,\A)$. As a result, one can always find a solution of $\SA^k(\X,\A)$ that is, in some sense, at least as symmetric as $\A$. We let $\Aut(\A)$ denote the group of automorphisms of $\A$ (where an \emph{automorphism} is a bijective homomorphism from $\A$ to itself whose inverse is also a homomorphism).
\begin{prop}
\label{prop_SA_and_automorphisms}
Let $k\in\N$, let $\X,\A$ be $\sigma$-structures, and suppose that $\X^\tensor{k}\to\mathbb{F}_{\Qconv}(\A^\tensor{k})$. Then there exists a homomorphism $\bar{g}:\X^\tensor{k}\to\mathbb{F}_{\Qconv}(\A^\tensor{k})$ having the following property: $E_\ba\ast \bar{g}(\bx)=E_{\vartheta(\ba)}\ast \bar{g}(\bx)$ for each $\bx\in X^k, \ba\in A^k, \vartheta\in \Aut(\A)$.
\end{prop}
\begin{proof}
Choose a homomorphism $g:\X^\tensor{k}\to\mathbb{F}_{\Qconv}(\A^\tensor{k})$, and let $\tau$ be an automorphism of $\A$. Consider the map $g^\tau:X^k\to \cT^{k;n}(\Q)$ defined by $E_\ba\ast g^\tau(\bx)\coloneqq E_{\tau^{-1}(\ba)}\ast g(\bx)$ for each $\bx\in X^k,\ba\in A^k$. Observe that 
\begin{align*}
\sum_{\ba\in A^k}E_\ba\ast g^\tau(\bx)=\sum_{\ba\in A^k}
E_{\tau^{-1}(\ba)}\ast g(\bx)=
\sum_{\ba\in A^k}
E_\ba\ast g(\bx)
=
1
\end{align*}   
for any $\bx\in X^k$, so $g^\tau(\bx)\in\Qconv^{(n^k)}$. We claim that $g^\tau$ yields a homomorphism from $\X^\tensor{k}$ to $\mathbb{F}_{\Qconv}(\A^\tensor{k})$. Let $R\in\sigma$ be a relation symbol of arity $r$, and let $\bx\in R^\X$. Since $g$ is a homomorphism, $g(\bx^\tensor{k})\in R^{\mathbb{F}_{\Qconv}(\A^\tensor{k})}$, so there exists some $q\in\Qconv^{(|R^\A|)}$ such that $g(\bx_\bi)=q_{/\pi_\bi}$ for each $\bi\in [r]^k$. Using that $\tau$ is an automorphism of $\A$, we can define the tuple $q^\tau\in \Qconv^{(|R^\A|)}$ by setting $\be_\bb\ast q^\tau\coloneqq \be_{\tau^{-1}(\bb)}\ast q$ for each $\bb\in R^\A$. We claim that $g^\tau(\bx_\bi)=q^\tau_{/\pi_\bi}$ for each $\bi\in [r]^k$; this would imply that $g^\tau(\bx^\tensor{k})\in R^{\mathbb{F}_{\Qconv}(\A^\tensor{k})}$. Given $\ba\in A^k$ and using Lemma~\ref{lem_multiplication_rule_P}, we find
\begin{align*}
\begin{array}{lllllllllllllll}
&\displaystyle E_\ba\ast q^\tau_{/\pi_\bi}
&=&
\displaystyle E_\ba\ast P_\bi\ast q^\tau
&=&
\displaystyle \sum_{\substack{\bb\in R^\A\\\bb_\bi=\ba}}\be_\bb\ast q^\tau
&=&
\displaystyle \sum_{\substack{\bb\in R^\A\\\bb_\bi=\ba}}\be_{\tau^{-1}(\bb)}\ast q\\[30pt]
&&=&
\displaystyle \sum_{\substack{\bb\in R^\A\\ (\tau(\bb))_\bi=\ba}}\be_{\bb}\ast q
&=&
\displaystyle \sum_{\substack{\bb\in R^\A\\ \tau(\bb_\bi)=\ba}}\be_{\bb}\ast q
&=&
\displaystyle \sum_{\substack{\bb\in R^\A\\ \bb_\bi=\tau^{-1}(\ba)}}\be_{\bb}\ast q
\end{array} 
\intertext{and}
\begin{array}{lllllllllllllll}
&\displaystyle E_\ba\ast g^\tau(\bx_\bi)
&=&
\displaystyle E_{\tau^{-1}(\ba)}\ast g(\bx_\bi)
&=&
\displaystyle E_{\tau^{-1}(\ba)}\ast  q_{/\pi_\bi}
&=&
\displaystyle E_{\tau^{-1}(\ba)}\ast  P_\bi\ast q\\[10pt]
&&=&
\displaystyle \sum_{\substack{\bb\in R^\A\\ \bb_\bi=\tau^{-1}(\ba)}}\be_\bb\ast q,
\end{array}
\end{align*}
thus proving the claims. 

We have shown that $g^\tau$ is a homomorphism from $\X^\tensor{k}$ to $\mathbb{F}_{\Qconv}(\A^\tensor{k})$ for each $\tau\in\Aut(\A)$. Clearly, any convex combination of such homomorphisms is still a homomorphism (since LP solutions are closed under convex combinations). It follows that
\begin{align*}
\bar{g}\coloneqq \frac{1}{|\Aut(\A)|}\sum_{\tau\in\Aut(\A)}g^\tau
\end{align*}
is a homomorphism, too.
Given $\bx\in X^k$, $\ba\in A^k$, and $\vartheta\in \Aut(\A)$, we find
\begin{align*}
E_{\vartheta(\ba)}\ast \bar{g}(\bx)
&=
\frac{1}{|\Aut(\A)|}\sum_{\tau\in\Aut(\A)}E_{\vartheta(\ba)}\ast g^\tau(\bx)\\
&=
\frac{1}{|\Aut(\A)|}\sum_{\tau\in\Aut(\A)}E_{(\tau^{-1}\circ\vartheta)(\ba)}\ast g(\bx)\\
&=
\frac{1}{|\Aut(\A)|}\sum_{\tau\in\Aut(\A)}E_{\tau^{-1}(\ba)}\ast g(\bx)\\
&=
\frac{1}{|\Aut(\A)|}\sum_{\tau\in\Aut(\A)}E_{\ba}\ast g^\tau(\bx)
=
E_\ba\ast\bar{g}(\bx),
\end{align*}
which concludes the proof.
\end{proof}
We shall notice, in Section~\ref{sec:approx}, that the explicit homomorphism $\tilde \X^\tensor{k}\to \mathbb{F}_{\Qconv}(\tilde{\K}_k^\tensor{k})$ whose existence proves Proposition~\ref{prop_SA_cliques} is invariant under the automorphism group of the clique -- that is, under the symmetric group. In other words, that homomorphism could be indicated by $\bar{g}$ in the language of Proposition~\ref{prop_SA_and_automorphisms}.

\section{Sherali-Adams and graph colouring}\label{sec:approx}

In this section, we use the tensorisation machinery developed in Sections~\ref{sec_tensorisation} and~\ref{sec_characterisation_acceptance} to prove the following main result.
\begin{thm*}[Theorem~\ref{main_theorem_approximate_colouring} restated]
No constant level of Sherali-Adams solves the approximate graph colouring problem;
i.e., for any fixed $3\leq c\leq d$, there is no constant $k$ such that the $k$-th level of Sherali-Adams solves $\PCSP(\K_c,\K_d)$.
\end{thm*}
The key step towards the proof of Theorem~\ref{main_theorem_approximate_colouring} is Proposition~\ref{prop_SA_cliques} -- which, in turn, relies on the following combinatorial fact. Recall that, for two tuples $\textbf{s}$ and $\textbf{t}$ of equal length, the notation $\textbf{s}\sim\textbf{t}$ means that $\textbf{s}\prec\textbf{t}$ and $\textbf{t}\prec\textbf{s}$.

\begin{lem}
\label{lem_involved_1511_1718}
Let $c,k\in\N$, and consider the tuples $\ba,\bi\in [k]^k$, $\bx\in [c]^k$. Then
\begin{align*}
|\{\bb\in [k]^k: \bb_\bi=\ba\mbox{ and }\bb\sim\bx\}|=
\left\{
\begin{array}{cl}
\frac{(k-|\{\bx_\bi\}|)!}{(k-|\{\bx\}|)!} & \mbox{if } \ba\sim\bx_\bi\\
0 & \mbox{otherwise.}
\end{array}
\right.
\end{align*}
\end{lem}
\begin{proof}
Without loss of generality, we shall assume that $c\geq k$. Let $S_{\ba,\bi,\bx}$ denote the set $\{\bb\in [k]^k: \bb_\bi=\ba\mbox{ and }\bb\sim\bx\}$, and let $s_{\ba,\bi,\bx}$ denote the cardinality of $S_{\ba,\bi,\bx}$.

Suppose first that $\ba\not\sim \bx_\bi$. In this case, if $\bb\in S_{\ba,\bi,\bx}$, then $\bb\sim\bx$, which implies $\bb_\bi\sim\bx_\bi$. Since $\ba=\bb_\bi$, this yields $\ba\sim\bx_\bi$, a contradiction. Hence, $s_{\ba,\bi,\bx}=0$ as required.

Suppose now that $\ba\sim \bx_\bi$. We use induction on the number $k-|\{\bx\}|$. If $k-|\{\bx\}|=0$, the elements in the tuple $\bx$ are all distinct. Hence, $S_{\ba,\bi,\bx}$ is the set of tuples $\bb\in [k]^k$ such that all the elements of $\bb$ are distinct and $\bb_\bi=\ba$. This amounts to assigning distinct values from $[k]\setminus \{\ba\}$ to each $b_i$ with $i\in [k]\setminus \{\bi\}$. Observing that $|\{\bi\}|=|\{\bx_\bi\}|=|\{\ba\}|$, this results in $(k-|\{\bx_\bi\}|)!$ choices, so the result holds in this case. Suppose now, for the inductive step, that $k-|\{\bx\}|>0$. Take $v,w\in [k]$ such that $v\neq w$ and $x_v=x_w$. From $|\{\bx\}|<k\leq c$ we have that $\{\bx\}\subsetneqq [c]$, so we can find an element $p\in [c]\setminus \{\bx\}$. Consider the tuple ${\bx}'\in [c]^k$ defined by ${x}'_\alpha=x_\alpha$ if $\alpha\neq w$, ${x}'_w=p$. Consider also the tuple ${\bi}'\in [k]^k$ defined by ${i}'_\beta=i_\beta$ if $i_\beta\neq w$, ${i}'_\beta=v$ otherwise. We claim that $\bx'_{\bi'}=\bx_\bi$. Indeed, for any $\beta\in [k]$, if $i_\beta\neq w$ we have $x'_{i'_\beta}=x'_{i_\beta}=x_{i_\beta}$, while if $i_\beta=w$ we have $x'_{i'_\beta}=x'_{v}=x_v=x_w=x_{i_\beta}$. It follows that $\ba\sim \bx'_{\bi'}$. Consider now a tuple $\bb\in S_{\ba,\bi,\bx}$. For any $r\in [k]\setminus \{\bb\}$ we can consider a new tuple $\bb^{(r)}\in [k]^k$ defined by $b^{(r)}_\alpha=b_\alpha$ if $\alpha\neq w$, $b^{(r)}_w=r$. 
We claim that 
\begin{align}
\label{eq_sets_1611_1856}
S_{\ba,\bi',\bx'}=\{\bb^{(r)}:\bb\in S_{\ba,\bi,\bx} \mbox{ and } r\in [k]\setminus \{\bb\}\}.
\end{align}
To prove the \enquote{$\supseteq$} inclusion, observe first that $\bb^{(r)}\sim\bx'$ follows from the fact that $r\not\in \{\bb\}$ and $p\not\in \{\bx\}$. We now need to show that $\bb^{(r)}_{\bi'}=\ba$. Take $\beta\in [k]$. If $i_\beta\neq w$, $b^{(r)}_{i'_\beta}=b^{(r)}_{i_\beta}=b_{i_\beta}=a_\beta$. If $i_\beta=w$, $b^{(r)}_{i'_\beta}=b^{(r)}_v=b_v$. From $\bb\sim\bx$ and $x_v=x_w$, it follows $b_v=b_w$, whence $b^{(r)}_{i'_\beta}=b_w=b_{i_\beta}=a_\beta$. This concludes the proof that $\bb^{(r)}\in S_{\ba,\bi',\bx'}$. To prove the \enquote{$\subseteq$} inclusion, take $\bb'\in S_{\ba,\bi',\bx'}$, let $r=b'_w$, and consider the tuple $\bd\in [k]^k$ defined by $d_\alpha=b'_\alpha$ if $\alpha\neq w$, $d_w=b'_v$. We claim that $r\in [k]\setminus \{\bd\}$. Indeed, if $r\in \{\bd\}$, then $r=b'_t$ for some $t\in [k]$, $t\neq w$. But then $b'_t=b'_w$, and from $\bb'\sim\bx'$ we would get $x'_t=x'_w$, whence $x_t=x'_t=x'_w=p$; this is a contradiction, since $p\not\in\{\bx\}$. Next, we claim that $\bd\in S_{\ba,\bi,\bx}$. By the definition of $\bd$, one readily checks that $\bd\sim\bx$. To show that $\bd_\bi=\ba$, take $\beta\in [k]$. If $i_\beta\neq w$, $d_{i_\beta}=b'_{i_\beta}=b'_{i'_\beta}=a_\beta$, while, if $i_\beta=w$, $d_{i_\beta}=d_w=b'_v=b'_{i'_\beta}=a_\beta$. It is also clear from the definition of $\bd$ that $\bb'=\bd^{(r)}$.
This concludes the proof of~\eqref{eq_sets_1611_1856}.

If $\bb\in S_{\ba,\bi,\bx}$, then $|[k]\setminus\{\bb\}|=k-|\{\bb\}|=k-|\{\bx\}|$. Therefore, it follows from~\eqref{eq_sets_1611_1856} that $s_{\ba,\bi',\bx'}=s_{\ba,\bi,\bx}\cdot(k-|\{\bx\}|)$.
Observe that $|\{\bx'\}|=|\{\bx\}|+1$, so $k-|\{\bx'\}|=k-|\{\bx\}|-1$. Applying the inductive hypothesis, we find
\begin{align*}
s_{\ba,\bi,\bx}
&=
\frac{1}{k-|\{\bx\}|}\cdot s_{\ba,\bi',\bx'}
=
\frac{1}{k-|\{\bx\}|}\cdot \frac{(k-|\{\bx'_{\bi'}\}|)!}{(k-|\{\bx'\}|)!}
=
\frac{1}{k-|\{\bx\}|}\cdot \frac{(k-|\{\bx_{\bi}\}|)!}{(k-|\{\bx\}|-1)!}
=
\frac{(k-|\{\bx_{\bi}\}|)!}{(k-|\{\bx\}|)!}
\end{align*} 
as required.
\end{proof}
\begin{prop*}[Proposition~\ref{prop_SA_cliques} restated]
Let $\X$ be a loopless digraph and let $k\in\N$, $k\geq 2$. Then $\SA^k(\X,\K_k)$ accepts.
\end{prop*}

\begin{proof}
Let $\tilde{\X}$ (resp. $\tilde{\K}_k$) be obtained from $\X$ (resp. $\K_k$) by adding the relation $R_k^{\tilde{\X}}=X^k$ (resp. $R_k^{\tilde{\K}_k}=[k]^k$). Clearly, $\SA^k(\X,\K_k)$ accepts if and only if $\SA^k(\tilde \X,\tilde \K_k)$ accepts. By virtue of Theorem~\ref{thm:main}, then, it is sufficient to show that $\tilde \X^\tensor{k}\to \mathbb{F}_{\Qconv}(\tilde{\K}_k^\tensor{k})$.
For brevity, we denote $\mathbb{F}_{\Qconv}(\tilde{\K}_k^\tensor{k})$ by $\textbf{F}$. 

Consider the function
\begin{align*}
g:X^k&\to \mathcal{T}^{k;k}(\Q)\\
\bx&\mapsto \frac{(k-|\{\bx\}|)!}{k!}G_\bx
\end{align*}
where $G_\bx\in\mathcal{T}^{k;k}(\Q)$ is the tensor defined by 
\begin{align*}
E_\ba\ast G_\bx=
\left\{
\begin{array}{cl}
1 & \mbox{if }\ba\sim\bx\\
0 & \mbox{otherwise}
\end{array}
\right.
\hspace{1cm}
\forall \ba\in [k]^k.
\end{align*}
Notice that, for any $\bx\in X^k$,
\begin{align*}
&\sum_{\ba\in [k]^k}E_\ba \ast G_\bx=|\{\ba\in [k]^k:\ba\sim\bx\}|
=k(k-1)(k-2)\dots(k-|\{\bx\}|+1)
=\frac{k!}{(k-|\{\bx\}|)!}.
\end{align*}
It follows that $\sum_{\ba\in [k]^k} E_\ba\ast g(\bx)=1$, so $g(\bx)\in\Qconv^{(k^k)}$ $\forall \bx\in X^k$. To conclude the proof, we need to show that $g$ yields a homomorphism from $\tilde \X^\tensor{k}$ to $\textbf{F}$.

Let ${R}$ be the binary relation symbol corresponding to the edge sets of $\X$ and $\K_k$. Take an element of ${R}^{\tilde{\X}^\tensor{k}}$ and write it as $(x,y)^\tensor{k}$, where $(x,y)\in R^{\tilde{\X}}$ (so $x\neq y$, because $\X$ is loopless). We need to show that $g((x,y)^\tensor{k})\in {R}^{\textbf{F}}$. Observe that $|{R}^{\tilde{\K}_k}|=k^2-k$. Take $q\coloneqq\frac{1}{k^2-k}\bone_{k^2-k}\in\Qconv^{(|{R}^{\tilde{\K}_k}|)}$. We claim that $g((x,y)_\bi)=q_{/\pi_\bi}$ $\forall \bi\in [2]^k$; this would imply that $g((x,y)^\tensor{k})\in {R}^{\textbf{F}}$. For $\ba\in [k]^k$, observe that
\begin{align}
\label{eqn_1511_1543}
\notag
E_\ba\ast q_{/\pi_\bi} 
&=
E_\ba\ast P_\bi \ast q
=
\frac{1}{k^2-k}E_\ba\ast P_\bi \ast \bone_{k^2-k}\\
\notag
&=
\frac{1}{k^2-k}\sum_{(x',y')\in {R}^{\tilde{\K}_k}} E_\ba\ast P_\bi \ast \be_{(x',y')}\\
&=
\frac{1}{k^2-k}|\{(x',y')\in {R}^{\tilde{\K}_k}: (x',y')_\bi=\ba\}|,
\end{align}
where we have used~\eqref{defn_tensor_Pi}. 
Suppose now that $\bi=(1,\dots,1)$. In this case,~\eqref{eqn_1511_1543} yields
\begin{align*}
E_\ba\ast q_{/\pi_\bi}
&=
\frac{1}{k^2-k}|\{(x',y')\in {R}^{\tilde{\K}_k}: (x',\dots,x')=\ba\}|
=
\left\{
\begin{array}{cl}
\frac{1}{k} &\mbox{if }\ba \mbox{ is constant}\\
0 & \mbox{otherwise}.
\end{array}
\right.
\end{align*}
On the other hand,
\begin{align*}
E_\ba\ast g((x,y)_\bi)
&=
E_\ba \ast g((x,\dots,x))
=
\frac{1}{k} E_\ba\ast G_{(x,\dots,x)}
=
\left\{
\begin{array}{cl}
\frac{1}{k} & \mbox{if }\ba \mbox{ is constant}\\
0 & \mbox{otherwise}.
\end{array}
\right.
\end{align*}
Hence, the claim holds in this case. The case $\bi=(2,\dots,2)$ follows analogously. Suppose now that $\bi\not\in\{(1,\dots,1),(2,\dots,2)\}$. In this case,~\eqref{eqn_1511_1543} yields
\begin{align*}
E_\ba\ast q_{/\pi_\bi} 
&=
\left\{
\begin{array}{cl}
\frac{1}{k^2-k} & \mbox{if }\ba\sim \bi\\
0 & \mbox{otherwise}.
\end{array}
\right.
\end{align*}
On the other hand,
\begin{align*}
E_\ba\ast g((x,y)_\bi)
&=
\frac{1}{k^2-k}E_\ba\ast G_{(x,y)_\bi}
=
\left\{
\begin{array}{cl}
\frac{1}{k^2-k} & \mbox{if }\ba\sim (x,y)_\bi\\
0 & \mbox{otherwise},
\end{array}
\right.
\end{align*}
so the claim holds in this case, too, since $\bi\sim(x,y)_\bi$ and \enquote{$\sim$} is transitive. It follows that $g$ preserves ${R}$.

We now claim that $g$ satisfies the consistency equation~\eqref{eqn_consistency}. Take $\bx\in X^k$ and $\bi,\ba\in [k]^k$. Using Lemma~\ref{lem_multiplication_rule} and Lemma~\ref{lem_involved_1511_1718}, we find
\begin{align*}
E_\ba\ast \Pi_\bi\ast g(\bx)
&=
\sum_{\substack{\bb\in [k]^k\\ \bb_\bi=\ba}}E_\bb\ast g(\bx)
=
\frac{(k-|\{\bx\}|)!}{k!}\sum_{\substack{\bb\in [k]^k\\ \bb_\bi=\ba}}E_\bb\ast G_\bx\\
&=
\frac{(k-|\{\bx\}|)!}{k!}|\{\bb\in [k]^k: \bb_\bi=\ba \mbox{ and } \bb\sim\bx\}|\\
&=
\left\{
\begin{array}{cl}
\frac{(k-|\{\bx\}|)!}{k!}\frac{(k-|\{\bx_\bi\}|)!}{(k-|\{\bx\}|)!} & \mbox{if } \ba\sim\bx_\bi\\
0 & \mbox{otherwise}
\end{array}
\right.
\\
&=
\left\{
\begin{array}{cl}
\frac{(k-|\{\bx_\bi\}|)!}{k!} & \mbox{if } \ba\sim\bx_\bi\\
0 & \mbox{otherwise}.
\end{array}
\right.
\end{align*}
On the other hand,
\begin{align*}
E_\ba\ast g(\bx_\bi)
&=
\frac{(k-|\{\bx_\bi\}|)!}{k!}E_\ba\ast G_{\bx_\bi}
=
\left\{
\begin{array}{cl}
\frac{(k-|\{\bx_\bi\}|)!}{k!} & \mbox{if }\ba\sim\bx_\bi\\
0 & \mbox{otherwise},
\end{array}
\right.
\end{align*}
so the claim holds. It follows by Proposition~\ref{lem_a_symmetry} that $g$ preserves the $k$-ary symbol $R_k$. Hence, $g$ is a homomorphism from $\tilde{\X}^\tensor{k}$ to $\textbf{F}$, which concludes the proof.
\end{proof}
\begin{rem}
\label{cor_SAk_no_solves_approx_graph_color}
A direct consequence of Proposition~\ref{prop_SA_cliques} is that $\SA^k$ does not solve $\PCSP(\K_c,\K_d)$ if $3\leq c\leq d$ and $2\leq k\leq c$.
To see this, let us denote by $\tilde{\A}$ the structure obtained from a given structure $\A$ by adding the relation $R_k^{\tilde\A}=A^k$. By Proposition~\ref{prop_SA_cliques}, $\SA^k(\K_{d+1},\K_k)$ (and thus $\SA^k(\tilde \K_{d+1},\tilde \K_k)$) accepts. By Theorem~\ref{thm:main}, $\tilde{\K}_{d+1}^\tensor{k}\to\mathbb{F}_{\Qconv}(\tilde{\K}_k^\tensor{k})$. Since $k\leq c$, $\K_k\to \K_c$, so $\tilde{\K}_k\to\tilde{\K}_c$, which implies by Lemma~\ref{prop_correspondence_morphisms_tensor_structures} $\tilde{\K}_k^\tensor{k}\to\tilde{\K}_c^\tensor{k}$ and by Lemma~\ref{lem_A_B_free_struc} in Appendix~\ref{app:minions} $\mathbb{F}_{\Qconv}(\tilde{\K}_k^\tensor{k})\to\mathbb{F}_{\Qconv}(\tilde{\K}_c^\tensor{k})$. Composing the two homomorphisms, we get $\tilde{\K}_{d+1}^\tensor{k}\to\mathbb{F}_{\Qconv}(\tilde{\K}_c^\tensor{k})$ so, by Theorem~\ref{thm:main}, $\SA^k(\tilde{\K}_{d+1},\tilde{\K}_c)$ (and thus $\SA^k({\K}_{d+1},{\K}_c)$) accepts. However, clearly, $\K_{d+1}\not\to \K_d$.
\end{rem}

\begin{rem}
A second consequence of Proposition~\ref{prop_SA_cliques} (combined with Proposition~\ref{prop_sherali_adams_exact}) is the complete characterisation of the behaviour of Sherali-Adams on cliques: Letting $c,d,k\in \N$ be at least $2$, $\SA^k(\K_c,\K_d)$ accepts if and only if $\min(k,c)\leq d$.

To show this, let us denote by $\tilde{\A}$ the structure obtained from a given structure $\A$ by adding the relation $R_t^{\tilde\A}=A^t$ for each $t\in [k]$. 

If $k\leq d$, by Proposition~\ref{prop_SA_cliques} $\SA^d(\K_c,\K_d)$ accepts, which implies that $\SA^d(\tilde{\K}_c,\tilde{\K}_d)$ accepts whence, by Proposition~\ref{prop_SAkaccepts_hom_tensor}, $\tilde{\K}_c^\tensor{d}\to\mathbb{F}_{\Qconv}(\tilde{\K}_d^\tensor{d})$. Applying Proposition~\ref{prop_sherali_adams_smaller_level}, we deduce that $\tilde{\K}_c^\tensor{k}\to\mathbb{F}_{\Qconv}(\tilde{\K}_d^\tensor{k})$. From Theorem~\ref{thm:main}, we find that $\SA^k(\tilde{\K}_c,\tilde{\K}_d)$ accepts and, hence, $\SA^k({\K}_c,{\K}_d)$ accepts.

If $c\leq d$, the following chain of implications holds:
\[
\begin{array}{lll}
&\K_c\to \K_d \hspace{.4cm} \Rightarrow \hspace{.4cm} \tilde{\K}_c\to\tilde{\K}_d\\[2pt]
\Rightarrow & \tilde{\K}_c^\tensor{k}\to\tilde{\K}_d^\tensor{k} \to \mathbb{F}_{\Qconv}(\tilde{\K}_d^\tensor{k})\\[2pt]
\Rightarrow
& \SA^k(\tilde{\K}_c,\tilde{\K}_d) \mbox{ accepts}
\hspace{.4cm} \Rightarrow
\hspace{.4cm} \SA^k({\K}_c,{\K}_d) \mbox{ accepts},
\end{array}
\]
where the first and fourth implications are trivial, the second follows from Lemma~\ref{prop_correspondence_morphisms_tensor_structures} and Lemma~\ref{lem_A_maps_free_structure_of_A} in Appendix~\ref{app:minions}, and the third from Theorem~\ref{thm:main}.

If $m\coloneqq \min(k,c)> d$, suppose, for the sake of contradiction, that $\SA^k(\K_c,\K_d)$ accepts, whence it follows that $\SA^k(\tilde \K_c,\tilde \K_d)$ accepts. The following chain of implications holds:
\[
\begin{array}{llll}
&\K_m\to \K_c \hspace{.4cm} \Rightarrow \hspace{.4cm} \tilde{\K}_m\to\tilde{\K}_c\\[2pt]
\Rightarrow & \tilde{\K}_m^\tensor{k}\to\tilde{\K}_c^\tensor{k}\to \mathbb{F}_{\Qconv}(\tilde{\K}_d^\tensor{k})\\[2pt] 
\Rightarrow & 
\tilde{\K}_m^\tensor{m}\to\mathbb{F}_{\Qconv}(\tilde{\K}_d^\tensor{m})\\[2pt]
\Rightarrow & \tilde{\K}_m\to \tilde{\K}_d \hspace{.4cm} \Rightarrow \hspace{.4cm} \K_m\to \K_d,
\end{array}
\]
where the first and fifth implications are trivial, the second follows from Lemma~\ref{prop_correspondence_morphisms_tensor_structures} and Theorem~\ref{thm:main}, the third from Proposition~\ref{prop_sherali_adams_smaller_level}, and the fourth from Proposition~\ref{prop_sherali_adams_exact} (notice that $\tilde{\K}_m$ and $\tilde{\K}_d$ are $m$-enhanced). Since $m>d$, this leads to a contradiction.  
\end{rem}

We now use Proposition~\ref{prop_SA_cliques} to prove Theorem~\ref{main_theorem_approximate_colouring} in full generality by exploiting the properties of the line digraph construction. Let $\X$ be a digraph, whose binary edge relation shall be denoted by $R^\X$. The \emph{line digraph} of $\X$ is the digraph $\delta\X$ on vertex set $R^\X$ whose binary relation is $S^{\delta\X}=\{((x,y),(y,z)): (x,y),(y,z)\in R^\X\}$. 
The line digraph construction has been utilised in~\cite{WZ20,GS20:icalp} as a polynomial-time (and in fact log-space) reduction between $\PCSP$s. In particular, the construction changes the chromatic number in a controlled way, as we now describe.

Consider the integer functions $a$ and $b$ defined by $a(p)=2^p$ and $b(p)={p\choose \floor{p/2}}$ for $p\in\N$, and notice that $a(p)\geq b(p)$ for each $p$. Let also $a^{(i)}$ (resp. $b^{(i)}$) be the function obtained by iterating $a$ (resp. $b$) $i$-many times, for $i\in \N$. The following result bounds the chromatic number of the line digraph in terms of that of the original digraph. 
\begin{thm}[\cite{HarnerE72}]
\label{prop_chromatic_number_line_digraph}
Let $\X$ be a digraph and let $p\in\N$.
\begin{itemize}
\item
If $\delta\X\to\K_p$, then $\X\to\K_{a(p)}$.
\item
If $\X\to\K_{b(p)}$, then $\delta\X\to \K_p$.
\end{itemize}
\end{thm}

It was suggested by Jakub Opr\v{s}al that the line digraph construction not only provides a polynomial-time reduction, 
but also preserves acceptance by Sherali-Adams. We now prove this fact using the tensorisation machinery.
\begin{prop*}[Proposition~\ref{prop_reduction_line_digraph} restated]
Let $k\in\N$ with $k\geq 2$, let $\X,\A$ be digraphs, and suppose that $\SA^{2k}(\X,\A)$ accepts. Then $\SA^k(\delta\X,\delta\A)$ accepts.
\end{prop*}
\begin{proof}
Let $\tilde{\X}$ (resp. $\tilde{\A}$) be obtained from $\X$ (resp. $\A$) by adding the relations $R_t^{\tilde{\X}}=X^t$ (resp. $R_t^{\tilde{\A}}=A^t$) for each $t\in [2k]$. Let also $\widetilde{\delta\X}$ (resp. $\widetilde{\delta\A}$) be obtained from $\delta\X$ (resp. $\delta\A$) by adding the relation $R_k^{\widetilde{\delta\X}}=(R^\X)^k$ (resp. $R_k^{\widetilde{\delta\A}}=(R^\A)^k$). Clearly, $\SA^{2k}(\X,\A)$ accepts if and only if $\SA^{2k}(\tilde\X,\tilde\A)$ accepts, and $\SA^k(\delta\X,\delta\A)$ accepts if and only if $\SA^k(\widetilde{\delta\X},\widetilde{\delta\A})$ accepts. To avoid cumbersome notation, we shall not write the symbols $\sim$ explicitly. By virtue of Theorem~\ref{thm:main}, there exists a homomorphism $\xi:\X^\tensor{2k}\to \mathbb{F}_{\Qconv}(\A^\tensor{2k})$, and to prove the proposition we need to build a homomorphism $\vartheta:\delta\X^\tensor{k}\to \mathbb{F}_{\Qconv}(\delta\A^\tensor{k})$. Given a tuple of arcs $\bx=(\bx_1,\dots,\bx_k)\in (R^\X)^k$ such that $\bx_\ell=(x_{\ell ,1},x_{\ell, 2})\in R^\X$ for each $\ell\in [k]$, we shall consider the tuple of vertices $\bx'=(x_{1,1},x_{1,2},x_{2,1},x_{2,2},\dots,x_{k,1},x_{k,2})\in X^{2k}$; given $\ba\in (R^\A)^k$, we define $\ba'\in A^{2k}$ analogously. We set
\begin{align*}
E_\ba\ast\vartheta(\bx)\coloneqq E_{\ba'}\ast\xi(\bx').
\end{align*}
Observe that 
\begin{align}
\label{eqn_2502_1250}
\sum_{\ba\in (R^\A)^k}E_\ba\ast\vartheta(\bx)
&=
\sum_{\ba\in (R^\A)^k}E_{\ba'}\ast\xi(\bx')
\leq
\sum_{\bb\in A^{2k}}E_{\bb}\ast\xi(\bx')
=
1.
\end{align}
If~\eqref{eqn_2502_1250} hold with strict inequality, we would have that $E_\bb\ast\xi(\bx')>0$ for some $\bb\in A^{2k}$ such that $(b_{2i-1},b_{2i})\not\in R^\A$ for some $i\in [k]$. Taking $\bi=(2i-1,2i)\in [2k]^2$, this means that $\bb_\bi\not\in R^\A$, while, by definition of $\bx'$, $\bx'_\bi\in R^\X$. Applying Proposition~\ref{lem_partial_homo}, we get a contradiction. It follows that $\vartheta(\bx)\in\Qconv^{(m^k)}$, where $m=|R^\A|$. We claim that $\vartheta$ yields a homomorphism from $\delta\X^\tensor{k}$ to $\mathbb{F}_{\Qconv}(\delta\A^\tensor{k})$.

To show that $\vartheta$ preserves the binary symbol $S$, take $((x,y),(y,z))\in S^{\delta\X}$, so $((x,y),(y,z))^\tensor{k}\in S^{\delta\X^\tensor{k}}$. We need to check that $\vartheta(((x,y),(y,z))^\tensor{k})\in S^{\mathbb{F}_{\Qconv}(\delta\A^\tensor{k})}$. Observe that $(x,y,z)\in X^3=R_3^\X$, so $(x,y,z)^\tensor{2k}\in R_3^{\X^\tensor{2k}}$ and, hence, $\xi((x,y,z)^\tensor{2k})\in R_3^{\mathbb{F}_{\Qconv}(\A^{\tensor{2k}})}$. It follows that there exists $w\in \Qconv^{(|R_3^\A|)}=\Qconv^{(n^3)}$ such that $\xi((x,y,z)_\bi)=w_{/\pi_\bi}$ for each $\bi\in [3]^{2k}$. Let $\ell\coloneqq |S^{\delta\A}|$, and consider the rational $\ell$-vector $q$ defined by $\be_{((a,b),(b,c))}\ast q\coloneqq \be_{(a,b,c)}\ast w$ for $((a,b),(b,c))\in S^{\delta\A}$. We claim that $q\in\Qconv^{(\ell)}$. Observe that 
\begin{align}
\label{eqn_2502_1630}
\sum_{((a,b),(b,c))\in S^{\delta\A}}\be_{((a,b),(b,c))}\ast q
&=
\sum_{((a,b),(b,c))\in S^{\delta\A}}\be_{(a,b,c)}\ast w
\leq
\sum_{(a,b,c)\in A^3}\be_{(a,b,c)}\ast w
=1.
\end{align}
Suppose~\eqref{eqn_2502_1630} holds with strict inequality. It follows that $\be_{(a,b,c)}\ast w>0$ for some $(a,b,c)\in A^3$ such that $(a,b)\not\in R^\A$ or $(b,c)\not\in R^\A$. Assume $(a,b)\not\in R^\A$ (the other case follows analogously). Consider the tuples $\bi=(1,2,3,3,\dots,3)\in [3]^{2k}$, $\ba=(a,b,c,c,\dots,c)\in A^{2k}$, and $\bj=(1,2)\in [2k]^2$. Note that $((x,y,z)_\bi)_\bj=(x,y)\in R^\X$ while $\ba_\bj=(a,b)\not\in R^\A$. It follows from Proposition~\ref{lem_partial_homo} that $E_\ba\ast\xi((x,y,z)_\bi)=0$, whence we find
\begin{align*}
0
&=
E_\ba\ast\xi((x,y,z)_\bi)
=
E_\ba\ast w_{/\pi_\bi}
=
E_\ba\ast P_\bi\ast w
=
\sum_{\substack{\bb\in A^3\\\bb_\bi=\ba}}\be_\bb\ast w
=
\be_{(a,b,c)}\ast w,
\end{align*}
a contradiction (where we have used Lemma~\ref{lem_multiplication_rule_P}). We deduce that~\eqref{eqn_2502_1630} holds with equality, so $q\in\Qconv^{(\ell)}$ as claimed. Consider now $\bj=(j_1,\dots,j_k)\in [2]^k$, and set $\bx=((x,y),(y,z))_\bj\in (R^\X)^k$. We claim that $\vartheta(\bx)=q_{/\pi_\bj}$. This would conclude the proof that $\vartheta$ preserves $S$. For $\ba\in (R^\A)^k$, we find
\begin{align}
\label{eqn_2502_1832_A}
E_\ba\ast q_{/\pi_\bj}
&=
E_\ba\ast P_\bj\ast q
=
\sum_{\substack{((a,b),(b,c))\in S^{\delta\A}\\((a,b),(b,c))_\bj=\ba}}\be_{((a,b),(b,c))}\ast q
=
\sum_{\substack{((a,b),(b,c))\in S^{\delta\A}\\((a,b),(b,c))_\bj=\ba}}\be_{(a,b,c)}\ast w.
\end{align}
Consider the tuple $\bi=(i_1,\dots,i_{2k})\in [3]^{2k}$ defined by $i_{2\ell-1}=j_\ell$, $i_{2\ell}=j_\ell+1$ for each $\ell\in [k]$. Observe that this choice guarantees $(x,y,z)_\bi=\bx'$. Therefore, 
\begin{align}
\label{eqn_2502_1832_B}
E_\ba\ast\vartheta(\bx)
&=
E_{\ba'}\ast\xi(\bx')
=
E_{\ba'}\ast\xi((x,y,z)_\bi)
=
E_{\ba'}\ast w_{/\pi_\bi}
=
E_{\ba'}\ast P_\bi\ast w
=
\sum_{\substack{(a,b,c)\in A^3\\(a,b,c)_\bi=\ba'}}\be_{(a,b,c)}\ast w.
\end{align}
Note that, given $(a,b),(b,c)\in R^\A$, $((a,b),(b,c))_\bj=\ba$ if and only if $(a,b,c)_\bi=\ba'$. Hence, combining~\eqref{eqn_2502_1832_A} and~\eqref{eqn_2502_1832_B} yields 
$
E_\ba\ast q_{/\pi_\bj}\leq E_\ba\ast\vartheta(\bx)
$,
and strict inequality can happen only if there exists a tuple $(a,b,c)\in A^3$ such that $\be_{(a,b,c)}\ast w>0$ and at least one among $(a,b)$ and $(b,c)$ does not belong to $R^\A$. As shown above, this leads to a contradiction. It follows that  
$
E_\ba\ast q_{/\pi_\bj}= E_\ba\ast\vartheta(\bx)
$, which proves the claim.

To show that $\vartheta$ preserves $R_k$, we shall prove that the consistency equation~\eqref{eqn_consistency} holds, and then apply Proposition~\ref{lem_a_symmetry}. Take $\bx\in (R^\X)^k$, $\bi=(i_1,\dots,i_k)\in [k]^k$, and $\ba\in (R^\A)^k$. Consider also the tuple $\bj=(j_1,\dots,j_{2k})\in [2k]^{2k}$ defined by $j_{2\ell-1}=2i_\ell-1$, $j_{2\ell}=2i_\ell$ for $\ell\in [k]$. We find
\begin{align}
\label{eqn_1633_2602}
\begin{array}{lllllllll}
&\displaystyle
E_\ba\ast\Pi_\bi\ast\vartheta(\bx)
&=&\displaystyle
\sum_{\substack{\bb\in (R^\A)^k\\\bb_\bi=\ba}}E_\bb\ast\vartheta(\bx)
&=&\displaystyle
\sum_{\substack{\bb\in (R^\A)^k\\\bb_\bi=\ba}}
E_{\bb'}\ast\xi(\bx')
&\leq& \displaystyle
\sum_{\substack{\bc\in A^{2k}\\\bc_\bj=\ba'}}E_\bc\ast\xi(\bx')\\[25pt]
=&\displaystyle
E_{\ba'}\ast\Pi_\bj\ast\xi(\bx')
&=&\displaystyle
E_{\ba'}\ast\xi((\bx')_\bj)
&=&\displaystyle
E_{\ba'}\ast\xi((\bx_\bi)')
&=&\displaystyle
E_\ba\ast\vartheta(\bx_\bi),
\end{array}
\end{align}
where the first and third equalities follow from Lemma~\ref{lem_multiplication_rule}, the fourth from Proposition~\ref{lem_a_symmetry} applied to $\xi$, the fifth from the fact that $(\bx')_\bj=(\bx_\bi)'$, and the inequality from the fact that $\bb_\bi=\ba$ if and only if $\bb'_\bj=\ba'$. If strict inequality hold in~\eqref{eqn_1633_2602}, there would exist some $\bc\in A^{2k}$ such that $E_\bc\ast\xi(\bx')>0$ and $(c_{2\ell-1},c_{2\ell})\not\in R^\A$ for some $\ell\in [k]$ -- which would contradict Proposition~\ref{lem_partial_homo}. It follows that $\vartheta$ satisfies the consistency equation, and is thus a homomorphism.
\end{proof}

We have now all the ingredients to strengthen Remark~\ref{cor_SAk_no_solves_approx_graph_color}
and prove that the approximate graph colouring problem is not solved by any constant level of Sherali-Adams.

\begin{proof}[Proof of Theorem~\ref{main_theorem_approximate_colouring}]
Suppose, for the sake of contradiction, that $\SA^k$ solves $\PCSP(\K_c,\K_d)$ for some $3\leq c\leq d$. Since $\SA^2$ is at least as strong as $\SA^1$, we can assume $k\geq 2$. We will now use Theorem~\ref{prop_chromatic_number_line_digraph} and Proposition~\ref{prop_reduction_line_digraph} to derive a contradiction to Proposition~\ref{prop_SA_cliques}. 

If $c\geq 4$, we can take $i\in\N$ such that $b^{(i)}(c)\geq k\cdot 2^i$. Consider the graph $\X=\K_{a^{(i)}(d)+1}$ and observe that, by Proposition~\ref{prop_SA_cliques} (cf.~Remark~\ref{cor_SAk_no_solves_approx_graph_color}), $\SA^{k\cdot 2^i}(\X,\K_{b^{(i)}(c)})$ accepts. Applying Proposition~\ref{prop_reduction_line_digraph} repeatedly, we find that $\SA^k(\delta^{(i)}\X,\delta^{(i)}\K_{b^{(i)}(c)})$ accepts, where $\delta^{(i)}$ denotes the application of the line digraph operator $i$-many times. Indicating by $\sim$ the $k$-enhanced completion of a structure, Theorem~\ref{thm:main} then yields 
\begin{align*}
(\widetilde{\delta^{(i)}\X})^\tensor{k}\to \mathbb{F}_{\Qconv}((\widetilde{\delta^{(i)}\K_{b^{(i)}(c)}})^\tensor{k}). 
\end{align*}
Applying the second part of Theorem~\ref{prop_chromatic_number_line_digraph} repeatedly, we find $\delta^{(i)}\K_{b^{(i)}(c)}\to\K_c$. In combination with Lemma~\ref{prop_correspondence_morphisms_tensor_structures} and Lemma~\ref{lem_A_B_free_struc} in Appendix~\ref{app:minions}, this yields
\begin{align*}
\mathbb{F}_{\Qconv}((\widetilde{\delta^{(i)}\K_{b^{(i)}(c)}})^\tensor{k})\to \mathbb{F}_{\Qconv}(\tilde\K_c^\tensor{k})
\end{align*}
so that, composing the two homomorphisms, we obtain
\begin{align*}
(\widetilde{\delta^{(i)}\X})^\tensor{k}\to\mathbb{F}_{\Qconv}(\tilde\K_c^\tensor{k}).
\end{align*}
It follows that $\SA^k(\delta^{(i)}\X,\K_c)$ accepts; since we are supposing that $\SA^k$ solves $\PCSP(\K_c,\K_d)$, this yields  $\delta^{(i)}\X\to\K_d$. Applying the first part of Theorem~\ref{prop_chromatic_number_line_digraph} repeatedly, we find $\X\to\K_{a^{(i)}(d)}$, a contradiction.

If $c=3$, we use the fact, observed in~\cite[Lemma~4.19]{KrokhinOWZ20}, that $\delta^{(2)}\K_4\to\K_3$. Let $\X$ be a digraph such that $\SA^{4k}(\X,\K_4)$ accepts. From Proposition~\ref{prop_reduction_line_digraph}, we find that $\SA^k(\delta^{(2)}\X,\delta^{(2)}\K_4)$ accepts. Reasoning as above, we obtain 
\begin{align*}
(\widetilde{\delta^{(2)}\X})^\tensor{k}\to\mathbb{F}_{\Qconv}((\widetilde{\delta^{(2)}\K_4})^\tensor{k})\to\mathbb{F}_{\Qconv}(\tilde{\K}_3^\tensor{k}),
\end{align*} 
so $\SA^k(\delta^{(2)}\X,\K_3)$ accepts, whence we get $\delta^{(2)}\X\to\K_d$ and, from the first part of Theorem~\ref{prop_chromatic_number_line_digraph}, $\X\to\K_{a^{(2)}(d)}$. This means that $\SA^{4k}$ solves $\PCSP(\K_4,\K_{a^{(2)}(d)})$, which yields a contradiction as discussed above.
\end{proof}

\section{Acceptance vs. solvability}
\label{sec:solv}

In light of the acceptance characterisation obtained in Theorem~\ref{thm:main},
one may be tempted to conjecture that $\SA^k$ solves a $\PCSP$ template $(\A,\B)$ if and only if there exists a minion homomorphism $\Qconv\to\Pol(\A^\tensor{k},\B^\tensor{k})$. The latter condition is certainly sufficient. Indeed, it follows from~\cite[Lemma~4.4]{BBKO21} that $\Qconv\to\Pol(\A^\tensor{k},\B^\tensor{k})$ is equivalent to $\mathbb{F}_{\Qconv}(\A^{\tensor{k}})\to\B^\tensor{k}$. If this is the case, Proposition~\ref{prop_SAkaccepts_hom_tensor} implies that any $\sigma$-structure $\X$ such that $\SA^k(\X,\A)$ accepts satisfies $\X^\tensor{k}\to\mathbb{F}_{\Qconv}(\A^{\tensor{k}})\to\B^\tensor{k}$. By virtue of Lemma~\ref{prop_correspondence_morphisms_tensor_structures}, this yields $\X\to\B$. Also, clearly, if $\SA^k(\X,\A)$ does not accept then $\X\not\to\A$. 

However,
for $k\geq 2$ the condition $\Qconv\to\Pol(\A^\tensor{k},\B^\tensor{k})$ is far from being necessary. A consequence of Theorem~\ref{thm_pol_AB_pol_tensor_iso}, which we shall prove in this section, is that $\Qconv\to\Pol(\A^\tensor{k},\B^\tensor{k})$ is in fact equivalent to $\Qconv\to\Pol(\A,\B)$; this last condition, in turn, is equivalent to $\BLP$ solving $\PCSP(\A,\B)$ (cf.~Theorem~\ref{thm:blp} in Appendix~\ref{app:BLP}). Since $\SA^k$ is strictly stronger than $\BLP$ (for example, when applied to $\CSP$s), we deduce that $\Qconv\to\Pol(\A^\tensor{k},\B^\tensor{k})$ is not necessary for $\SA^k$ to solve $\PCSP(\A,\B)$.

We start off with two technical results -- Proposition~\ref{prop_iso_tensor_cartesian} and Lemma~\ref{lem_k_enhanced_cartesian} -- both needed in the proof of Theorem~\ref{thm_pol_AB_pol_tensor_iso}. 

An \emph{isomorphism} between two $\sigma$-structures (or two minions) is a bijective homomorphism whose inverse is also a homomorphism. 
\begin{prop}
\label{prop_iso_tensor_cartesian}
Let $L,k\in \N$ and let $\A$ be a $\sigma$-structure. Then there exists an isomorphism $\xi:{(\A^L)}^\tensor{k}\to {\left(\A^\tensor{k}\right)}^L$. 
\end{prop}
\begin{proof}
Consider the map $\xi:{(A^L)}^k\to {(A^k)}^L$ defined as follows: Given $\ba_1,\dots,\ba_k\in A^L$, let $M\in A^{L\times k}$ be the matrix satisfying $M\be_i=\ba_i$ $\forall i\in [k]$, and set $\xi((\ba_1,\dots,\ba_k))\coloneqq(M^T\be_1,\dots,M^T\be_L)$. Clearly, $\xi$ is bijective. We now show that $\xi$ is a homomorphism from ${(\A^L)}^\tensor{k}$ to ${\left(\A^\tensor{k}\right)}^L$ and $\xi^{-1}$ is a homomorphism from ${\left(\A^\tensor{k}\right)}^L$ to ${(\A^L)}^\tensor{k}$.
 
To prove that $\xi$ is a homomorphism, choose $R\in\sigma$ of arity $r$, let $\tau=(\tau_1,\dots,\tau_r)\in R^{\A^L}$, where $\tau_i\in A^L$ $\forall i\in [r]$, and consider $\tau^\tensor{k}\in R^{{(\A^L)}^\tensor{k}}$. Let $Q_{\tau}\in A^{L\times r}$ be the matrix satisfying $Q_{\tau}\be_i=\tau_i$ $\forall i\in [r]$. By the definition of the $L$-th power of a relational structure, $Q_{\tau}^T\be_j\in R^\A$ $\forall j\in [L]$; hence, $(Q_{\tau}^T\be_j)^\tensor{k}\in R^{\A^\tensor{k}}$ $\forall j\in [L]$. Consider now the tensor $\varphi\in\cT^{k;r}((A^k)^L)$ defined by $E_\bi\ast\varphi\coloneqq\left((Q_\tau^T\be_1)_\bi,\dots,(Q_\tau^T\be_L)_\bi\right)$ $\forall \bi\in [r]^k$. Observe that, for any $j\in [L]$, $\left[(Q_\tau^T\be_j)_\bi\right]_{\bi\in [r]^k}=(Q_{\tau}^T\be_j)^\tensor{k}\in R^{\A^\tensor{k}}$, so $\varphi\in R^{{\left(\A^\tensor{k}\right)}^L}$. We claim that $\xi(\tau^\tensor{k})=\varphi$; this would imply that $\xi$ is a homomorphism. Let $\bi=(i_1,\dots,i_k)\in [r]^k$, and consider the matrix $K\in \{0,1\}^{r\times k}$ satisfying $K\be_j=\be_{i_j}$ $\forall j\in [k]$. We find
\begin{align*}
E_\bi\ast\xi\left(\tau^\tensor{k}\right)
&=
\xi\left(E_\bi\ast \tau^\tensor{k}\right)
=
\xi(\tau_\bi)
=
\xi((\tau_{i_1},\dots,\tau_{i_k}))\\
&=
\xi\left(\left(Q_\tau\be_{i_1},\dots,Q_\tau\be_{i_k}\right)\right)\\
&=
\xi\left(\left(Q_\tau K\be_1,\dots,Q_\tau K\be_k\right)\right)\\
&=
\left(K^TQ_\tau^T\be_1,\dots,K^TQ_\tau^T\be_L\right)\\
&=
\left(\left(Q_\tau^T\be_1\right)_\bi,\dots,\left(Q_\tau^T\be_L\right)_\bi\right)
=
E_\bi\ast\varphi,
\end{align*}
from which the claim follows.

To prove that $\xi^{-1}$ is a homomorphism, choose $R\in\sigma$ of arity $r$, and take $\varphi\in R^{{\left(\A^\tensor{k}\right)}^L}$ -- that is, $\varphi$ is a tensor in $\cT^{k;r}((A^k)^L)$ such that, for each $j\in [L]$, \[
\varphi_{(j)}\coloneqq \left[\left(E_\bi\ast \varphi\right)_j\right]_{\bi\in [r]^k}\in R^{\A^\tensor{k}},
\]
where $\left(E_\bi\ast \varphi\right)_j\in A^k$ is the $j$-th element in the tuple $E_\bi\ast \varphi\in (A^k)^L$. Therefore, for each $j\in [L]$, $\varphi_{(j)}=\ba_{(j)}^\tensor{k}$ for some $\ba_{(j)}\in R^\A$. Consider the matrix $Q\in A^{L\times r}$ satisfying $Q^T\be_j=\ba_{(j)}$ $\forall j\in [L]$, and observe that $\tau\coloneqq\left(Q\be_1,\dots,Q\be_r\right)\in R^{\A^L}$, so $\tau^\tensor{k}\in R^{{(\A^L)}^\tensor{k}}$. We claim that $\xi^{-1}(\varphi)=\tau^\tensor{k}$; this would imply that $\xi^{-1}$ is a homomorphism. Let $\bi=(i_1,\dots,i_k)\in [r]^k$, and consider the matrix $K$ introduced above. We find
\begin{align*}
E_\bi\ast \xi^{-1}(\varphi)
&=
\xi^{-1}(E_\bi\ast\varphi)
=
\xi^{-1}\left(\left(\left(E_\bi\ast\varphi\right)_1,\dots,\left(E_\bi\ast\varphi\right)_L\right)\right)\\
&=
\xi^{-1}\left(\left(E_\bi\ast\varphi_{(1)},\dots,E_\bi\ast\varphi_{(L)}\right)\right)\\
&=
\xi^{-1}\left(\left(E_\bi\ast\ba_{(1)}^\tensor{k},\dots,E_\bi\ast\ba_{(L)}^\tensor{k}\right)\right)\\
&=
\xi^{-1}\left(\left((\ba_{(1)})_\bi,\dots,(\ba_{(L)})_\bi\right)\right)\\
&=
\xi^{-1}\left(\left(K^T\ba_{(1)},\dots,K^T\ba_{(L)}\right)\right)\\
&=
\xi^{-1}\left(\left(K^TQ^T\be_1,\dots,K^TQ^T\be_L\right)\right)\\
&=
\xi^{-1}\left(\left((QK)^T\be_1,\dots,(QK)^T\be_L\right)\right)\\
&=
\left(QK\be_1,\dots,QK\be_k\right)
=
\left(Q\be_{i_1},\dots,Q\be_{i_k}\right)\\
&=
\tau_\bi
=
E_\bi\ast\tau^\tensor{k},
\end{align*}
from which the claim follows.
\end{proof}
\begin{lem}
\label{lem_k_enhanced_cartesian}
Let $k,L\in\N$, and let $\A$ be a $k$-enhanced $\sigma$-structure. Then $\A^L$ is $k$-enhanced, too.
\end{lem}
\begin{proof}
Since $\A$ is $k$-enhanced, $\sigma$ contains a $k$-ary symbol $R_k$ such that $R_k^\A=A^k$. We claim that $R_k^{\A^L}=(A^L)^k$. Clearly, we have $R_k^{\A^L}\subseteq (A^L)^k$. To prove the other inclusion, consider $k$ tuples $\ba_1,\dots,\ba_k\in A^L$, and let $M$ be the matrix in $A^{L\times k}$ satisfying $M\be_i=\ba_i$ for $i\in [k]$. Note that, for each $j\in [L]$, $M^T\be_j\in A^k=R_k^\A$. Hence, by the definition of the $L$-th power of a relational structure, $(\ba_1,\dots,\ba_k)\in R_k^{\A^L}$. This concludes the proof of the claim and of the lemma.
\end{proof}
We say that two minions $\mathscr{M}$ and $\mathscr{N}$ are \emph{homomorphically equivalent} if $\mathscr{M}\to\mathscr{N}$ and $\mathscr{N}\to\mathscr{M}$.
\begin{thm}
\label{thm_pol_AB_pol_tensor_iso}
Let $k\in\N$ and let $(\A,\B)$ be a $\PCSP$ template. Then $\Pol(\A,\B)$ and $\Pol(\A^\tensor{k},\B^\tensor{k})$ are homomorphically equivalent; if, in addition, $\A$ is $k$-enhanced, then $\Pol(\A,\B)$ and $\Pol(\A^\tensor{k},\B^\tensor{k})$ are isomorphic.
\end{thm}
\begin{proof}
For $L\in\N$, consider the maps
\[
\begin{array}{lrcl}
\rho_L:&\Hom(\A^L,\B)&\to&\Hom\left({(\A^L)}^\tensor{k},\B^\tensor{k}\right),\\
&f&\mapsto& f^\ast\\[5pt]
\rho'_L:&\Hom\left({(\A^L)}^\tensor{k},\B^\tensor{k}\right)&\to&\Hom(\A^L,\B)\\
&g&\mapsto& g_\ast
\end{array}
\]
introduced in the proof of Lemma~\ref{prop_correspondence_morphisms_tensor_structures}. In addition, consider the isomorphism $\xi_L:{(\A^L)}^\tensor{k}\to {\left(\A^\tensor{k}\right)}^L$ described in Proposition~\ref{prop_iso_tensor_cartesian}. We define the map $\psi:\Pol(\A,\B)\to\Pol(\A^\tensor{k},\B^\tensor{k})$ by setting $f\mapsto \rho_L(f)\circ\xi_L^{-1}$ for $f\in\left(\Pol(\A,\B)\right)^{(L)}$. Moreover, we define the map $\psi':\Pol(\A^\tensor{k},\B^\tensor{k})\to\Pol(\A,\B)$ by setting $g\mapsto \rho'_L(g\circ\xi_L)$ for $g\in \left(\Pol(\A^\tensor{k},\B^\tensor{k})\right)^{(L)}$. Clearly, both $\psi$ and $\psi'$ preserve the arities. We now show that they also preserve the minors. Choose a map $\pi:[L]\to [L']$ for some $L,L'\in\N$, and let ${P}$ be the matrix in $\{0,1\}^{L'\times L}$ satisfying ${P}\be_i=\be_{\pi(i)}$ for $i\in [L]$.

Consider a map $f:\A^L\to\B$ and $L'$ tuples $\ba_1,\dots,\ba_{L'}\in A^k$, and let $M\in A^{k\times L'}$ be the matrix satisfying $M\be_i=\ba_i$ for $i\in[L']$. Observe that
{
\allowdisplaybreaks
\begin{align*}
\psi(f)_{/\pi}((\ba_1,\dots,\ba_{L'}))
&=
\psi(f)_{/\pi}((M\be_1,\dots,M\be_{L'}))\\
&=
\psi(f)((M\be_{\pi(1)},\dots,M\be_{\pi(L)}))\\
&=
\psi(f)((M{P}\be_1,\dots,M{P}\be_L))\\
&=
\rho_L(f)\left(\xi_L^{-1}((M{P}\be_1,\dots,M{P}\be_L))\right)\\
&=
\rho_L(f)(({P}^TM^T\be_1,\dots,{P}^TM^T\be_k))\\
&=
f^*(({P}^TM^T\be_1,\dots,{P}^TM^T\be_k))\\
&=(f({P}^TM^T\be_1),\dots,f({P}^TM^T\be_k)),\\
\psi(f_{/\pi})((\ba_1,\dots,\ba_{L'}))
&=
\psi(f_{/\pi})((M\be_1,\dots,M\be_{L'}))\\
&=
\left[\rho_{L'}(f_{/\pi})\right]\left(\xi_{L'}^{-1}((M\be_1,\dots,M\be_{L'}))\right)\\
&=
\left[\rho_{L'}(f_{/\pi})\right]((M^T\be_1,\dots,M^T\be_k))\\
&=
(f_{/\pi})^*((M^T\be_1,\dots,M^T\be_k))\\
&=\left(f_{/\pi}(M^T\be_1),\dots,f_{/\pi}(M^T\be_k)\right)\\
&=
(f({P}^TM^T\be_1),\dots,f({P}^TM^T\be_k)),
\end{align*}
}
so $\psi(f)_{/\pi}=\psi(f_{/\pi})$. Hence, $\psi$ is a minion homomorphism.

Consider now a map $g:\left(\A^\tensor{k}\right)^L\to\B^\tensor{k}$ and a tuple $\ba=(a_1,\dots,a_{L'})\in A^{L'}$. We find
\begin{align*}
\psi'(g)_{/\pi}(\ba)
&=
\psi'(g)({P}^T\ba)
=
[\rho'_L(g\circ\xi_L)]({P}^T\ba)\\
&=
(g\circ\xi_L)_*({P}^T\ba)
=
\be_1^T[(g\circ\xi_L)(({P}^T\ba,\dots,{P}^T\ba))]\\
&=
\be_1^T[g(\xi_L(({P}^T\ba\bone_k^T\be_1,\dots,{P}^T\ba\bone_k^T\be_k)))]\\
&=
\be_1^T[g((\bone_k\ba^T{P}\be_1,\dots,\bone_k\ba^T{P}\be_L))],\\
\psi'(g_{/\pi})(\ba)
&=
[\rho'_{L'}(g_{/\pi}\circ\xi_{L'})](\ba)
=
(g_{/\pi}\circ\xi_{L'})_*(\ba)\\
&=
\be_1^T\left[(g_{/\pi}\circ\xi_{L'})((\ba,\dots,\ba))\right]\\
&=
\be_1^T\left[g_{/\pi}(\xi_{L'}((\ba\bone_k^T\be_1,\dots,\ba\bone_k^T\be_k)))\right]\\
&=
\be_1^T\left[g_{/\pi}((\bone_k\ba^T\be_1,\dots,\bone_k\ba^T\be_{L'}))\right]\\
&=
\be_1^T\left[g((\bone_k\ba^T\be_{\pi(1)},\dots,\bone_k\ba^T\be_{\pi(L)}))\right]\\
&=
\be_1^T\left[g((\bone_k\ba^T{P}\be_1,\dots,\bone_k\ba^T{P}\be_L))\right],
\end{align*}
so $\psi'(g)_{/\pi}=\psi'(g_{/\pi})$. Hence, $\psi'$ is a minion homomorphism, too. This concludes the proof that $\Pol(\A,\B)$ and $\Pol(\A^\tensor{k},\B^\tensor{k})$ are homomorphically equivalent.

Suppose now that $\A$ is $k$-enhanced. By virtue of Lemma~\ref{lem_k_enhanced_cartesian}, $\A^L$ is $k$-enhanced for any $L\in\N$. Using part $(ii)$ of Lemma~\ref{prop_correspondence_morphisms_tensor_structures}, we deduce that, in this case, the maps $\rho_L$ and $\rho'_L$ are mutually inverse bijections. Therefore, given $f:\A^L\to\B$ and $g:\left(\A^\tensor{k}\right)^L\to\B^\tensor{k}$, we find
\begin{align*}
\begin{array}{lllllllllllll}
(\psi'\circ\psi)(f)
&=&
\psi'(\rho_L(f)\circ\xi_L^{-1})
&=&
\rho'_L(\rho_L(f)\circ\xi_L^{-1}\circ\xi_L)
&=&
\rho'_L(\rho_L(f))
&=&
f,\\[7pt]
(\psi\circ\psi')(g)
&=&
\psi(\rho'_L(g\circ\xi_L))
&=&
\rho_L(\rho'_L(g\circ\xi_L))\circ\xi_L^{-1}
&=&
g\circ\xi_L\circ\xi_L^{-1}
&=&
g.
\end{array}
\end{align*}
Hence, $\psi$ and $\psi'$ are mutually inverse minion homomorphisms, so $\Pol(\A,\B)$ and $\Pol(\A^\tensor{k},\B^\tensor{k})$ are isomorphic.
\end{proof}

As seen at the beginning of the section, a consequence of Theorem~\ref{thm_pol_AB_pol_tensor_iso} is that the condition $\Qconv\to\Pol(\A^\tensor{k},\B^\tensor{k})$ does not characterise solvability of a $\PCSP$ template $(\A,\B)$ through $\SA^k$ and, in fact, it is equivalent to solvability through $\BLP$. In the next, final remark we argue that the theory of \emph{tensor rank} (cf.~\cite{lim2013tensors}) could play a role in understanding the power of Sherali-Adams. 

\begin{rem}
\label{remark_rank_one_segre}
One primary message of this work is that the multilinear structure $\mathbb{F}_{\Qconv}(\A^\tensor{k})$ determines acceptance for the $k$-th level of Sherali-Adams in the same way as the linear structure 
$\mathbb{F}_{\Qconv}(\A)$ determines acceptance for $\BLP$. What is the connection between these two structures? We recall from Proposition~\ref{prop_f_q_a_k} that there exists a natural homomorphism $\psi$ from $\mathbb{F}_{\Qconv}(\A^\tensor{k})$ to $[\mathbb{F}_{\Qconv}(\A)]^\tensor{k}$. We have also seen that $[\mathbb{F}_{\Qconv}(\A)]^\tensor{k}$ carries the same information as $\mathbb{F}_{\Qconv}(\A)$ (in the sense of Lemma~\ref{prop_correspondence_morphisms_tensor_structures}); indeed, the problem of understanding whether an input structure $\X$ is homomorphic to $\mathbb{F}_{\Qconv}(\A)$ is equivalent to the problem of understanding whether $\X^\tensor{k}$ is homomorphic to $[\mathbb{F}_{\Qconv}(\A)]^\tensor{k}$. Now, intuitively speaking, the relations in $\mathbb{F}_{\Qconv}(\A^\tensor{k})$ should be seen as sets of tensors having generic rank, while the relations in $[\mathbb{F}_{\Qconv}(\A)]^\tensor{k}$ are sets of rank-one tensors obtained as outer Segre products of vectors (cf.~Footnote~\ref{footnote_segre}). Then, the homomorphism $\psi$ in the proof of Proposition~\ref{prop_f_q_a_k} can be viewed as a \emph{compressive-sensing} procedure taking a tensor $T$ as input and returning an approximation of $T$ having rank one as output. In other words, the structure corresponding to $\BLP$ is essentially a rank-one compression of the structure corresponding to $\SA^k$. Reconstructing the signal from a rank-one compression is impossible in general, so we do not have a homomorphism from $[\mathbb{F}_{\Qconv}(\A)]^\tensor{k}$ to $\mathbb{F}_{\Qconv}(\A^\tensor{k})$ for $k\geq 2$ -- which corresponds to the fact that Sherali-Adams is strictly more powerful than $\BLP$. 

This informal argument suggests a strategy to characterise the power of the Sherali-Adams hierarchy in terms of the (known) characterisation for $\BLP$ by studying the possible ways to decompose the tensors in $\mathbb{F}_{\Qconv}(\A^\tensor{k})$ in combinations of rank-one components -- i.e., the so-called \emph{rank-retaining decompositions} of the tensors (for example, see~\cite{kruskal1977three}).
\end{rem}

\begin{appendices}

\section{BLP relaxation}
\label{app:BLP}
Let $\X$ and $\A$ be $\sigma$-structures. We introduce a variable $\lambda_x(a)$ for every $x\in X$, $a\in A$, and a variable $\lambda_{R,\bx}(\ba)$ for every $R\in\sigma,\bx\in R^\X,\ba\in R^\A$. The system
\begin{align}
\label{eqn_BLP}
\begin{array}{lll}
\displaystyle\sum_{a\in A}\lambda_x(a)&=1 \hspace{2cm}& x\in X\\
\displaystyle\sum_{\substack{\ba\in R^\A\\a_i=a}}\lambda_{R,\bx}(\ba)&=\lambda_{x_i}(a) & R\in\sigma, \bx\in R^\X, a\in A, i\in [\ar(R)]
\end{array}
\end{align}
admits a $0$--$1$ solution exactly when $\X\to\A$. If the variables are allowed to take values in the interval $[0,1]$, we can check the existence of a solution to~\eqref{eqn_BLP} in polynomial time (in the size of $\X$). This yields the so-called basic linear programming ($\BLP$) relaxation~\cite{Kun12:itcs}: We say that $\BLP(\X,\A)$ \emph{accepts} if~\eqref{eqn_BLP} has a solution such that all variables take real (equivalently, rational) values in $[0,1]$. From~\cite{BBKO21}, this is equivalent to the existence of a homomorphism $\X\to\mathbb{F}_{\Qconv}(\A)$.

By construction, if $\X\to\A$ then
$\BLP(\X,\A)$ accepts. We say that $\BLP$ \emph{solves} a $\PCSP$ template $(\A,\B)$ if whenever $\BLP(\X,\A)$ accepts we have $\X\to\B$.
Solvability of $\PCSP$s (and, thus, $\CSP$s) through $\BLP$ is characterised algebraically in the following result.
\begin{thm}[\cite{BBKO21}]\label{thm:blp}
Let $(\A,\B)$ be a $\PCSP$ template. The following are equivalent:
  \begin{enumerate}
    \item[(1)] $\BLP$ solves $\PCSP(\A,\B)$;
    \item[(2)] $\Pol(\A,\B)$ admits a minion homomorphism from $\Qconv$;
    \item[(3)] $\Pol(\A,\B)$ contains symmetric operations\footnote{An $L$-ary operation $f:A^L\to B$ is called \emph{symmetric} if
$f(a_1,\ldots,a_L)=f(a_{\pi(1)},\ldots,a_{\pi(L)})$ for every $a_1,\ldots,a_L\in
A$ and every permutation $\pi:[L]\to[L]$.} 
of all arities.
  \end{enumerate}
\end{thm}

\section{Two facts about minions}
\label{app:minions}

We present two simple facts about minions. The first is implicitly proved in~\cite{BBKO21} for the special case of minions of operations.

\begin{lem}\label{lem_A_maps_free_structure_of_A}
Let $\mathscr{M}$ be a nonempty minion and let $\A$ be a $\sigma$-structure. Then $\A\to\mathbb{F}_{\mathscr{M}}(\A)$.
\end{lem}
\begin{proof}
Since $\mathscr{M}$ is nonempty, it contains a unary element $M\in\mathscr{M}^{(1)}$. Consider the map 
\begin{align*}
f:A&\to\mathscr{M}^{(n)}\\
a&\mapsto M_{/\rho_a}
\end{align*}
where $\rho_a:[1]\to [n]=A$ is defined by $\rho_a(1)=a$. Take $R\in\sigma$ of arity $r$, and consider a tuple $\ba=(a_1,\dots,a_r)\in R^\A$. Let $m=|R^\A|$, and consider the function $\pi:[1]\to [m]$ defined by $\pi(1)=\ba$. Let $q=M_{/\pi}\in\mathscr{M}^{(m)}$. For each $i\in [r]$, consider the function $\pi_i:[m]\to [n]$ defined by $\pi_i(\bb)=b_i$, where $\bb=(b_1,\dots,b_r)\in R^\A$. Observe that $\rho_{a_i}=\pi_i\circ\pi$ for each $i\in[r]$. We obtain
\begin{align*}
f(\ba)&=(f(a_1),\dots,f(a_r))
=
(M_{/\rho_{a_1}},\dots,M_{/\rho_{a_r}})
=
(M_{/\pi_1\circ\pi},\dots,M_{/\pi_r\circ\pi})\\
&=  
  ((M_{/\pi})_{/\pi_1},\dots,(M_{/\pi})_{/\pi_r})
=
(q_{/\pi_1},\dots,q_{/\pi_r})
\in 
R^{\mathbb{F}_{\mathscr{M}}(\A)},
\end{align*}
thus showing that $f$ is a homomorphism from $\A$ to $\mathbb{F}_{\mathscr{M}}(\A)$.
\end{proof}

\begin{lem} \label{lem_A_B_free_struc}
Let $\mathscr{M}$ be a minion and let $\A,\B$ be $\sigma$-structures with $\A\to\B$. Then $\mathbb{F}_{\mathscr{M}}(\A)\to\mathbb{F}_{\mathscr{M}}(\B)$.
\end{lem}
\begin{proof}
Fix a homomorphism $f:\A\to\B$, and consider the map $\xi:\mathscr{M}^{(|A|)}\to\mathscr{M}^{(|B|)}$ defined by $M\mapsto M_{/f}$ for $M\in \mathscr{M}^{(|A|)}$. We need to show that $\xi$ yields a homomorphism from $\mathbb{F}_{\mathscr{M}}(\A)$ to $\mathbb{F}_{\mathscr{M}}(\B)$. Let $R\in\sigma$ be a relation symbol of arity $r$, and take a tuple $(M_1,\dots,M_r)\in R^{\mathbb{F}_{\mathscr{M}}(\A)}$. This means that there exists $q\in \mathscr{M}^{(|R^\A|)}$ such that $M_i=q_{/\pi_i}$ for each $i\in [r]$, where $\pi_i:R^\A\to A$ is defined by $\ba\mapsto a_i$ for $\ba=(a_1,\dots,a_r)\in R^\A$. We have that $\xi((M_1,\dots,M_r))=({M_1}_{/f},\dots,{M_r}_{/f})$. Consider the map $g:R^\A\to R^\B$ defined by $\ba\mapsto f(\ba)$ for $\ba\in R^\A$, and take $\tilde{q}= q_{/g}\in \mathscr{M}^{(|R^\B|)}$. We claim that ${M_i}_{/f}=\tilde{q}_{/\tilde{\pi}_i}$ for $i\in [r]$, where $\tilde{\pi}_i:R^\B\to B$ is defined by $\bb\mapsto b_i$ for $\bb=(b_1,\dots,b_r)\in R^\B$. This would conclude the proof of the lemma. Observe that ${M_i}_{/f}=(q_{/\pi_i})_{/f}=q_{/f\circ \pi_i}$ and $\tilde{q}_{/\tilde{\pi}_i}=(q_{/g})_{/\tilde{\pi}_i}=q_{/\tilde{\pi}_i\circ g}$. The claim follows since $f\circ\pi_i=\tilde{\pi}_i\circ g$.
\end{proof}

\end{appendices}

{\small
\bibliographystyle{plainurl}
\bibliography{cz}
}

\end{document}